%% file: main.tex
\documentclass[prx,
aps,
twocolumn,
superscriptaddress,
reprint,
floatfix
]{revtex4-2}

\pdfoutput=1
\usepackage{dsfont,amsmath,amssymb,amsthm, bbm,graphicx,mathtools,bm,ulem,physics}
\usepackage[dvipsnames]{xcolor}
\usepackage[toc,page]{appendix}
\usepackage{enumitem,hyperref}
\usepackage{thmtools,thm-restate}
\usepackage[english]{babel}
\usepackage{enumitem} 


\hypersetup{
	colorlinks=true,  
	linkcolor=blue,   
	citecolor=blue,   
	urlcolor=blue     
}

\renewcommand{\eqref}[1]{Eq.~(\ref{#1})}
\newcommand{\figref}[1]{Fig.~\ref{#1}}

\newcommand{\secref}[1]{Section~\ref{#1}}
\newcommand{\appref}[1]{Appendix~\ref{#1}}

\newcommand{\lemref}[1]{Lemma~\ref{#1}}
\newcommand{\thmref}[1]{Theorem~\ref{#1}}

\newtheorem{innercustomlem}{Lemma}
\newenvironment{customlem}[1]
  {\renewcommand\theinnercustomlem{#1}\innercustomlem}
  {\endinnercustomlem}

\def\>{\rangle}
\def\<{\langle}
\def\K{ {\cal K} }
\def\E{ {\cal E} }
\def\P{ {\cal P} }
\def\H{ {\cal H} }

\def\S{ {\cal S} }

\def\U {{\cal U}}

\def\R {{\cal R}}
\def\G {{\cal G}}

\def\F{ {\cal F} }
\def\A{ {\cal A} }
\def\B{ {\cal B} }
\def\O{ {\cal O} }
\def\J{ {\cal J} }
\def\H{ {\cal H} }
\def\W{ {\cal W} }

\def\Z{ {\cal Z} }

\def\P{ {\cal P} }
\def\C { {\cal{C}}}
\def\D{ {\cal D} }

\def\id{ \mathbbm{1} }
\def\tr{ \mbox{tr} }

\def\I{ \mathcal{I} }

\def\p {\mathbf{p}}

\def\r {\mathbf{r}}
\def\w {\mathbf{w}}
\def\x {\mathbf{x}}

\def\z {\mathbf{z}}

\def\bmk{ {\bm{k}} }
\def\bmp{ {\bm{p}} }

\def\bmx{ {\bm{x}} }
\def\bmy{ {\bm{y}} }
\def\bmz{ {\bm{z}} }
\def\bmp{ {\bm{p}} }
\def\bmr{ {\bm{r}} }
\def\bms{ {\bm{s}} }
\def\bma{ {\bm{a}} }
\def\bmu{ {\bm{u}} }
\def\bmv{ {\bm{v}} }
\def\bmw{ {\bm{w}} }
\def\bme{ {\bm{e}} }

\DeclareMathOperator{\supp}{supp}

\DeclareMathOperator{\CNOT}{CNOT}
\DeclareMathOperator{\SWAP}{MV}

\def\tr{ \mathrm{tr} }

\def\>{\rangle}
\def\<{\langle}

\renewcommand{\emph}{\textit}

\newtheorem{theorem}{Theorem}

\newtheorem{lemma}[theorem]{Lemma}

\newtheorem{corollary}{Corollary}
\newtheorem{result}{Result}

\begin{document}	

\title{General entropic constraints on CSS codes within magic distillation protocols}
\author{Rhea Alexander}
    \affiliation{Department of Physics, Imperial College London, London SW7 2AZ, UK}
	\affiliation{School of Physics and Astronomy, University of Leeds, Leeds, LS2 9JT, UK}
\author{Si Gvirtz-Chen}
	\affiliation{School of Physics and Astronomy, University of Leeds, Leeds, LS2 9JT, UK}
\author{Nikolaos Koukoulekidis}
    \affiliation{Department of Physics, Imperial College London, London SW7 2AZ, UK}

\author{David Jennings}
    \affiliation{Department of Physics, Imperial College London, London SW7 2AZ, UK}
    \affiliation{School of Physics and Astronomy, University of Leeds, Leeds, LS2 9JT, UK}

\begin{abstract}

Magic states are fundamental building blocks on the road to fault-tolerant quantum computing. CSS codes play a crucial role in the construction of magic distillation protocols. Previous work has cast quantum computing with magic states for odd dimension $d$ within a phase space setting in which universal quantum computing is described by the statistical mechanics of quasiprobability distributions. Here we extend this framework to the important $d=2$ qubit case and show that we can exploit common structures in CSS circuits to obtain distillation bounds capable of out-performing previous monotone bounds in regimes of practical interest. Moreover, in the case of CSS code projections, we arrive at a novel cut-off result on the code length $n$ of the CSS code in terms of parameters characterising a desired distillation, which implies that for fixed target error rate and acceptance probability, one needs only consider CSS codes below a threshold number of qubits. These entropic constraints are not due simply to the data-processing inequality but rely explicitly on the stochastic representation of such protocols.
\end{abstract}
	
\maketitle

\section{Introduction}

Work towards achieving fault-tolerant quantum computing is currently seeing rapid progress on qubit systems on many computational platforms \cite{campbell_roads_2017,Raussendorf2001,Raussendorf2013,Nickerson2014,Nikahd2017,chao2018,lin_pieceable_2020,Bourassa2021,chamberland2020}. In particular, the surface code~\cite{Bravyi_1998,Dennis_2002,Freedman_2001} is a leading framework that allows Clifford operations to be implemented transversally on blocks of physical qubits. However, Clifford operations are not universal for quantum computing~\cite{gottesman1997stabilizer,gottesman1998heisenberg}, and in fact it is impossible to encode any universal gateset transversally~\cite{Eastin2009}. A prominent method of circumventing this problem is the magic state injection model, wherein the Clifford group is promoted to a universal gateset by injecting copies of special non-stabilizer states known as magic states~\cite{negativity_resource,Campbell2012ReedMuller}. While magic states can only be prepared in surface codes with relatively high error rates, it is possible to reduce the noise per copy by converting many noisy magic states into fewer higher-fidelity magic states using only stabilizer operations~\cite{original_magic_states}. This process, known as \emph{magic state distillation}, allows magic states to be produced with arbitrarily high purity, and thereby enables universal quantum computation within surface code models. 

Almost all protocols~\cite{original_magic_states,reichardt2005quantum,BravyiHaah2012overhead,Haah2017magicstate,Hastings2018sublogarithmic} to-date for qubit magic distillation are based on a subclass of stabilizer codes known as Calderbank-Shor-Steane (CSS) codes~\cite{CalderbankShor_CSS_codes_orginal,steane_CSS_codes_orginal}. CSS codes can be constructed from two classical linear codes, allowing one to draw on a plethora of results from classical coding theory to construct quantum codes with desirable properties. For instance, it has been shown that CSS codes are optimal when it comes to constructing quantum error correcting codes that support a transversal $T$-gate~\cite{Rengaswamy2020optimality}, a key feature in many of the aforementioned distillation protocols. Although significant progress has been made to reduce the overhead of such protocols~\cite{Litinski2019costly}, distillation is still estimated to dominate the total resource cost of performing a computation in the magic state injection model. Therefore, a better understanding of the extent to which this cost can be reduced is of great practical interest.

Recent work~\cite{koukoulekidis2022constraints} has developed a framework for analysing magic distillation in odd-dimensional systems by taking key insights from a rich literature of majorization theory and applying them to discrete phase space representations of magic states. In odd dimensions, Gross's Wigner function~\cite{gross2006hudson} provides a representation wherein distillable magic states correspond to quasiprobability distributions containing negativity on a discrete phase space~\cite{negativity_resource}. By contrast, stabilizer states correspond to probability distributions, and stabilizer operations in general are represented by stochastic transformations~\cite{gross2006hudson,Appleby2005}. Thus when computation is restricted to the stabilizer setting, one obtains a classical stochastic model that can be studied using entropic theory and in particular relative majorization~\cite{Veinott_1971, horodecki2013fundamental, Blackwell_1953, Ruch_1976, Lostaglio_2019}. Ref.~\cite{koukoulekidis2022constraints} extended majorization tools to negatively-represented magic states, and found that a dense subset of $\alpha$-R\'{e}nyi entropies $H_\alpha$ remain well-defined and meaningful as quantifiers of disorder on quasiprobability distributions under stochastic processing, leading to fundamental constraints on magic distillation protocols in the form of thermodynamic laws.

Since most quantum algorithms are formulated for systems of qubits, the important question of whether this framework can be extended to qubit systems remains. There are, however, many well-known obstacles in constructing valid Wigner representations for qubits (related to the fact that $2^{-1}$ does not exist modulo 2~\cite{Appleby2005}). Many constructions for Wigner functions, including Gross's, cannot be extended to qubits~\cite{gross_thesis,Wigner_second_inextendable}, while others represent some pure stabilizer states negatively~\cite{stabilizer_state_positive_rep_1, stabilizer_state_positive_rep_2}, which breaks the link established in odd dimensions between Wigner negativity and quantum computational speed-up~\cite{negativity_magic_link_1,negativity_magic_link_2}. While substantial work has been done to develop phase space representations wherein all qubit stabilizer operations are non-negatively represented~\cite{Raussendorf2020,Zurel2020}, channels are typically not mapped to linear, let alone stochastic, transformations under such representations. 

Progress can be made by identifying subsets of qubit stabilizer operations that can be represented stochastically, while nevertheless remaining capable of universal quantum computation via magic state injection. In this paper, we make use of a Wigner representation for qubits introduced in Ref.~\cite{Catani_2017} that shares many desirable features with Gross's representation, such as the linear representation of channels. Drawing on results from Ref.~\cite{delfosse2015}, we show that CSS circuits -- the subset of stabilizer circuits wherein CSS states play the role of stabilizer states -- remain stochastic in this representation. Since CSS circuits are known to be capable of universal quantum computation via magic state injection~\cite{delfosse2015,Catani2018CCSSP_universal}, they provide a setting where we can extend the statistical mechanics framework for universal quantum computing developed in Ref.~\cite{koukoulekidis2022constraints} to qubits.

The structure of our paper is as follows. In~\secref{sec:basic_theory}, we introduce the phase space representation of qubit states and channels we will use, and identify some regimes where this representation becomes stochastic. Building from this, in~\secref{sec:majorization} we develop majorization techniques to analyse stochastically-represented magic distillation protocols on qubits. Finally, in~\secref{sec:structure_CSS}, we apply these tools to derive general entropic constraints (in the form of upper and lower bounds on code length) for distillation protocols that project onto CSS codes, which exploit structures basic to this distillation strategy.

\section{Main results}

 \begin{figure}
	\centering
	\includegraphics[width=0.75\linewidth]{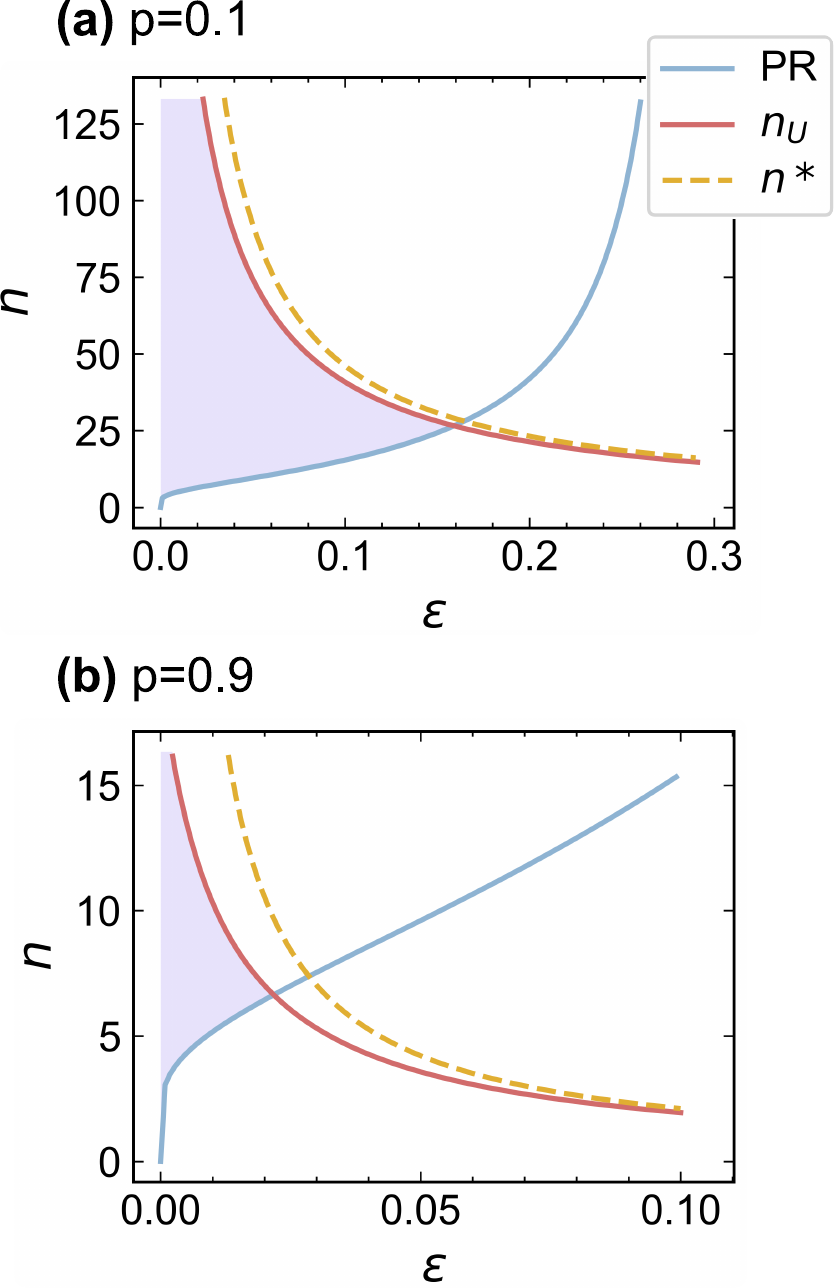}
	\caption{\textbf{(Finite range on CSS code lengths for magic state distillation protocols).} We plot upper and lower bounds on the number of copies $n$ of the noisy Hadamard state $(1-\epsilon) \ketbra{H} + \epsilon \frac{\id}{2}$ required to distil a single output qubit $\ket{H}$ with output error rate $\delta=10^{-9}$ by projecting onto an $[[n,1]]$ CSS code. The shaded purple region shows the range of code lengths allowed by the tightest numeric upper bound (red curve) from \thmref{thrm:lower_bound_n_CSS_code} and the lower bound from projective robustness (PR) introduced in Ref.~\cite{Regula2022Probabilistic} (blue curve). The analytic upper bound $n^*$ (dashed yellow curve) defined in \eqref{eq:analytic_upper_bound_local} is shown to form a good approximation to the numeric bound. \textbf{(a)} When target acceptance probability $p$ is low ($p=0.1$) the upper bounds are less constraining; \textbf{(b)} By increasing to $p=0.9$, the upper bounds become considerably tighter. In both cases, there is a cut-off input error $\epsilon$ beyond which no CSS code projection protocol can achieve the desired combination of output error and acceptance probability. 
		\label{fig:bound_comparison}}
\end{figure}
We show that CSS circuits can be represented by stochastic maps on a well-defined multiqubit phase space. By exploiting techniques from majorization, in \thmref{thrm:general_qubit_bound} we extend the statistical mechanical framework of Ref.~\cite{koukoulekidis2022constraints} to CSS qubit quantum computation.

We further find that, similar to Ref.~\cite{campbell_browne}, every CSS circuit can be decomposed in terms of protocols that project onto CSS codes, which therefore constitute the core machinery for magic distillation in CSS circuits. For such CSS code projection protocols, we obtain novel upper bounds on the code length $n$ as a function of the number of output qubits $k$, acceptance probability $p$ and input and output error rates $\epsilon$ and $\delta$ respectively (which are typically related to the code distance $D$ via $\delta = O(\epsilon^D)$). Our main result is the following, which we generalise to odd dimensional systems and arbitrary stabilizer codes in \thmref{thrm:upper_bound_n_m}.
\begin{result}
Consider the distillation of $k$ copies of a pure qubit magic state $\psi$ from a supply of the noisy magic state $\rho$, where both $\psi$ and $\rho$ have real density matrices in the computational basis. Any magic distillation protocol that projects onto the codespace of an $[[n,k]]$ CSS code and can use $n$ copies of $\rho$ to distil out a $k$-qubit state $\rho'$ at output error ${\delta \geq \norm{\rho' - \psi^{\otimes k} }_1}$ and acceptance probability $p$ must have a code length $n$ such that
\begin{align}
    n  \ge \frac{ k \left[ 1 - H_\alpha (W_{\psi})  \right]  - \frac{\alpha}{1-\alpha} \log\big( \frac{p}{1 + \delta 2^{5/2} }\big)  }{\left[ 1 -H_\alpha (W_{\rho}) \right]}
\end{align}
for all $\alpha \in \A$ for which $H_\alpha(W_\rho) < 1$, and
\begin{align}
    n  \le \frac{ k \left[ H_\alpha (W_{\psi}) - 1 \right]  +\frac{\alpha}{1-\alpha} \log \big( \frac{p}{1 + \delta 2^{5/2} }\big)  }{\left[ H_\alpha (W_{\rho}) - 1 \right]}, \label{eq:analytic_upper_bound_local}
\end{align}
for all $\alpha \in \A$ for which $H_\alpha(W_\rho) >1$, where $H_\alpha(W_\rho)$ is a R\'{e}nyi entropy measure computed on the Wigner representation of the quantum state $\rho$ defined in \eqref{eq:W_rho}.
\end{result}

Combining these bounds with prior work on projective robustness of magic~\cite{Regula2022Probabilistic}, we constrain the code length of any $[[n,k]]$ CSS code projection that could achieve a desired distillation process to lie within a finite range. An example for the distillation of a single Hadamard state are shown in~\figref{fig:bound_comparison}. In this case, our analytic upper bounds on the number of noisy magic states needed to distil out a single copy at the output take on the particularly simple form:
\begin{align} \label{eq:simple_trade_off}
     n \le  \log_{f(\epsilon)} \left[\frac{1+6\delta}{p}\right]^2 \eqqcolon n^*,
 \end{align}
where the base of the logarithm is given by ${f(\epsilon) \coloneqq [1-\epsilon +\frac{\epsilon^2}{2}]^{-1} \ge 1}$. In the spirit of Ref.~\cite{Fang2020}, \eqref{eq:simple_trade_off} and more generally \thmref{thrm:upper_bound_n_m}, can be viewed as expressing trade-off relations between various distillation parameters. These fundamental no-go results may be instructive when constructing stabilizer code projection protocols with optimized parameters. 

\section{Stochastic representation for CSS circuits on qubits}
\label{sec:basic_theory}

In this section, we review the qubit representation $W_\rho$ introduced in Ref.~\cite{Catani_2017} that forms the backbone of our work. We expand upon its properties and confirm that it respects all sequential and parallel composition of processes, the former of which crucially gives rise to a well-defined input-output relation $W_{\E(\rho)} = W_\E W_\rho$. Moreover, the latter implies that the representation of product states factorizes over subsystems, a property which will prove computationally advantageous given that inputs to magic distillation protocols typically take on the form $\rho^{\otimes n}$. Furthermore, we show that all magic distillation protocols executed by CSS circuits are stochastically represented, which means that such protocols can be analysed using majorization theory and admit a description in terms of classical statistical mechanics on quasiprobability distributions.

\subsection{Phase space representation of qubit states}
\label{subsec:our_wigner_rep}
We first establish some convenient notation. Let ${\bmu \coloneqq (u_1, \dots, u_n)\in \mathbb{Z}^n_2}$ denote a binary vector. Furthermore, given any single qubit operator $O$, let us denote
\begin{equation}
O(\bmu) \coloneqq O^{u_1} \otimes \dots \otimes O^{u_n}.
\end{equation}

With this in place, consider an $n$-qubit quantum system with total Hilbert space $\H_2^n := \H_2^{\otimes n}$. We associate to this system a phase space $\P_n := \mathbb{Z}_2^{n} \times \mathbb{Z}_2^n$, where $\P_n$ consists of all vectors $(\bmu_x, \bmu_z)$, and has a symplectic inner product $[\bmu, \bmv]$ defined as
\begin{align}
[\bmu, \bmv] \coloneqq \bmu_z \cdot \bmv_x - \bmv_z \cdot \bmu_x \equiv \bmu_z \cdot \bmv_x + \bmv_z \cdot \bmu_x,
\end{align}
where arithmetic is carried out modulo 2.

We are now in a position to define our chosen representation over $\P_n$ (we refer the reader to \appref{appx:wigner_rep} for further details and proofs of the properties presented). We first define $n$-qubit displacement operators $\{D_\bmu\}$, where $\bmu \coloneqq (\bmu_x,\bmu_z) \in \P_n$, via strings of single qubit Pauli operators $X$ and $Z$ as
\begin{align}
    D_{\bmu} \coloneqq Z(\bmu_z) X(\bmu_x),
    \label{eq:Dx}
\end{align}
which generate the Heisenberg-Weyl group ${H(2)^{\times n}}$ on $n$-qubits modulo phase factors~\cite{Zyczkowski2006Geometry}.
These displacement operators satisfy 
\begin{align}
	D_\bmu D_\bmv = (-1)^{[\bmu, \bmv]} D_\bmv D_\bmu.
\label{eq:displacement_reorder}
\end{align}

Using these displacement operators, we can construct the following representation
\begin{align}
\label{eq:W_rho}
W_\rho(\bmu) \coloneqq \frac{1}{2^n}\tr[A^\dagger_\bmu \rho]
\end{align}
for any $n$-qubit state $\rho$, where $\{ A_\bmu \}$ are the set of $2^{2n}$ \textit{phase point operators} on $n$-qubits, which are defined as 
\begin{align}
    A_\bmu = \frac{1}{2^n} \sum_{\bmv \in \P_n} (-1)^{[\bmu,\bmv]} D_\bmv.
    \label{eq:A}
\end{align}

It can be shown (see \appref{appx:A_and_W_rho_props}) that these phase point operators share the following properties with those defining Gross's Wigner representation of $n$ qudits with Hilbert space dimension $d$ on a phase space $\P \coloneqq \mathbb{Z}^n_d \times \mathbb{Z}^n_d$:
\begin{enumerate}[label=\normalfont \textbf{(A\arabic*)}]
	\item \label{A_main_text_property:qubit_phase_point_op_tensor_product}$A_{\bmu_X \oplus \bmu_Y} = A_{\bmu_X} \otimes A_{\bmu_Y}$ on a bipartite system $XY$, where $\bmu_X$ and $\bmu_Y$ are respectively points in the phase spaces of subsystems $X$ and $Y$.
	\item $\tr[A^\dagger_\bmu A_\bmv] = d^n \delta_{\bmu,\bmv}$, \label{A_main_text_property:orthogonality}
	\item $\tr[A_\bmu] = 1$, \label{A_main_text_property:unit_trace}
	\item $\sum_{\bmu \in \P} A_\bmu = d^n \id^{\otimes n}$. \label{A_main_text_property:identity}
\end{enumerate}

These properties imply that $W_\rho$ provides an informationally complete and normalized representation of general $n$-qubit states, i.e.,
\begin{equation}
\sum_{\bmu} W_\rho (\bmu)  = 1,
\end{equation}
for any quantum state $\rho$. Like Gross's Wigner function, an immediate consequence of \eqref{eq:displacement_reorder} is that the representation $W_\rho$ transforms covariantly under the displacement operators, namely,
\begin{equation}
W_{D_{\bmv} ^\dagger \rho D_{\bmv}} (\bmu) = W_\rho (\bmu + \bmv),
\end{equation}
for all $\bmu,\bmv \in \P_n$. In fact, everything in the construction of this representation has proceeded in direct analogy to Gross's, except for the lack of phase factors ensuring the Hermiticity of the displacement operators. As a result, the phase point operators in \eqref{eq:A} are no longer Hermitian, which in turn implies that $W_\rho$ is generally complex.

However, it turns out that the real and imaginary parts of $W_\rho$ are related to the quantum state in the following simple way (a proof is given in \appref{appx:rebit_rep}):
\begin{restatable}[]{lemma}{rhoWrealim}
	\label{lemma:rho_W_real_im_correspondence}
	Given any $n$-qubit quantum state $\rho$, 
	\begin{align}
	& \mathfrak{Re}[W_\rho(\bmu)]=W_{\mathfrak{Re}(\rho)}(\bmu) \\
	& \mathfrak{Im}[W_\rho(\bmu)]=W_{\mathfrak{Im}(\rho)}(\bmu) 
	\end{align}
	for all $\bmu \in \P_n$, where $\mathfrak{Re}(\rho)$ and $\mathfrak{Im}(\rho)$ are respectively the real and imaginary parts of the density matrix of $\rho$ in the computational basis.
\end{restatable}
This immediately implies 
\begin{corollary}
	\label{corollary:quasiprob_rebit}
	The representation $W_\rho$ of an $n$-qubit state $\rho$ is a quasiprobability distribution if and only if $\rho$ is an $n$-rebit state, i.e., the density matrix $\rho$ is real in the computational basis.
\end{corollary}
To simplify our analysis, we will therefore focus on rebit states for the majority of this work, although in \secref{sec:beyond_rebits} we show how we can handle arbitrary qubit states by treating the real and imaginary components separately resulting in an overcomplete quasiprobability representation. Typically, however, we will consider the case of distilling the following Hadamard state
\begin{align}
    \ket{H} \coloneqq \cos \frac{\pi}{8} \ket{0} + \sin  \frac{\pi}{8} \ket{1}, 
\end{align}
which is equivalent to the canonical magic state ${\ket{A} \coloneqq T \ket{+} = \frac{1}{\sqrt{2}} (\ket{0} + e^{i \frac{\pi}{4}} \ket{1}}$ up to a Clifford unitary~\cite{Bravyi2016CliffordEquiv}, where $T\coloneqq \text{diag}(1,e^{i \frac{\pi}{4}})$ is the $T$-gate. The Hadamard state can thus be used in a stabilizer gadgetisation circuit to implement the $T$-gate~\cite{original_magic_states}.

\subsection{Phase space representation of channels}
\label{subsec:channel_rep}
The representation of qubit states induces a corresponding representation of qubit channels. Let $\E$ be an arbitrary channel from $n$ to $m$ qubits, and 
\begin{equation}
\J(\E) = (\I \otimes \E)(\ketbra{\phi^+_n})
\end{equation}
be its associated Choi state~\cite{watrous_2018}, where $|\phi^+_n\>$ is the canonical maximally entangled state on two copies of the input system. We now define a representation~\cite{Wang_2019} of a quantum channel $\E$ as
\begin{align}
    W_{\E}(\bmv |\bmu) \coloneqq 2^{2n} W_{\J(\E)}(\bmu \oplus \bmv),
\label{eq:Wigner-Choi}
\end{align}
for all $\bmv \in \P_m$, and $\bmu \in \P_n$. Under this representation, every channel becomes a matrix mapping the representation of an input state to the representation of the output state. More precisely, if $\sigma = \E(\rho)$, then 
\begin{align}
W_\sigma(\bmv) = \sum_{\bmu \in \P_n } W_\E(\bmv | \bmu) W_{\rho}(\bmu),
\label{eq:W_E_in_action}
\end{align}
for all $\bmv \in \P_m$. 

Furthermore, the representation $W_\E$ respects the sequential and parallel composition of channels, i.e., 
\begin{align}
	W_{\E \circ \F} &= W_\E W_\F \label{eq:sequential_respect},\\
	W_{\E \otimes \F} &= W_\E \otimes W_\F. \label{eq:parallel_respect}
\end{align}
One useful implication of \eqref{eq:parallel_respect} is that when $\E$ and $\F$ respectively prepare states $\rho$ and $\sigma$, we obtain
\begin{equation}
	W_{\rho \otimes \sigma} = W_\rho \otimes W_\sigma,
\end{equation}
which informs us that our chosen representation factorizes over subsystems for product states.

The transition matrix formed by $W_\E(\bmv | \bmu)$ preserves normalization since 
\begin{align}
\sum_{\bmv \in \P_m} W_\E(\bmv | \bmu)=1, 
\label{eq:W_normalization}
\end{align} 
for any $\bmu \in \P_n$ and any quantum channel $\E$. (Proofs of Eqs.~(\ref{eq:W_E_in_action}) through (\ref{eq:W_normalization}) can be found in Appendix \ref{appx:wigner_channel_rep}). Therefore, a quantum channel $\E$ from $n$ qubits to $m$ qubits is represented by a stochastic matrix if and only if $W_{\E} (\bmv | \bmu) \ge 0$ for all $\bmu,\bmv$. 
By inspection of~\eqref{eq:Wigner-Choi} we equivalently have that the quantum channel $\E$ is stochastically represented if and only if its Choi state $\J(\E)$ on $n+m$ qubits is represented by a genuine probability distribution on the phase space $\P_{n+m}$. 

Unlike the odd-dimensional case, the channel $\E$ is not guaranteed a stochastic representation whenever $\J(\E)$ is a stabilizer state. This is an immediate consequence of the sequential and parallel composition rules of Eqs.~(\ref{eq:sequential_respect}) and (\ref{eq:parallel_respect}), which imply that our qubit representation cannot be non-negative over the full stabilizer subtheory~\cite{schmid2021stabilizer}. However, we will show in the next section that $\E$ is stochastically represented if $\J(\E)$ belongs to an important subset of stabilizer states known as CSS states.

 \begin{figure}[t]
	\centering
	\includegraphics[width=0.7\linewidth]{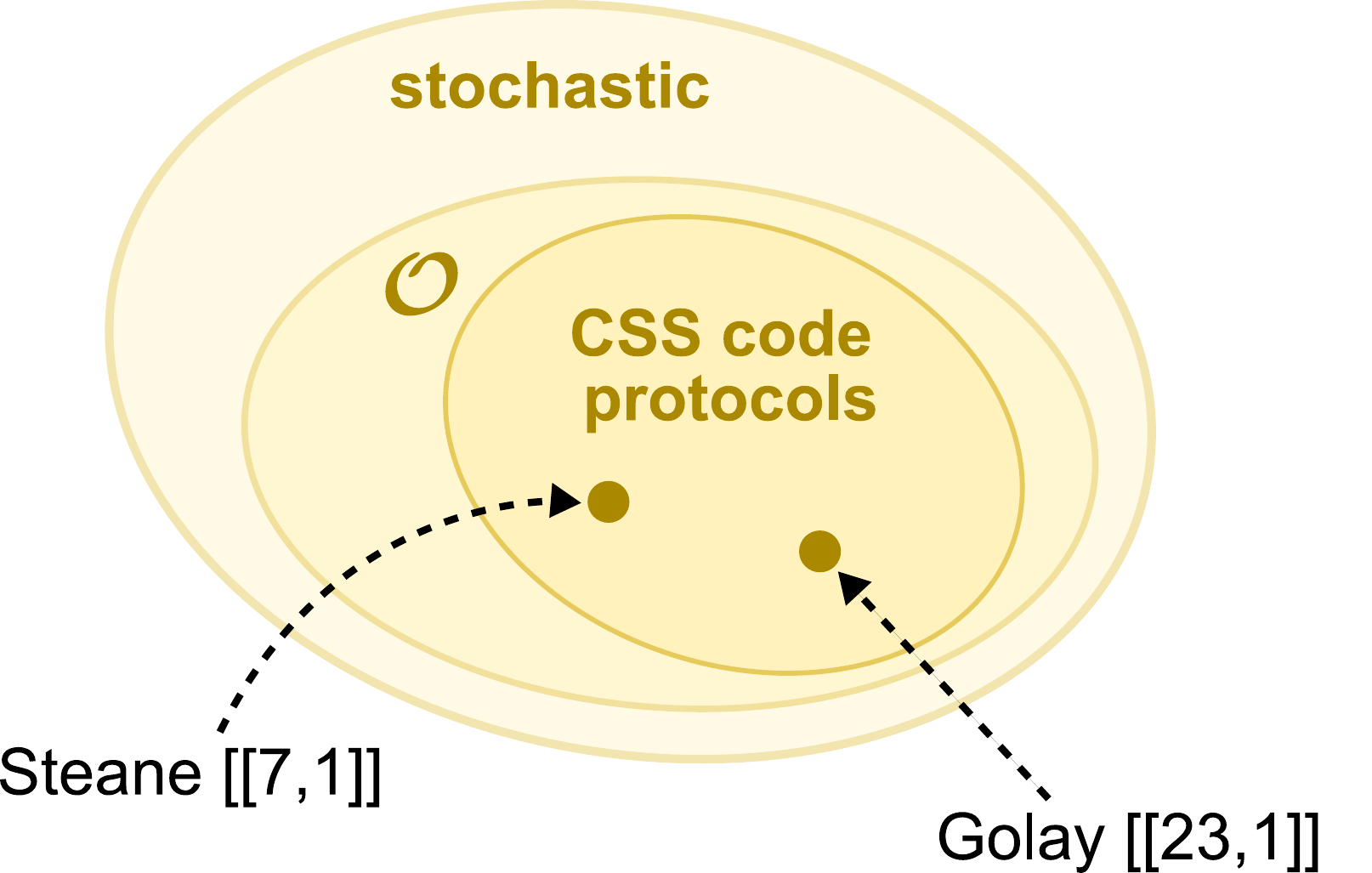} 
	\caption{ \textbf{(Schematic of our approach)}. We find that the set of completely CSS-preserving protocols $\O$ are stochastically represented. Such protocols contain the family of CSS code projections as subset, examples of which include $7$-$1$ and $23$-$1$ protocols based respectively on the Steane $[[7,1]]$ and Golay $[[23,1]]$ codes~\cite{reichardt2005quantum}.  \label{fig:stochastic_schematic} }
\end{figure} 

\subsection{CSS states and circuits}
\label{subsec:CSS_states_ops}
We now identify a class of qubit distillation protocols that arise naturally in fault-tolerant quantum computing, are sufficiently large to enable universal quantum computation, and admit a stochastic representation. In particular, we show that a channel is stochastically represented if its Choi state is CSS. Building on this result, we construct a stochastically-represented subset of stabilizer circuits wherein CSS states play the role of stabilizer states, which has been shown to be capable of universal quantum computation with magic state injection. 

A pure CSS state on $n$ qubits is any stabilizer state whose stabilizer group can be generated by $n$ Pauli observables that are individually of $X$-type or $Z$-type only. For instance, $\ket{\phi^+}\coloneqq \frac{1}{\sqrt{2}} (\ket{00} + \ket{11})$ has the stabilizer group
\begin{align}
    \S(\ket{\phi^+}) = \langle X_1 X_2, Z_1 Z_2 \rangle,
\end{align}
and is therefore CSS. By contrast, $\ket{\psi} \coloneqq \id \otimes H \ket{\phi^+}$ is stabilized by
\begin{align}
    \S(\ket{\psi}) = \langle X_1 Z_2, Z_1 X_2 \rangle,
\end{align}
and is not CSS because its stabilizer generators necessarily mix $X$ and $Z$. As they are generators of stabilizer groups defining CSS states, we group $X$- and $Z$-type Pauli observables together as CSS observables. Letting the set of all pure CSS states be denoted $\Omega_{css}$, we further define the set of all CSS states $\D_{css}$ as the convex hull of $\Omega_{css}$.

The representation we have chosen coincides on all rebit states with an earlier one introduced by Ref.~\cite{delfosse2015} (see Appendix \ref{appx:rebit_rep}), in which it was shown that a Discrete Hudson's theorem~\cite{gross2006hudson} can be recovered for qubits when one restricts to rebits. More precisely, it was shown that any pure $n$-rebit state is non-negatively represented if and only if it is CSS. Therefore, $W_\rho$ is a valid probability distribution for all $\rho \in \D_{css}$, and we conclude that
\begin{theorem}
	\label{theorem:J_CSS_stoch_rep}
	A quantum channel $\E$ from $n$ to $m$ qubits is stochastically represented if $\J(\E)$ is a CSS state on $n+m$ qubits.
\end{theorem} 

\thmref{theorem:J_CSS_stoch_rep} can be leveraged to identify stochastically-represented qubit stabilizer operations in a systematic way. A channel $\E$ from a system of qubits $B$ is CSS-preserving if $\E(\rho)$ is a CSS state for all $\rho \in \D_{css}$, and completely CSS-preserving if, given any CSS state $\rho_{AB}$ on another system of qubits $A$ as well as $B$, $\I_A \otimes \E_B(\rho_{AB})$ is always CSS. We now note that the maximally entangled state $\ket{\phi^+_n}$ over two sets of $n$ qubits is CSS for all $n$ (see Appendix \ref{appx:wigner_channel_rep}). Therefore, if $\E$ is completely CSS-preserving, $\J(\E)$ must be CSS. By \thmref{theorem:J_CSS_stoch_rep}, it follows that every completely CSS-preserving channel is stochastically represented. 

To motivate the class of completely CSS-preserving channels as operationally significant, we highlight that they cover at least the following subset of stabilizer circuits (see Appendix \ref{appx:CSS_operations} for proof):
\begin{restatable}[CSS circuits]{lemma}{CCSSPpowers}
	\label{lemma:CCSSP_powers}Any sequence of the following stabilizer operations:
	\begin{enumerate}
		\item introducing a CSS state on any number of qubits,
		\item performing a completely CSS-preserving gate on any number $n$ of qubits, i.e., a member of the group 
		\begin{equation}
			\G(n) \coloneqq \langle \CNOT(i,j), Z_i, X_i \rangle_{i,j = 1,\dots,n, i \neq j}.
		\end{equation}
		\item projectively measuring a CSS observable (with the possibility of classical control conditioned on outcome),
		\item discarding any number of qubits, 
	\end{enumerate}
as well as statistical mixtures of such sequences, is completely CSS-preserving.
\end{restatable}

By using CSS-preserving rather than completely CSS-preserving gates, channels covered by \lemref{lemma:CCSSP_powers} can be promoted to the subset of stabilizer circuits where CSS states play the role of stabilizer states. However, both groups of gates are equally powerful for magic distillation (see discussion in Appendix \ref{appx:collective_hadamard_omission}). Thus we directly refer to the set of channels covered by \lemref{lemma:CCSSP_powers} as ``CSS circuits”, and conclude that all such circuits are stochastically represented.

We emphasise the computational power of CSS circuits. Firstly, they are capable of universal quantum computation when supplemented by rebit magic states~\cite{delfosse2015,Catani2018CCSSP_universal} they can also distil~\cite{guerin_masters}. Moreover, the gateset $\G(n)$ constitutes all gates that can be implemented fault-tolerantly using defect braiding in surface codes~\cite{defect_braiding_Raussendorf2007}. Finally, we will see in \secref{sec:structure_CSS} that CSS circuits form the basis of many existing magic distillation protocols constructed around CSS codes.

\section{Entropic constraints on completely CSS-preserving protocols}
\label{sec:majorization}
The standard approach to obtaining constraints on magic distillation is tracking a magic monotone~\cite{veitch2014resource,Howard2018Application,Wang_2020,Fang2020,Fang2022Nogo2}, which is any property of a quantum system that cannot be increased under some class of magic non-generating operations (e.g. stabilizer operations). The paradigmatic example is mana~\cite{veitch2014resource}, the total negativity in the Wigner representation of a state. However, this approach operates at the state level and therefore does not incorporate any additional distinguishing physics of distillation protocols. In contrast, the recent work~\cite{koukoulekidis2022constraints} considers how a class of magic distillation protocols transform a \emph{pair} of quantum states -- one a noisy magic state, the other a stabilizer state singled out by the characteristic physics of those protocols. 

Here we briefly review the approach taken in  Ref.~\cite{koukoulekidis2022constraints} to extend relative majorization to quasiprobability distributions, and how that leads to the extension of a dense subset of $\alpha$-R\'{e}nyi divergences from classical statistical mechanics to quantify the non-classical order in magic states under distillation. We then adapt this work for rebit magic state distillation using CSS circuits. 

\subsection{Statistical mechanics of quasiprobability distributions}
\label{subsec:maj_intro}
At the heart of statistical mechanics are the notions of disorder and deviations from equilibrium. In classical statistical mechanics, this leads to the thermodynamic entropy ${H(\p) = - \sum_i p_i \log p_i}$, which is essentially the unique measure of disorder of a statistical distribution ${\p= (p_1, \dots, p_N)}$.

In odd-dimensional systems~\cite{negativity_resource} or restricted qubit models~\cite{delfosse2015,Catani2018CCSSP_universal,RaussendorfDan2017}, magic states that promote an efficiently simulable part of quantum mechanics to universal quantum computation must have negativities in their representation within a phase space model. Despite this negativity, it is still possible to arrive at a well-defined statistical mechanical description that circumvents the fact that the Boltzmann entropy is not well-defined. The key observation we exploit is that the framework of majorization remains well-defined when extended to quasiprobability distributions, and is a more fundamental concept than the traditional entropy.

Given two probability distributions $\p = (p_1, \dots, p_N)$ and ${\p' = (p_1', \dots, p_M')}$, we would like to determine which of them is ``more disordered'' than the other. This can be done by comparing $\p$ to some reference probability distribution ${\r=(r_1, \dots ,r_N)}$ of our choice, and $\p'$ to some other reference probability distribution $\r' =(r'_1, \dots ,r'_M)$, also of our choice. We then say that $(\p, \r)$ relatively majorizes $(\p', \r')$ and write $(\p,\r) \succ (\p', \r')$ if there exists a stochastic map $A$ that sends the first pair of distributions into the second, namely
\begin{align}
(A \p , A\r) &= (\p', \r').
\end{align}
It was shown in Ref.~\cite{koukoulekidis2022constraints} that this definition can be extended to the case of quasiprobability distributions in the first argument, and the following result was established to provide an entropic measure in terms of the $\alpha$-R\'{e}nyi divergences.
\begin{theorem}[\cite{koukoulekidis2022constraints}] Let $\w= (w_1, \dots , w_N)$ and ${\w' = (w'_1 ,\dots , w'_M)}$ be any two quasiprobability distributions and let $\r = (r_1 ,\dots , r_N)$ and $\r' = (r'_1, \dots, r'_M)$ be any two probability distributions with non-zero components. If $(\w, \r) \succ (\w',\r')$ then
\begin{equation}
D_\alpha( \w || \r) \ge D_\alpha (\w' || \r'),
\end{equation}
for all $ \alpha \in \A =\left\{\frac{2a}{2b-1} \, : \, a,b \in \mathbb{N}, a \ge b   \right\} \cup \left\{ \infty \right\}$.
\label{thm:koukouD_alpha_monotone}
\end{theorem}
Here $D_\alpha(\w || \r)$ is an extension of the classical $\alpha$-R\'{e}nyi divergence to the case of $\w$ being a quasiprobability distribution. This extension requires $\alpha \in \A$ in order for the expression
\begin{equation}\label{eq:D}
	D_\alpha(\w||\r) \coloneqq \frac{1}{\alpha-1} \log \sum_{i=1}^N w_i^\alpha r_i^{1-\alpha},
\end{equation}
to be well-defined\footnote{Note that since the set $\A$ is dense in $\alpha \in (1, \infty)$, one could extend the definition of $D_\alpha(\w || \r)$ to all $\alpha \ge 1$ by continuity.}. In the case of $\r$ being the uniform distribution $\r = (1/N, 1/N, \dots, 1/N)$, we have that
\begin{equation}
D_\alpha(\w || \r) = \log N - H_\alpha(\w),
\end{equation}
where $H_\alpha(\w)\coloneqq (1-\alpha)^{-1} \log \sum_i w_i^\alpha$ is the $\alpha$-R\'{e}nyi entropy evaluated on $\w$. Another result of Ref.~\cite{koukoulekidis2022constraints} is that $\w$ has negativity if and only if $H_\alpha(\w)$ is negative for $\alpha$ close to $1$ and diverges to $-\infty$ in the limit $\alpha \rightarrow 1^+$. This provides a well-defined and meaningful notion of negative entropy in a statistical mechanical setting. 

\subsection{Application to completely CSS-preserving magic distillation}
\label{subsec:maj_CSS}
Since completely CSS-preserving protocols are stochastically represented, the following family of entropic constraints on rebit magic distillation applies them all (see \appref{appx:general_bound} for proof): 
\begin{restatable}[]{theorem}{generalBound}
\label{thrm:general_qubit_bound}
Let $\rho$ be a noisy rebit magic state and $\tau$ be a CSS state in the interior of $\D_{css}$. If there exists a completely CSS-preserving protocol $\E$ such that ${\E(\rho^{\otimes n}) =\rho'}$ and $\tau' \coloneqq \E(\tau^{\otimes n})$ is also in the interior of $\D_{css}$, then 
\begin{align}
 \Delta D_\alpha  \ge 0
\end{align}
for all $\alpha\in \A$, where
\begin{align}
 \Delta D_\alpha \coloneqq    nD_\alpha(W_{\rho} || W_\tau ) -  D_\alpha(W_{\rho'} || W_{\tau'} ).
\end{align}
\end{restatable}

The reference process $\tau^{\otimes n} \mapsto \tau'$ in \thmref{thrm:general_qubit_bound} can be used in three different ways: (1) as a variational parameter, (2) to account for limitations the physical hardware carrying out magic distillation, or (3) to capture structure distinctive to a family of protocols and thereby produce entropic constraints specific to that family, which we now elaborate on in turn.

For (1), we simply treat $\tau$ as a variational parameter, which can be optimized over $\D_{css}$ to obtain the following set of monotones \footnote{We note that an analogous set of magic monotones can be defined for systems of odd-prime dimension via:
\begin{equation}
    \Gamma_\alpha(\rho) = \inf_{\tau \in STAB} D_\alpha(W_\rho || W_\tau).
\end{equation}} on completely CSS-preserving protocols
\begin{align}
    \Lambda_\alpha(\rho) \coloneqq \inf_{\tau \in \D_{css}} D_\alpha (W_\rho || W_\tau),
    \label{eq:Lambda_monotone}
\end{align}
for all $\alpha \in \A$. To see this, we note that we have $D_\alpha(W_\rho||W_\tau)\ge 0$ for all $\rho$, $\tau$, with equality if and only if $\rho=\tau$ (see \lemref{lemma:nick_restatement}). Given any rebit state $\rho$, let $\tau_\rho$ be a solution to the optimization problem in \eqref{eq:Lambda_monotone}. Then if there exists a completely CSS-preserving protocol $\E$ such that ${\E(\rho) = \rho'}$, we obtain
\begin{align}
    \Lambda_\alpha(\rho) = D_\alpha (W_\rho || W_{\tau_\rho}) &\ge D_\alpha(W_{\rho'} || W_{\E(\tau_\rho)}) \notag\\ 
    &\ge \Lambda_\alpha (\rho') ,
\end{align}
where the first inequality follows from generalised relative majorization and the second inequality follows by the definition in \eqref{eq:Lambda_monotone}.
Therefore $\{ \Lambda_\alpha\}_{\alpha \in \A}$ form an infinite set of monotones on all completely CSS-preserving protocols. It is straightforward to verify that $\Lambda_\alpha$ are sub-additive, i.e., $\Lambda_\alpha(\rho^{\otimes n}) \ge n \Lambda_\alpha (\rho)$ (this follows from the additivity of the generalized $\alpha$-R\'{e}nyi divergences). Therefore, these $\Lambda_\alpha$-monotones allow us to set global bounds on any completely CSS-preserving protocol. More precisely, if there exists a completely CSS-preserving protocol $\E$ such that $\E(\rho^{\otimes n} ) =  \rho'$, then the overhead $n$ is lower bounded as 
\begin{align}
    n \ge \frac{\Lambda_\alpha(\rho')}{\Lambda_\alpha (\rho)}.
\end{align}

For (2), we can use the reference process to take into account limitations in the hardware carrying out magic distillation. For instance, Ref.~\cite{koukoulekidis2022constraints} uses the reference process to preserve the Gibbs state in order to encode a background temperature or free energy production in the distillation hardware. 

The final way (3) of using the reference process is demonstrated in the next section. We show explicitly how to the reference process may be chosen to produce entropic constraints specialised for CSS code projection protocols.

\section{Entropic constraints on CSS code projection protocols} \label{sec:structure_CSS}
In this section, we apply \thmref{thrm:general_qubit_bound} to CSS code projection protocols, and obtain lower and upper bounds on their code length ($\sim$ resource cost). In some parameter regimes, the new lower bounds outperform those due to magic monotones such as generalised robustness~\cite{Seddon2021Quantifying} and projective robustness~\cite{Regula2022Probabilistic}. 
To our knowledge, these constitute the first set of trade-off relations on distillation parameters that act as fundamental upper bounds on the resource cost for a family of distillation protocols.
 
\subsection{CSS code projections}
An elementary protocol for magic distillation, proposed in the seminal work of Bravyi and Kitaev~\cite{original_magic_states}, uses projection onto a quantum error correcting code. This protocol begins by taking in $n$ copies of a noisy magic state $\rho$ and post-selecting the no-error outcome from the syndrome measurement of an $[[n,k]]$ stabilizer code $\C$. Doing so has the effect of projecting onto $\C$, so the protocol proceeds to decode onto $k$ output qubits and discard the remaining syndrome qubits. In general, any such code projection protocol only succeeds probabilistically with some acceptance probability $p$. Nevertheless, if the likelihood of an undetectable error is less than the input error rate $\epsilon$, the post-selected output state will have a higher fidelity per qubit with respect to the target magic state than $\rho$. Many existing magic distillation protocols are based on CSS codes, such as the 15-to-1 protocol~\cite{original_magic_states} based on the $[[15,1]]$ punctured Reed-Muller code~\cite{knill1996threshold,steane1999quantum}, as well as straightfoward code projection protocols based on the $[[7,1]]$ Steane and $[[23,1]]$ Golay CSS codes analysed in~\cite{reichardt2005quantum}. 

It has long been known that any $n$-to-1 magic distillation protocol can be decomposed as a sum of stabilizer code projections followed by Clifford post-processing~\cite{campbell_browne}. This result implies that the optimal fidelity with respect to a target magic state, though not necessarily optimal acceptance probability, can always be achieved by a stabilizer code projection. In a similar way, we can show (\thmref{theorem:CSS_campbell_browne}) that any CSS circuit carrying out an $n$-to-$k$ magic distillation protocol is a sum of CSS code projections followed by completely CSS-preserving post-processing. (In fact, the proof line we give also allows one to generalise the result of Ref.~\cite{campbell_browne} to arbitrary $n$-$k$ stabilizer protocols).

An $n$-to-$k$ CSS code projection protocol is an operation $\K$ from $n$ to $k$ qubits that acts as 
\begin{equation}
	\K(\cdot) \coloneqq \tr_{k+1,\dots,n}[\U \circ \P(\cdot)], \label{eq:K_operation}
\end{equation}
where $\U$ and $\P$ are respectively a unitary decoding channel and codespace projection for an $[[n,k]]$ CSS code. Given $n$ copies of a noisy magic state $\rho$, $\K$ acts as 
\begin{equation}
	\K(\rho^{\otimes n}) = p\rho',
\end{equation}
where $\rho'$ is the output magic state on $k$ qubits and we have defined the acceptance probability ${p\coloneqq \tr[P\rho^{\otimes n}]}$ for a single successful run of $\K$. Distillation is successful if the output $\rho'$ from a successful run has a greater fidelity per qubit with respect to a target (pure) magic state of choice than $\rho$.

Since code projection protocols are not trace-preserving, the majorization constraints do not immediately apply. However, this can be remedied by preparing a specially designated CSS state $\sigma$ on $k$ qubits whenever an $n$-to-$k$ code projection protocol fails, while continuing to distinguish between successful (labelled `0') and unsuccessful (labelled `1') runs of the protocol by recording this information in an ancillary qubit. We can therefore extend $\K$ into the following trace-preserving operation $\E$:
\begin{align}
    \E(\cdot) \coloneqq \K(\cdot) \otimes \ketbra{0} + \tr[\overline{\P} (\cdot)] \sigma \otimes \ketbra{1}, \label{eq:CSS_code_reduction_channel}
\end{align}
where ${\overline{\P} \coloneqq \overline{P}(\cdot)\overline{P} \coloneqq (\id^{\otimes n} - P)(\cdot)(\id^{\otimes n} - P)}$ performs the projection onto the orthogonal complement of $\C$, and $\sigma$ is an arbitrary CSS state. We conclude that there exists an $n$-to-$k$ CSS code projection such that $\rho^{\otimes n} \mapsto p \rho'$ if and only if there exists a trace-preserving $n$-to-$k$ CSS code projection $\E$ identified in \eqref{eq:CSS_code_reduction_channel} such that 
\begin{align}
\label{eq:rho_prime_p}
\E[\rho^{\otimes n}]&= p \rho' \otimes \ketbra{0} + (1-p) \sigma \otimes \ketbra{1} \coloneqq \rho_p. 
\end{align}

\subsection{General constraints on CSS code projections}
We now exploit \thmref{thrm:general_qubit_bound} to derive majorization conditions that apply across \emph{all} $n$-to-$k$ CSS code projection protocols. Crucially, the trace-preserving CSS code projection identified in \eqref{eq:CSS_code_reduction_channel} can be implemented as a CSS circuit (see Appendix \ref{appx:CSS_code_proj} for proof), which leads to the following Lemma:
\begin{restatable}[]{lemma}{ntoonestoch} \label{lemma:n_to_1_stoch}
Every trace-preserving CSS code projection can be executed as a sequence of completely CSS preserving operations, and is therefore stochastically represented.
\end{restatable}

A natural reference process can be chosen for all trace-preserving CSS code projections by exploiting the fact that their successful components are sub-unital. To see this, we first note that the identity operator on $n$ qubits can be decomposed as $\id^{\otimes n} = P + \overline{P}$ for the codespace projector $P$ of \emph{any} $[[n,k]]$ CSS code $\C$. The successful component $\K$ in the trace-preserving code projection of $\C$ therefore acts as 
\begin{align}
\K(\id^{\otimes n}) = \K(P+\overline{P}) = \K(P).
\end{align}
Since $P$ is the logical identity on $k$ logical qubits, i.e.,
\begin{equation}
 P=  \sum_{\bmk \in \{0,1\}^k} \ketbra{\bmk_L} \equiv \id_L,
\end{equation} 
the decoding of $P$ in \eqref{eq:K_operation} must give an output state that is proportional to the maximally mixed state on $k$ physical qubits, so we obtain $\K(\id^{\otimes n}) \propto \id^{\otimes k}$. This confirms the successful component of any trace-preserving CSS code projection to be sub-unital. Furthermore, since $P$ is a rank-$2^k$ projector, the acceptance probability associated with this protocol is $ p=\tr\left[P \frac{\id^{\otimes n}}{2^n}\right]= 2^{k-n}$. Putting everything together, we find that every $n$-to-$k$ trace-preserving CSS code projection $\E$ maps the maximally mixed state to
\begin{align}
\E \left[\left(\frac{\id}{2}\right)^{\otimes n}\right] &= \tau_{n,k} \label{eq:code_projection_reference_process}
\end{align}
where we have defined the output state
\begin{align}
\tau_{n,k} &\coloneqq 2^{k-n} \frac{\id^{\otimes k}}{2^k} \otimes \ketbra{0}+(1-2^{k-n})\sigma \otimes \ketbra{1}.
\end{align}
Since $\frac{\id}{2}$ and $\tau_{n,k}$ are both full-rank CSS states (for the appropriate choice of $\sigma\in \D_{CSS}$), \eqref{eq:code_projection_reference_process} is a valid reference process for \emph{all} trace-preserving CSS code projections. 

We therefore conclude that, if there exists an $n$-to-$k$ CSS code projection such that $\rho^{\otimes n} \mapsto p\rho'$, then
\begin{align} \label{eq:Delta_D_alpha_defn}
	\Delta D_\alpha \coloneqq nD_\alpha\left(W_\rho\bigg|\bigg|W_{\frac{\id}{2}}\right) - D_\alpha\left(W_{\rho_p}||W_{\tau_{n,k}}\right) \geq 0
\end{align}
for all $\alpha \in \A$. We define $\Delta D_\alpha$ over the restricted domain $n \in [k, \infty]$ (as the number of logical qubits cannot exceed the number of physical qubits). We highlight the satisfying fact that $\Delta D_\alpha$ is independent of the choice of CSS state $\sigma$ in \eqref{eq:CSS_code_reduction_channel}. This follows from resource-theoretic arguments as well as properties of the $\alpha$-R\'{e}nyi divergence (see Appendix \ref{appx:bounds_sigma_indep} for details).

\subsection{Bounds on the code length of CSS code projection protocols}
\begin{figure}[t]
	\centering
	\includegraphics[width=0.95\linewidth]{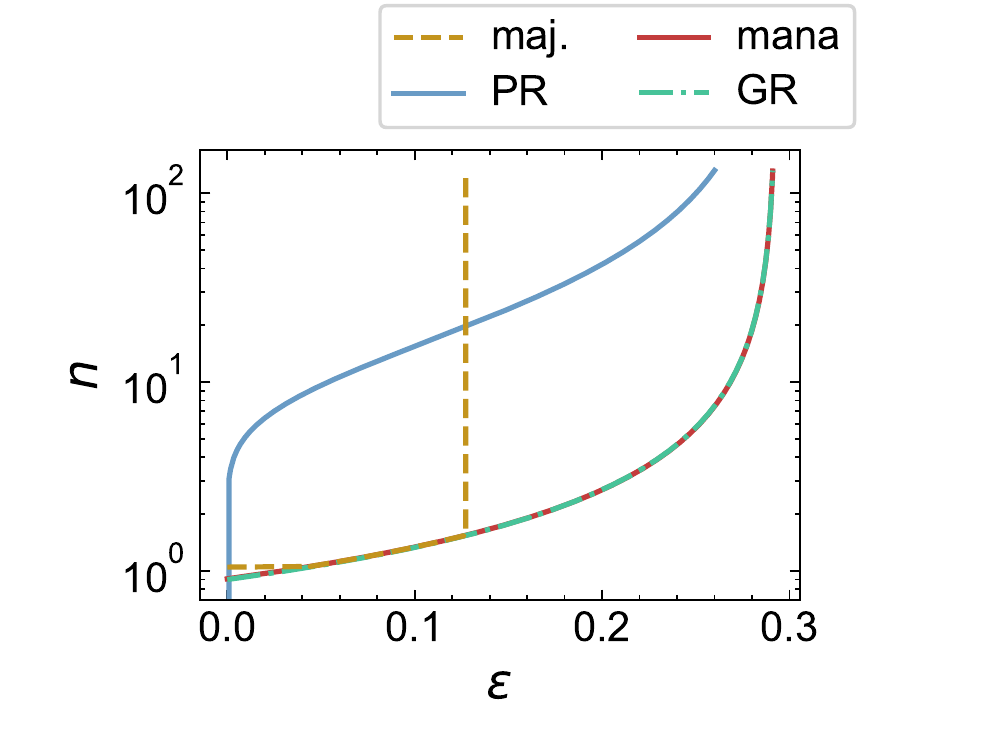}
	\caption{\textbf{(Lower bound comparison).} We plot lower bounds on the number of copies $n$ of the noisy Hadamard state $(1-\epsilon) \ketbra{H} + \epsilon \frac{\id}{2}$ required to distil a single output qubit $\ket{H}$ with output error rate $\delta=10^{-9}$ and acceptance probability $p=0.9$ under a CSS code projection protocol as a function of input error rate $\epsilon$. Our tightest lower bound from majorization (maj.) is shown to be tighter those from mana~\cite{veitch2014resource} and generalized robustness (GR)~\cite{Seddon2021Quantifying}. However, it only outperforms the lower bound from projective robustness (PR)~\cite{Regula2022Probabilistic} in the high $p$, high $\epsilon$ regime. \label{fig:lower_bounds_comparison}}
\end{figure}

The entropic constraints of \eqref{eq:Delta_D_alpha_defn} on CSS code projections can be used to bound many metrics on their performance as magic distillation protocols. In this section, we apply these constraints to bounding the code length of any CSS code projection protocol that could achieve some target combination of noise reduction and acceptance probability from a given supply of noisy magic states.  

We first highlight some properties of the relative entropy difference $\Delta D_\alpha$ in the following lemma, proofs of which can be found in~\appref{appx:fn_properties_proof}.
\begin{restatable}[]{lemma}{fnProperties}\label{lemma:properties_fn}
	The following properties of the relative entropy difference $\Delta D_\alpha$ hold for any noisy input rebit magic state $\rho$, output $k$-rebit magic state $\rho'$, acceptance probability $p<1$ and $\alpha \in \A$:
	\begin{enumerate}[label=\normalfont (\roman*)]
		\item \label{fproperty_concavity} $\Delta D_\alpha$ is concave over the domain $n\in [k,\infty]$.
		\item \label{fproperty_lim_one} $\Delta D_\alpha$ is negative in the limit where $n=k$: \begin{equation}\lim_{n\rightarrow k^+}\Delta D_\alpha<0.\end{equation}. 
		\item \label{fproperty_lim_infty} If $H_\alpha[W_\rho]>1$, then $\Delta D_\alpha$ is also negative in the asymptotic limit
		\begin{equation}\lim_{n\rightarrow  \infty}\Delta D_\alpha<0.\end{equation}
	\end{enumerate}
\end{restatable}

An immediate consequence of \lemref{lemma:properties_fn} is that $\Delta D_\alpha$ is either negative for all $n$, which implies no CSS code projection protocol can distil $\rho'$ from a supply of $\rho$ with acceptance probability $p$, or $\Delta D_\alpha$ has one or two roots located at $n^\alpha_{L}$ and $n^\alpha_U$. These roots therefore constitute lower and upper bounds on the code length $n$ of any CSS code projection that can carry out the desired distillation. We formalise these observations in the following Theorem:
\begin{theorem}
	\label{thrm:lower_bound_n_CSS_code}
	Let $\rho$ be a noisy rebit magic state. Any $n$-to-$k$ CSS code projection can distil out the $k$-rebit magic state $\rho'$ from a supply of $\rho$ at acceptance probability $p$  must have a code length $n$ such that
	\begin{align}
	n &\ge n_L^\alpha \coloneqq  \begin{cases}
		\inf_n \{ n : \Delta D_\alpha \ge 0  \} & \text{ if } \exists n: \Delta D_\alpha \geq 0,\\
		\infty & \text{ otherwise.}
	\end{cases} \label{eq:n_L_alpha}\\
	n &\le n_U^\alpha  \coloneqq \begin{cases}
		\sup_n \{ n : \Delta D_\alpha \ge 0  \} & \text{ if } \exists n: \Delta D_\alpha \geq 0,\\
		-\infty & \text{ otherwise.} 
	\end{cases}
	 \label{eq:n_U_alpha}
	\end{align}
	for all $\alpha \in \A$. Moreover, given any $\alpha$ such that ${H_\alpha[W_\rho]> 1}$, the second expression yields a \textit{finite} upper bound on $n$. 
\end{theorem}
For sufficiently low $k$, these bounds can be computed numerically using basic root-finding methods. However, we also find analytic upper and lower bounds on $n$ in \secref{sec:extensions_n_m}. The values of $n^\alpha_L$ and $n^\alpha_U$ when $D_\alpha < 0$ for all $n$ were chosen to indicate that no $n$-to-$k$ CSS code projection can carry out the desired distillation. 

We emphasise that $n$ in \thmref{thrm:lower_bound_n_CSS_code} refers to the code length (related to the resource cost $C$ by $C= \frac{ n}{p k}$) in a single run of a distillation protocol, as opposed to the the asymptotic overhead. However, single-run $n$ still constitutes a useful metric for analysing the actual resource cost of a given stage of a protocol. Moreover, distillation costs are typically dominated by the final round of a multi-stage distillation protocol (see Ref.~\cite{campbell2017roads} and references contained therein), so we expect the above bounds to be particularly informative in this context.

In \figref{fig:lower_bounds_comparison}, we plot the tighest lower bound produced by \thmref{thrm:lower_bound_n_CSS_code} on the code length of $n$-to-1 CSS code projection protocols for Hadamard state distillation. We consider input magic states of the form ${\rho(\epsilon) \coloneqq (1-\epsilon) \ketbra{H} + \epsilon \frac{\id}{2}}$, which can be generally assumed irrespective of the particular error model. This is because any state $\rho$ can be converted into this canonical form by applying the pre-processing channel 
\begin{equation}
\E(\rho) \coloneqq \frac{1}{2} (\id \rho \id + H \rho H) ,
\end{equation}
i.e. twirling with respect to the Clifford subgroup generated by the single Hadamard gate. In all parameter regimes, our lower bound is observed to be tighter than mana~\footnote{Under the restriction to the rebit sub-theory of quantum computing subject to CSS circuits, the link between negativity and magic has been restored~\cite{delfosse2015}. It therefore follows that mana is a monotone under the class of distillation protocols considered here. In fact, the monotonicity of mana is equivalent to $\Delta D_\alpha \ge 0$ in the limit $\alpha \rightarrow 1$.}, and the generalized robustness bound in Theorem 13 of Ref.~\cite{Seddon2021Quantifying}. Furthermore, in the high $p$, high $\epsilon$ regime our lower bound gives tighter constraints than the projective robustness bound~\cite{Regula2022Probabilistic}. In particular, there is a cut-off input error rate $\epsilon \approx 0.12$ at which our lower bound shoots up to infinity because, for any input error greater than this cut-off, one can always find some $\alpha$ such that $\Delta D_\alpha < 0$ for all $n$, so no CSS code projection can carry out the desired distillation given a higher input error rate (see \secref{subsec:why_upper_bounds} for physical intuition on the origin of this behaviour). In the low $p$ regime, our \textit{upper bounds} are still able to give additional constraints on code length beyond those given by the projective robustness bound. In particular, \figref{fig:bound_comparison} puts together information from our upper bounds and the lower bound from projective robustness to show that no CSS code projection can achieve some target combinations of output error and acceptance probability beyond a cut-off input error rate.  

\subsubsection{Extension to non-qubit code projection}
\label{sec:extensions_n_m}

\begin{figure}[t]
	\centering
	\includegraphics[width=0.9\linewidth]{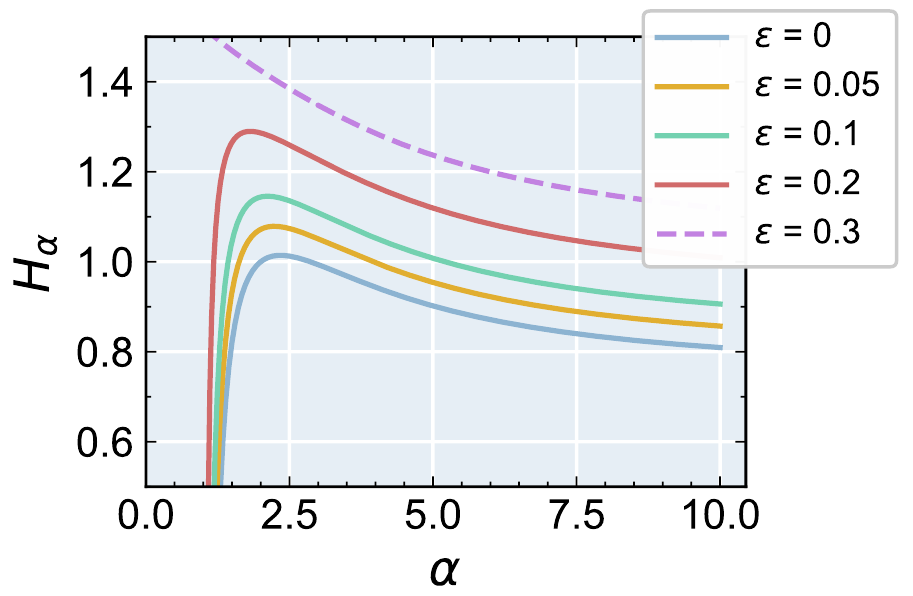} 
	\caption{ \textbf{(Wigner-R\'{e}nyi entropies \& magic distillation)} We plot the condition in \thmref{thrm:lower_bound_n_CSS_code} for the existence of finite upper bounds on $n$ in an $n$-to-1 CSS distillation for the qubit Hadamard state distillation, signified by region where $H_\alpha[W_{\rho(\epsilon)}]>1$. Even in the limit of zero input error $\epsilon=0$ we obtain a valid set of permissible $\alpha$, which implies that Hadamard state distillation under $n$-to-$1$ CSS code projection is ruled out in the asymptotic limit $n\rightarrow \infty$. We further highlight that the error rate $\epsilon =0.3$ (dashed curve) is outside of the region where $\rho(\epsilon)$ is magic ($0\le \epsilon < 1 - \frac{1}{\sqrt{2}}$), and therefore $W_{\rho(\epsilon)}$ is a proper probability distribution at $\epsilon = 0.3$, which is why $H_\alpha$ is only seen to satisfy standard monotonicity properties at this input error. We also highlight that the $\alpha \rightarrow 1$ divergence corresponds to a pole in $H_\alpha[W_\rho]$ for magic state $\rho$, and its residue is the mana of the state.\label{fig:entropy_conditions} }
\end{figure} 

Mathematically, the appearance of upper bounds on the code length of qubit CSS code projection protocols comes from the concavity of the objective function $\Delta D_\alpha $ in $n$. We now demonstrate that this feature is not peculiar to qubits, and in fact arises whenever stabilizer code projections on any quantum system are stochastic under a Wigner representation sufficiently similar to Gross's.

More precisely, we say that a Wigner representation $W$ of $n$ qudits with Hilbert space dimension $d$ is a \emph{generalised Gross's Wigner representation} if it represents each state $\rho$ as the function $W_\rho$ on a phase space $\mathbb{Z}^n_d \times \mathbb{Z}^n_d$,
\begin{align}
	W_\rho(\bmu) \coloneqq \frac{1}{d^n} \tr\left[A^\dagger_\bmu \rho\right],
\end{align} 
via phase point operators $\{A_\bmu\}$ satisfying \ref{A_main_text_property:qubit_phase_point_op_tensor_product}-\ref{A_main_text_property:identity}. For all stabilizer code projections on odd-dimensional systems, $W$ can simply be Gross's Wigner representation~\cite{koukoulekidis2022constraints}. For CSS code projections on qubits, $W$ can be the representation from \eqref{eq:W_rho}. We then have the following analytic upper and lower bounds on the resource cost of code projection protocols:

\begin{restatable}[Qudit code bounds]{theorem}{UpperBoundnk}
\label{thrm:upper_bound_n_m}
Consider the distillation of $k$ copies of a pure magic state $\psi$ from a supply of the noisy magic state $\rho$, where $\psi$ and $\rho$ are $d$-dimensional qudit states that are real-represented under a generalised Gross's Wigner representation $W$. Any stochastically-represented distillation protocol that projects onto the codespace of an $[[n,k]]$ stabilizer code and can use $n$ copies of $\rho$ to distil out a $k$-qudit state $\rho'$ with acceptance probability $p$ and output error $\delta \geq \norm{\rho' - \psi^{\otimes k} }_1$ must have a code length $n$ such that
\begin{align}
    n  \ge \frac{ k \left[ \log d - H_\alpha (W_{\psi})  \right]  - \frac{\alpha}{1-\alpha} \log\big( \frac{p}{1 + \delta d^{5/2} }\big)  }{\left[ \log d -H_\alpha (W_{\rho}) \right]},
    \label{eq:analytic_lower}
\end{align}
for all $\alpha \in \A$ for which $H_\alpha(W_\rho) < \log d$, and
\begin{align}
    n  \le \frac{ k \left[ H_\alpha (W_{\psi}) - \log d \right]  +\frac{\alpha}{1-\alpha} \log \big( \frac{p}{1 + \delta d^{5/2} }\big)  }{\left[ H_\alpha (W_{\rho}) - \log d \right]},
     \label{eq:analytic_upper}
\end{align}
for all $\alpha \in \A$ for which $H_\alpha(W_\rho) >\log d$.
\end{restatable}

One might be concerned that the conditions ${H_\alpha(W_{\rho(\epsilon)})> \log d}$ given in Theorems~\ref{thrm:lower_bound_n_CSS_code} and \ref{thrm:upper_bound_n_m} for the existence of a finite upper bound on $n$ are never actually satisfied. This turns out not to be the case. For ${n\text{-to-}1}$ CSS code projection protocols for Hadamard distillation, there always exists a valid set of $\alpha$-values such that $H_\alpha[W_{\rho(\epsilon)}]>1$ for all $\epsilon$, so we always obtain a finite upper bound on $n$. This can be seen by upper bounding the R\'{e}nyi entropy of $W_{\rho(\epsilon)}$ as
\begin{align}
    H_\alpha\left[W_{\rho(\epsilon)}\right] & \ge  H_\alpha \left[W_{\rho(0)}\right] ,
\end{align}
and then examining \figref{fig:entropy_conditions}, which shows that there exists finite range of $\alpha$ such that $H_\alpha > 1$ at $\epsilon=0$. 

The trade-off relations given in \thmref{thrm:upper_bound_n_m} can be rearranged to bound other parameters such as the acceptance probability $p$. In \figref{fig:real_protocols}, we compare the tightest upper bound on $p$ produced by \thmref{thrm:upper_bound_n_m} for $n$-to-1 Hadamard distillation to those attained in existing protocols based on CSS codes given in Refs.~\cite{original_magic_states} and \cite{reichardt2005quantum}. We see that the acceptance probabilities of basic code projection protocols using the $[[7,1]]$ Steane and $[[23,1]]$ Golay codes in \figref{fig:real_protocols}(a) are orders of magnitude less than our upper bounds, suggesting that substantial room for improvement is not ruled out. Interestingly, in \figref{fig:real_protocols}(b) our upper bound is very close to the actual acceptance probability of the protocol based on the $[[15,1]]$ Reed-Muller code in Ref.~\cite{original_magic_states}, which we speculate may hint at something fundamental about the role of the intermediate Clifford corrections used in that protocol.

\begin{figure}[t]
	\centering
	\includegraphics[width=0.9\linewidth]{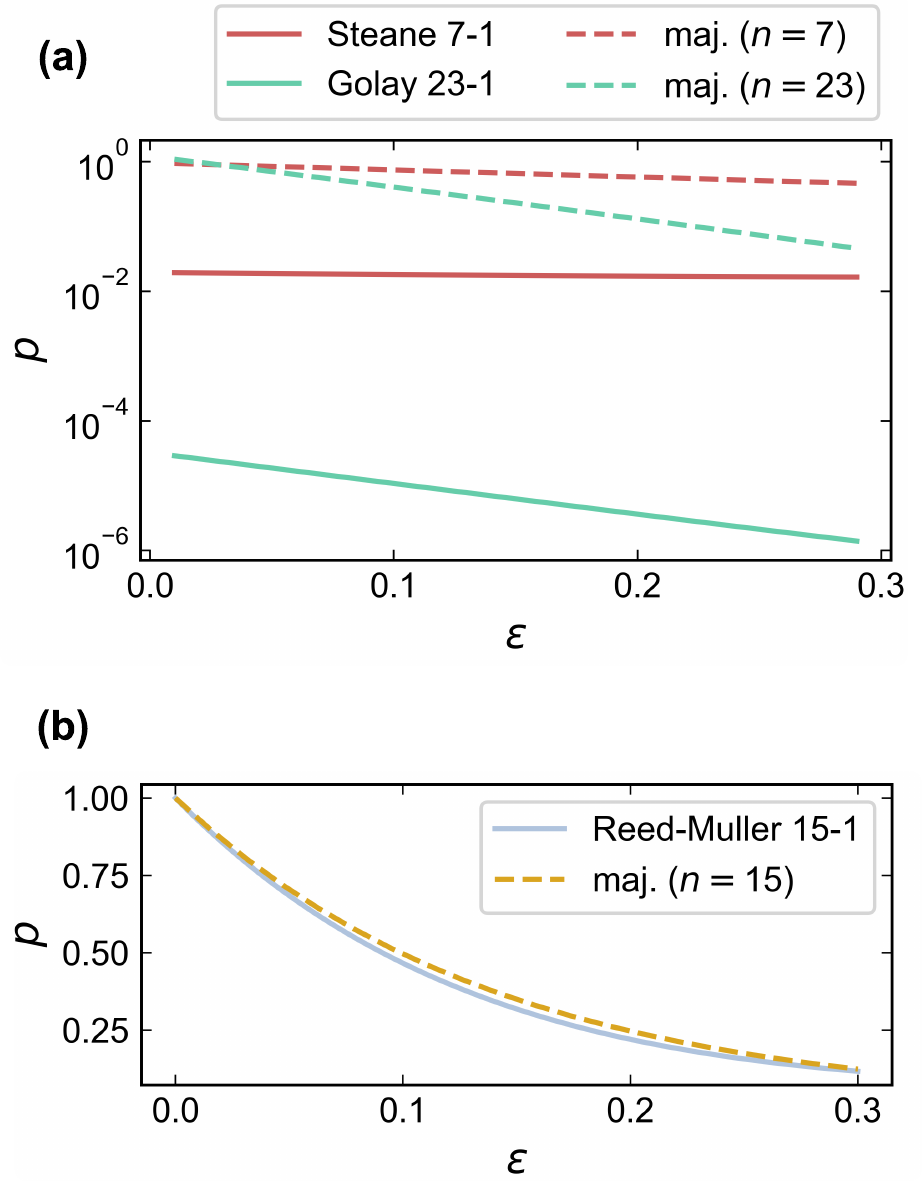}
	\caption{\textbf{(Explicit protocol comparison).} (a) We compare the majorization upper bounds (dashed lines) on the acceptance probability $p$ with which one can distil a noisy Hadamard state $(1-\epsilon) \ketbra{H} + \epsilon \frac{\id}{2}$ via an $n$-to-1 code projection against actual acceptance probabilities attained using the Steane code (purple) at $n=7$ and the Golay code (green) at $n=23$ (detailed in Ref.\cite{reichardt2005quantum}). Attained acceptance probabilities are orders of magnitude less than our upper bounds. (b) We plot the majorization upper bound (dashed line) on the acceptance probability $p$ of any 15-to-1 CSS code projection protocol with which one can distil the noisy magic state $(1-\epsilon) \ketbra{A} + \epsilon \frac{\id}{2}$. Interestingly, our bound is very close to the actual acceptance probability for the 15-to-1 protocol (blue line) given in~\cite{original_magic_states}, though we emphasise this latter protocol is \textit{not} a straightforward CSS code projection.}
		\label{fig:real_protocols}
\end{figure} 

\thmref{thrm:upper_bound_n_m} takes on a particularly simple form in the particular case of $n$-to-$1$ CSS code projection protocols for Hadamard state distillation. By evaluating the $\alpha=2$ condition explicitly, we find that, if there exists a CSS code projection that sends $n$ copies of $\rho(\epsilon)$ to a $k$-qubit state $\rho'$ with acceptance probability $p$ and output error $\delta \geq \norm{\rho' - \ketbra{H}^{\otimes k}}_1$, then $n$ is upper bounded as 
\begin{align}
     n\le n^* &\coloneqq 2 \log_{f(\epsilon)} \left[\frac{1+6\delta}{p}\right],
    \label{eq:n_star}
 \end{align}
 where the logarithm base is ${f(\epsilon) \coloneqq [1-\epsilon +\frac{\epsilon^2}{2}]^{-1}}$. This expression captures the fact that under a CSS code projection protocol, there is a fundamental trade-off between acceptance probability and output fidelity. For instance, \eqref{eq:n_star} shows that given a supply of noisy magic states ($\epsilon>0$), we cannot use CSS code projection to distil a perfect magic state ($\delta=0)$ with certainty ($p=1$), which was first shown in Ref.~\cite{Fang2020}. To further investigate this trade-off, in \figref{fig:dpi_vs_maj}(b) we plot the \textit{maximum achievable fidelity} with respect to the Hadamard state that can be achieved by an $n$-to-$1$ CSS code projection $\K$ via
\begin{equation}
    F_{\max}(\rho) = \max_{\K} \left\{\bra{H} \rho' \ket{H} \, : \, \K(\rho^{\otimes n}) \mapsto p \rho'\right\},
\end{equation}
where the maximization is performed over the set of all $n$-to-1 CSS code projection protocols.

\subsubsection{Why do we expect upper bounds?}
\label{subsec:why_upper_bounds}

Taking CSS circuits as our free operations, the appearance of upper bounds on $n$ might first seem to contradict a resource theory perspective, where we might expect $n+1$ copies of a noisy magic state to be at least as good as $n$ copies at distilling magic, since discarding subsystems is itself a CSS circuit. However, by specialising to stabilizer code projection protocols, we necessarily introduce a trade-off between $n$ and acceptance probability $p$, which we now show in a simple calculation. For any $n$-to-$k$ stabilizer code projection protocol for $d$-dimensional qudit distillation, the acceptance probability $p$ is given by how much of $n$ copies of the noisy input magic state $\rho$ projects onto the $d^k$-dimensional codespace spanned by the logical basis $\{\ket{j_L}\}_{j=0,\dots,d^k-1}$ of the code. Letting $\lambda_{\max}(\cdot)$ denote a state's largest eigenvalue, we immediately identify the following upper bound on $p$,
\begin{align}
    p = &\sum_{j=0}^{d^k-1} \bra{j_L}\rho^{\otimes n} \ket{j_L} \notag\\
    &\le d^k \lambda_{\max}(\rho^{\otimes n}) = d^k [\lambda_{\max}(\rho)]^n,
\end{align}
which falls monotonically towards $0$ as $n\rightarrow \infty$. At an intuitive level, the trade-off between $n$ and $p$ occurs because the codespaces of $[[n,k]]$ stabilizer codes remain the same size as we increase $n$, and so take up a vanishingly small part in the support of all the noisy input magic states used. Under the requirement that we have some threshold acceptance probability (below which the expected overhead would be too large), a corresponding upper bound on $n$ is then expected.  

\subsubsection{Comparison with the data-processing inequality (DPI)}
\label{subsec:DPI}

We have seen that stochastically represented CSS code projection protocols give rise to a set of upper bounds on $n$ (\thmref{thrm:upper_bound_n_m}). By comparing to the data-processing inequality (DPI), we see that although the existence of upper bounds is a general feature of code projection protocols, exploiting the stochasticity in the representation of CSS code projections gives strictly stronger bounds.

According to the DPI, if there exists a code projection (CSS or otherwise) that can distil out the $k$-qubit magic state $\rho'$ from $n$ copies of the noisy magic state $\rho$ with acceptance probability $p$, then 
\begin{equation}
	\Delta \tilde{D}_\alpha\coloneqq n\tilde{D}_\alpha\left(\rho\bigg|\bigg|\frac{\id}{2}\right) - \tilde{D}_\alpha\left(\rho_p||\tau_{n,k}\right) \ge 0,
\label{eq:DPI_constraint}
\end{equation}
for all $\alpha\in (1,\infty)$~\cite{beigi2013sandwiched}, where $\tilde{D}_\alpha (\rho ||\tau)$ is the sandwiched $\alpha$-R\'{e}nyi divergence \cite{mosonyi2014convexity,wilde2014strong} on the normalized quantum states $\rho$ and $\tau$, which is defined as
\begin{align}
    \tilde{D}_\alpha (\rho||\tau) \coloneqq \frac{1}{\alpha-1} \log \tr\left[\left(\tau^{ \frac{1-\alpha}{2\alpha}} \rho \tau^{ \frac{1-\alpha}{2\alpha}}\right)^\alpha \right].
\end{align}
We then have the following upper bound on $n$:
	\begin{align}
	n \le n^{DPI}_U  \coloneqq \min_{\alpha \in (1,\infty)} \max_n \{ n : \Delta \tilde{D}_\alpha \ge 0  \} .
	\end{align}
which is finite whenever $\alpha$ is such that $H_\alpha(\rho) >1$ (proof essentially identical to that of \thmref{thrm:lower_bound_n_CSS_code}). 

Given that we also obtain upper bounds on code length from simple data processing of quantum states, we should ask: does majorization give genuinely new constraints on magic state distillation beyond the DPI? Since our majorization conditions are a consequence of the stochastic representation of stabilizer circuits, while the DPI arises from the fact that all quantum channels are CPTP, this question may be loosely rephrased as asking whether stochasticity imposes any additional constraints beyond those imposed by CPTP on magic state distillation by stabilizer code projection. 

\figref{fig:dpi_vs_maj} allows us to answer this question in the affirmative, since the upper bound on code length due to majorization is observed to be stronger than that due to the DPI over a wide range of parameter regimes for CSS code projection protocols. In particular, extracting $n^{maj}_U \coloneqq \min_\alpha n^\alpha_U$ as the tightest bound due to majorization from Theorem \ref{thrm:lower_bound_n_CSS_code}, we find that, in the low acceptance probability $p$ and low input error $\epsilon$ regime of \figref{fig:dpi_vs_maj}(a), the difference in upper bounds $\Delta n_U \coloneqq n_U^{DPI} - n_U^{maj}$ (the amount by which majorization ``beats" DPI) is of the order $\Delta n_U  = O(10^4)$. We thus conclude that the constraints on CSS protocols stemming from majorization go beyond those from the DPI.

\begin{figure}[t]
	\centering
	\includegraphics[width=0.8\linewidth]{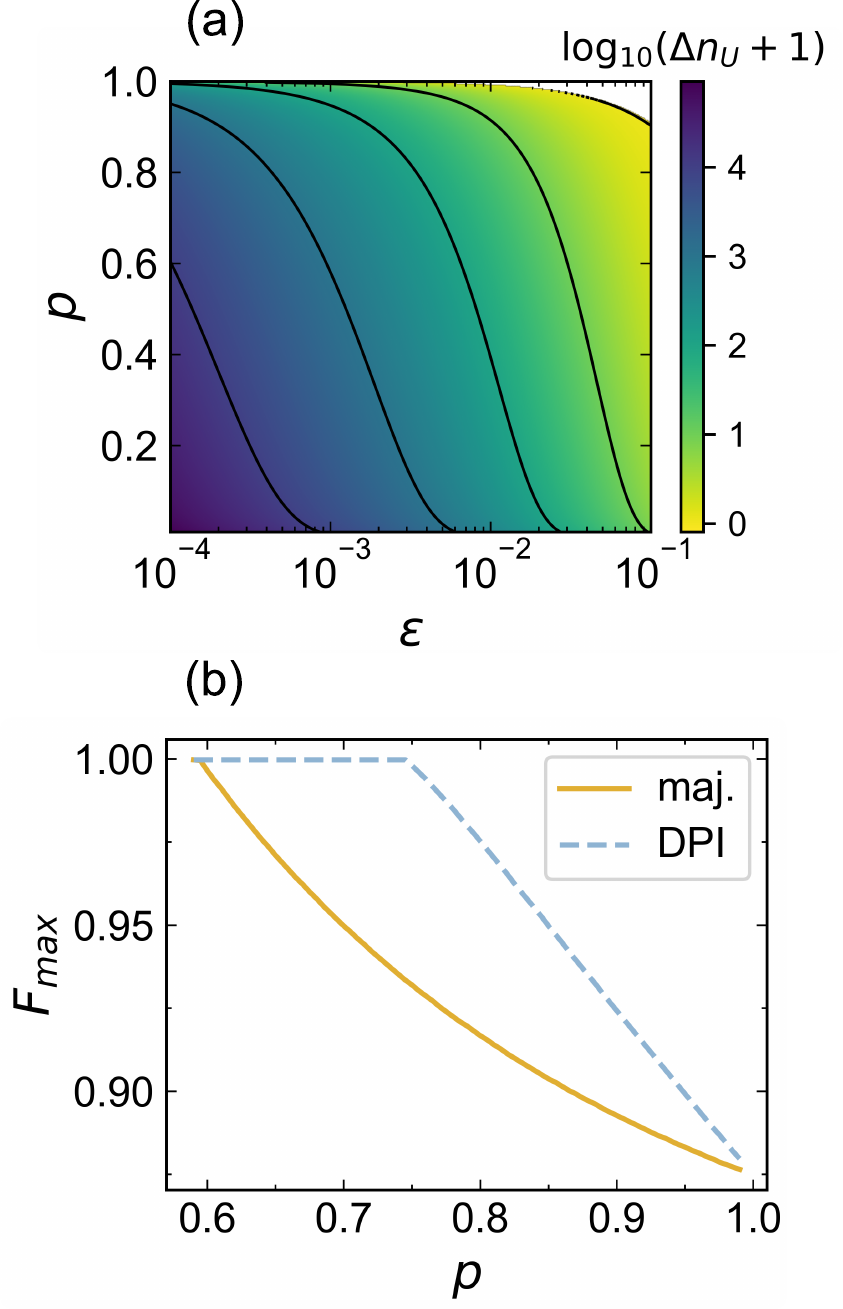}
	\caption{\textbf{(Majorization gives independent constraints over DPI).} \textbf{(a)} Shown is how (scaled) $\Delta n_U \coloneqq n^{DPI}_U - n^{maj}_U$ varies over all possible values of acceptance probability $p$ and a realistic range of input error $\epsilon$, with fixed $\delta = 10^{-9}$. Whenever we have $\log_{10}(\Delta n_U+1)>0$ means that upper bounds from majorization give tighter constraints than the DPI, reaching $\Delta n_U = O(10^4)$ in the low $p$, low $\epsilon$ regime. \textbf{(b)} We show the trade-off relation given by bounds on the maximum achievable fidelity $F_{\max}(\rho)$ vs. target acceptance probability $p$, under an $n$-$1$ CSS code projection, where $\rho=\frac{3}{4}\ketbra{H} + \frac{1}{8} \id$. For $p$ above a given threshold ($\approx 0.6$) no perfect distillation is theoretically possible, even for $n\rightarrow \infty$ copies of the input state. Majorization (maj.) is shown to give stronger constraints than that of DPI.
		\label{fig:dpi_vs_maj}}
\end{figure}

\section{Beyond rebit distillation}
\label{sec:beyond_rebits}

We have simplified our analysis thus far by restricting our attention to rebit magic states such as $\ket{H}$. However, many protocols such as the seminal 15-to-1 Bravyi-Kitaev protocol~\cite{original_magic_states} distil magic states such as $\ket{A} \propto   \ket{0} + e^{i \frac{\pi}{4}} \ket{1}$, which have complex density matrices in the computational basis. An argument can be made that since $\ket{A}$ is Clifford-equivalent to $\ket{H}$, this state should be considered equally resourceful for magic distillation. However, to address this concern more directly, we extend our majorization relations to states with complex Wigner representations. A discussion of how this can be achieved is given in \appref{appx:complex_maj}. The basic idea is that we can define valid $2d^2$-dimensional quasiprobability  distributions by forming the direct sum of the real and imaginary parts of the original distribution:
\begin{align}
 \w_\rho &\coloneqq \Re{W_\rho} \oplus \Im{W_\rho}, 
 \end{align}
while for our reference CSS states we can simply take:
 \begin{align}
 \r_\tau &\coloneqq \frac{1}{2}W_\tau \oplus W_\tau.
\end{align}
 It then follows that if there exists a completely CSS preserving operation $\E$ such that $\E(\rho) = \rho' $ and $\E(\tau) = \tau'$, then 
    \begin{align}
        D_\alpha (\w_\rho || \r_\tau) \ge  D_\alpha (\w_{\rho'} || \r_{\tau'}),
    \end{align}
 for all $\alpha \in \A$. Further study of the significance of these complex relative majorization conditions may be of foundational interest, but we leave this for future work.

\section{Discussion and Outlook}

We have shown that the statistical mechanical framework of Ref.~\cite{koukoulekidis2022constraints} can be extended to the experimentally significant case of qubit systems by focusing on the processing of magic states under CSS circuits -- i.e., the subset of stabilizer circuits where CSS states play the role of stabilizer states. To achieve this, we made use of a Wigner representation first introduced in Ref.~\cite{Catani_2017} wherein completely CSS-preserving channels correspond to stochastic transformations on phase space. This set of channels include CSS circuits, which are sufficient for universal quantum computing~\cite{delfosse2015,Catani2018CCSSP_universal} and consist precisely of the gateset that can be performed fault-tolerantly on surface code constructions~\cite{Kitaev_Surface_Code2003}.

Within this framework, we showed that relative majorization can be used to encode particular properties of an important class of distillation protocols that project onto CSS codes, in terms of which all protocols carried out by CSS circuits can be decomposed. We established general entropic constraints on such protocols in terms of upper and lower bounds on the code length $n$. 

In the context of achieving full fault-tolerance, a natural extension of this work would be to generalize our results to more sophisticated protocols. For instance, we might ask how the use of $m$ intermediate Clifford corrections in between the measurements of stabilizer generators might affect these fundamental constraints. One would expect to be able to obtain more refined bounds as a function of $m$. Moreover, while many protocols are based on CSS codes in part due to their relative ease of construction via tri-orthogonal matrices~\cite{BravyiHaah2012overhead}, from an operational perspective it would be of interest to see whether we can extend to the complete set of stabilizer operations on qubit systems.

We have also obtained a set of monotones $\{\Lambda_\alpha\}$ for completely CSS-preserving magic distillation, each of which forms a convex optimization problem. We speculate that an analogous monotone can be constructed for any resource theory for which the free operations are a subset of operations that completely preserve Wigner positivity. From the perspective of quantum optics experiments, wherein Gaussian operations and probabilistic randomness are readily available, it may be of interest  to consider the case of continuous variable systems under the set of Gaussian operations and statistical mixtures~\cite{Albarelli2018nonGaussianity}. Since the individual $\alpha$-R\'{e}nyi divergences on quasiprobability distributions were seen in~\secref{subsec:DPI} to typically produce stronger constraints than the corresponding constraints given by $\alpha$-R\'{e}nyi divergences on quantum states, it would be interesting to see how well these quasidistribution-based monotones perform relative to known state-based counterparts.

On a technical note regarding majorization theory, we point to two interesting directions for further study. 
Firstly, complex majorization constraints arise naturally when we extend our setup from rebit to all qubit states, where Wigner representations can become complex due to the non-Hermiticity of the operator basis $\{A_{\bm{u}}\}$.
We expect such constraints to take the form of a duplet of constraints applying separately to the real parts and imaginary parts of the Wigner function. In the context of non-Hermitian quantum mechanics~\cite{moiseyev2011}, results on complex majorization would also benefit theories that require an ordering between Hamiltonian eigenvalues, such as quantum thermodynamics. Secondly, the Wigner representation of Ref.~\cite{Catani_2017} recovers the covariance over symplectic affine transformations on qubit phase spaces, a property shared by Gross's Wigner function on odd-dimensional systems. This added structure on the phase space was ignored by our analysis, but could be utilised to tighten the obtained bounds in future work. In particular, as explained in the discussion of Ref.~\cite{koukoulekidis2022constraints}, the stochastic majorization used in our analysis is only a special case of $G$-majorization, where $G$ can be taken as a subgroup of the stochastic group such as the symplectic group. It can then be shown~\cite{giovagnoli_cyclic_1996,giovagnoli_1985,steerneman_1990} that one should expect to obtain a set of finite lower bound constraints on distillation, which will be tighter than stochastic majorization constraints.

\section*{Acknowledgements}
RA and NK are supported by the EPSRC Centre for Doctoral Training in Controlled Quantum Dynamics. SGC is supported
by the Bell Burnell Graduate Scholarship Fund. DJ is supported by the Royal Society and a University Academic Fellowship.

\input{appendix.tex}

\bibliographystyle{apsrev4-2}
\bibliography{paper_references}

\end{document}

%% file: appendix.tex
\appendix

\renewcommand{\appendixname}{APPENDIX}

\section{\uppercase{Wigner representation}}
\label{appx:wigner_rep}

For any $n$-qubit state $\rho$, we define the following complex-valued function $W_\rho : \P_n \mapsto \mathbb{C}$ as
\begin{align}
W_\rho(\bmu) \coloneqq \frac{1}{2^n}\tr[A^\dagger_\bmu \rho],
\end{align}
where $\{ A_\bmu \}$ are the set of $2^{2n}$ \textit{phase point operators} on $n$-qubits, which are defined as follows 
\begin{align}
	A_{\bm{0}} \coloneqq \frac{1}{2^n} \sum_{\bmu \in \P_n} D_\bmu, \quad A_{\bmu} \coloneqq D_\bmu A_{\bm{0}} D^\dagger_\bmu.
	\label{eq:A_phase_point}
\end{align}
As a consequence of \eqref{eq:displacement_reorder}, these phase-point operators can alternatively be expressed as 
\begin{align}
    A_\bmu = \frac{1}{2^n} \sum_{\bmv \in \P_n} (-1)^{[\bmu,\bmv]} D_\bmv,
\end{align}
which further reveals that every $A_\bmu$ is real in the computational basis.

Despite being complex-valued, $W_\rho$ transforms covariantly under the displacement operators - informally speaking, $\rho$ is shifted by the displacement operators around phase space - just like Gross' representation in odd dimensions. Concretely, we consider the Wigner representation of $D_\bmv \rho D_\bmv^\dagger$ for an arbitrary phase space displacement $\bmv$, which is
\begin{align}
    W_{D_\bmv \rho D^\dagger_\bmv}(\bmu) &= \frac{1}{2^n} \tr[A^\dagger_\bmu D_\bmv \rho D_\bmv^\dagger] = \frac{1}{2^n} \tr[D_\bmv^\dagger A^\dagger_\bmu D_\bmv \rho ] \notag \\ 
    &= \frac{1}{2^n} \tr\left[\sum_{\bma \in \P_n} (-1)^{[\bmu,\bma]} (D_\bma D_\bmv)^\dagger D_\bmv \rho\right]. \label{eq::local_1}
\end{align} 

Inserting the commutation relation for the displacement operators in \eqref{eq:displacement_reorder} into \eqref{eq::local_1}, we obtain
\begin{align}
	W_{D_\bmv \rho D^\dagger_\bmv}(\bmu)& = \frac{1}{2^n} \tr\left[\left(\sum_{\bma \in \P_n} (-1)^{[\bmu + \bmv,\bma]} D^\dagger_\bma\right) \rho\right]\notag \\ &= W_\rho(\bmu + \bmv),
\end{align}
which confirms that $W_\rho$ transforms covariantly under the action of the displacement operators.

\subsection{Properties of qubit phase point operators and Wigner function}
\label{appx:A_and_W_rho_props}

We first establish that the phase point operators of a joint system are simply tensor products of phase point operators on its subsystems:
\begin{enumerate}[label=\normalfont \textbf{(A\arabic*)}]	
	\item \textit{(Factorization).} \label{A_property:qubit_phase_point_op_tensor_product} On a bipartite system with subsystems $X$ and $Y$, $A_{\bmu_X \oplus \bmu_Y} = A_{\bmu_X} \otimes A_{\bmu_Y}$. 
\end{enumerate}
\begin{proof}
	From the definition of $D_\bmu$ in \eqref{eq:Dx}, it is clear that 
	\begin{align}
	D_\bmu = D_{\bmu_X} \otimes D_{\bmu_Y}.
	\end{align}
	Let $n_X$ and $n_Y$ be the numbers of qubits in subsystems $X$ and $Y$ respectively. Then the zero phase point operator on the bipartite system, $A_{\bm{0}}$, is
	\begin{align}
	A_{\bm{0}} &\coloneqq \frac{1}{2^{n_X+n_Y}} \sum_{\bmu \in \P_{XY}} D_\bmu \notag \\ 
	&= \frac{1}{2^{n_X+n_Y}} \sum_{\bmu_X \in \P_X, \bmu_Y \in \P_Y} D_{\bmu_X \oplus \bmu_Y} \notag \\ 
	&= \frac{1}{2^{n_X+n_Y}} \sum_{\bmu_X \in \P_X} \sum_{\bmu_Y \in \P_Y} D_{\bmu_X} \otimes D_{\bmu_Y} \notag \\ 
	&= A_{\bm{0}_X} \otimes A_{\bm{0}_Y}, 
	\end{align}
	which in turn implies that any phase point operator $A_\bmu \coloneqq A_{\bmu_X \oplus \bmu_Y}$ for some $\bmu_X \in \P_X$ and $\bmu_Y \in \P_Y$ is
	\begin{align}
	A_{\bmu} &\coloneqq  A_{\bmu_X \oplus \bmu_Y} =  D_{\bmu_X \oplus \bmu_Y} A_{\bm{0}} D^\dagger_{\bmu_X \oplus \bmu_Y} \notag \\ &= \left(D_{\bmu_X} A_{\bm{0}_X} D_{\bmu_X}^\dagger\right) \otimes \left(D_{\bmu_Y} A_{\bm{0}_Y} D_{\bmu_Y}^\dagger\right) \notag \\ &= A_{\bmu_X} \otimes A_{\bmu_Y} ,
	\end{align} as claimed.
\end{proof}

Property \ref{A_property:qubit_phase_point_op_tensor_product} enables us to break down any $n$-qubit phase-point operator $A_\bmu$ as a tensor product of single-qubit phase-point operators,
\begin{align}
A_\bmu = \bigotimes_{i=1}^n A_{\bmu_j}, \quad \bmu = \bigoplus_{i=1}^n \bmu_j, 
\label{eq:A_u_breakdown}
\end{align}
where $\bmu_j \in \mathbb{Z}_2 \times \mathbb{Z}_2$ is a co-ordinate in the phase space of the $j$th qubit \emph{only}. It is therefore instructive to calculate the single-qubit phase point operators, which are
\begin{align}
A_{0,0} &= \frac{1}{2}(\id + X + Z + iY),\notag \\ A_{0,1} &= \frac{1}{2}(\id - X + Z - iY), \notag \\ A_{1,0} &= \frac{1}{2}(\id + X - Z - iY), \notag \\ A_{1,1} &= \frac{1}{2}(\id - X - Z + iY).
\label{eq:single_qubit_A}
\end{align}

We next demonstrate how the explicit forms of single-qubit phase point operators can be leveraged via \eqref{eq:A_u_breakdown} to prove two further properties for general $n$-qubit phase point operators. In particular, we show how distinct $n$-qubit phase point operators are orthogonal under the Hilbert-Schmidt inner product:
\begin{enumerate}[label=\normalfont \textbf{(A\arabic*)}]
	\setcounter{enumi}{1}
	\item \textit{(Orthogonality).} \label{A_property:orthogonality} Let $A_\bmu$ and $A_\bmv$ be two $n$-qubit phase point operators. Then $\tr[A^\dagger_{\bmu} A_{\bmv}] = 2^n \delta_{\bmu,\bmv}$.
\end{enumerate}	
\begin{proof}
	Let us first decompose $\bmu$ and $\bmv$ as $\bmu = \bigoplus_{i=1}^n \bmu_j$ and ${\bmv = \bigoplus_{i=1}^n \bmv_j}$, where $\bmu_j$ and $\bmv_j$ are phase point co-ordinates on the $j$th qubit only. By \eqref{eq:single_qubit_A}, 
	\begin{align}
	\tr[A^\dagger_{\bmu_j} A_{\bmv_j}] = 2 \delta_{\bmu_j,\bmv_j}.
	\end{align}
	Therefore,
	\begin{align}
	\label{eq:A_orthogonal}
	\tr[A^\dagger_\bmv A_\bmu] = \prod_{j=1}^n \tr[A^\dagger_{\bmu_j} A_{\bmv_j} ]   &= \prod_{j=1}^n 2\delta_{\bmv_j,\bmu_j} \notag\\ &= 2^n \delta_{\bmu,\bmv}.
	\end{align} 
	as claimed.
\end{proof}

There are $\abs{P_n} = \abs{\mathbb{Z}_2^n \times \mathbb{Z}_2^n} = 4^n$ phase point operators on $n$-qubits. Property \ref{A_property:orthogonality} thus implies $\{A_\bmu\}_{\bmu \in \P_n}$ forms an orthogonal complex basis for the complex vector space  $\mathbb{C}^{2^n \times 2^n}$ of $2n \times 2n$ complex matrices under the Hilbert-Schmidt inner product. Therefore, $W_\rho$ is an \emph{informationally complete} representation of $n$-qubit states. More precisely, any $n$-qubit quantum state $\rho$ can be uniquely decomposed as
	\begin{equation}
	\rho = \sum_{\bmu} W_\rho(\bmu) A_{\bmu},
	\end{equation}
	where $W_\rho (\bmu) \coloneqq \frac{1}{2^n} \tr[A_\bmu^\dagger \rho]$ is a complex function on $\P_n$. 

Furthermore, every phase point operator has trace 1:
\begin{enumerate}[label=\normalfont \textbf{(A\arabic*)}]
	\setcounter{enumi}{2}
	\item \textit{(Unit trace).} \label{A_property:unit_trace} Let $A_\bmu$ be any $n$-qubit phase point operator. Then we have $\tr[A_{\bmu}] = 1$.
\end{enumerate}
\begin{proof}
	Let us first decompose $\bmu$ as $\bmu = \bigoplus_{j=1}^n \bmu_j$, where $\bmu_j$ is a point in the phase space of the $j$th qubit. We see that $\tr[A_{\bmu_j}] = 1$ from \eqref{eq:single_qubit_A}. Therefore,
	\begin{align}
	\tr[A_\bmu] = \prod_{j=1}^n \tr[A_{\bmu_j}] = \prod_{j=1}^n 1 = 1,
	\end{align} 
 which completes the proof.\end{proof}

Property \ref{A_property:unit_trace} implies that all $n$-qubit functions are \emph{normalized}. Since any $n$-qubit state $\rho$ has trace 1,
	\begin{align}
	1 = \tr[\rho] &= \tr\left[\sum_{\bmu \in \P_n} W_\rho(\bmu) A_\bmu\right] \notag \\
 &= \sum_{\bmu \in \P_n} W_\rho(\bmu) \tr[A_\bmu] \notag \\ &= \sum_{\bmu \in \P_n} W_\rho(\bmu),
	\end{align}
	where the last equality is established by the unit trace of phase point operators.

We will also find it useful to identify the following property of phase point operators:
\begin{enumerate}[label=\normalfont \textbf{(A\arabic*)}]
	\setcounter{enumi}{3}
	\item \textit{(Completeness Relation).} \label{A_property:identity}
	\begin{align}
	\sum_{\bmu \in \P_n} A_\bmu = 2^n \id_n
	\end{align}
\end{enumerate}
\begin{proof}
	Adopting the decomposition of each $A_\bmu$ in \eqref{eq:A_u_breakdown}, we see that
	\begin{align}
	\sum_{\bmu \in \P_n} A_\bmu &= \sum_{\bmu_1 \in \P_1},\dots,\sum_{\bmu_n \in \P_1} \left(\bigotimes_{j=1}^n A_{\bmu_j}\right) \notag \\ &= \bigotimes_{j=1}^n \left(\sum_{\bmu_j \in \P_1} A_{\bmu_j} \right).
	\end{align}
	Using the explicit forms of single-qubit phase-point operators in \eqref{eq:single_qubit_A}, we calculate that
	\begin{align}
	\sum_{\bmv \in \P_1} A_\bmv = 2 \id,
	\end{align}
	from which the result immediately follows.
\end{proof}

\subsection{Wigner representation for rebits}
\label{appx:rebit_rep}
Any $n$-qubit state $\rho$ can be decomposed as
\begin{equation}
\rho = \left[\frac{1}{2} \left(\rho + \rho^T\right)\right] + i\left[\frac{-i}{2}\left(\rho - \rho^T\right)\right],
\end{equation}
where the transposition is taken with respect to the computational basis. Because $\rho^* = \rho^T$ in any basis, we can identify 
\begin{align}
\rho^{(0)} &\coloneqq \frac{1}{2} (\rho + \rho^T) = \mathfrak{Re}(\rho), \\
\rho^{(1)} &= i\left[\frac{-i}{2}(\rho - \rho^T)\right] = \mathfrak{Im}(\rho), \label{eq:rho_real_im_defn}
\end{align}
i.e., $\rho^{(0)}$ and $\rho^{(1)}$ are respectively the real and imaginary components of the density matrix of $\rho$ in the computational basis. 

We first prove Lemma \ref{lemma:rho_W_real_im_correspondence}, which shows there is a direct correspondence between the real/imaginary components of a state's Wigner representation and those of its density matrix in the computational basis.
\rhoWrealim*
\begin{proof}
	Adopting the identification $\rho^{(0)} = \mathfrak{Re}(\rho)$ and ${\rho^{(1)} = \mathfrak{Im}(\rho)}$, we can then decompose $W_\rho(\bmu)$ as
	\begin{align}
	W_\rho(\bmu) = \frac{1}{2^n}\tr\left[A^\dagger_\bmu \rho^{(0)}\right] + i \frac{1}{2^n} \tr\left[A^\dagger_\bmu \rho^{(1)}\right].
	\end{align}
	Since $A_\bmu$ is real, and both $\rho^{(0)}$ and $\rho^{(1)}$ are real by construction, we conclude that $\tr\left[A^\dagger_\bmu \rho^{(0)}\right]$ and $\tr\left[A^\dagger_\bmu \rho^{(1)}\right]$ are both real for all $\bmu \in \P_n$. Therefore,
	\begin{align}
		&\mathfrak{Re}(W_\rho(\bmu)) = \frac{1}{2^n}\tr\left[A^\dagger_\bmu \rho^{(0)}\right] = W_{\rho^{(0)}}(\bmu)\\ &\mathfrak{Im}(W_\rho(\bmu)) = \frac{1}{2^n}\tr\left[A^\dagger_\bmu \rho^{(1)}\right] = W_{\rho^{(1)}}(\bmu)
	\end{align} 
\end{proof}

An $n$-rebit Wigner representation $W^{(0)}_\rho$ was introduced by Delfosse et al.~\cite{delfosse2015}, which is defined as
\begin{align}
W^{(0)}_\rho(\bmu) \coloneqq \frac{1}{2^n} \tr\left[A^{(0)\dagger}_\bmu \rho\right], \label{eq:W_delfosse}
\end{align} 
for all $\bmu \in \P_n$, where
\begin{align}
A_{\bmu}^{(0)} &\coloneqq  \frac{1}{2^n} \sum_{\bma \in \P^0_n} (-1)^{[\bmu,\bma]} D_{\bma} \label{eq:ax_rebit}
\end{align}
for the subset of phase space
\begin{align}	
\P^0_n \coloneqq \left\{ \bmu \, : \, \bmu_x \cdot \bmu_z = 0  \right\}. \label{eq:phase_space_rebit}
\end{align}

The definition of qubit displacement operators in \eqref{eq:Dx} implies the relations
\begin{align}
	&\bmu_x \cdot \bmu_z = 0 \Leftrightarrow D^\dagger_\bmu = D^T_\bmu = D_\bmu \label{eq:phase_space_symmetric}  \\
	&\bmu_x \cdot \bmu_z = 1 \Leftrightarrow D^\dagger_\bmu = D^T_\bmu = -D_\bmu, \label{eq:phase_space_antisymmetric} 
\end{align}
We then see from \eqref{eq:phase_space_symmetric} that $\P^{(0)}_n$ is the subset of phase space co-ordinates for \emph{real symmetric} displacement operators. The difference between $W_\rho$ and $W^{(0)}_\rho$ thus comes down to the fact that $A_\bmu$ sums over all displacement operators whereas $A^{(0)}_\bmu$ only sums over real symmetric ones, which implies $A^{(0)}_\bmu$ is itself real symmetric. 

We further observe from \eqref{eq:phase_space_antisymmetric} that the complement of $\P^{(0)}_n$ in $\P_n$,
\begin{align}
\P^{(1)}_n  &\coloneqq   \{ \bmu \, : \, \bmu_x \cdot \bmu_z = 1  \},
\end{align}
is the subset of phase space co-ordinates for \emph{real anti-symmetric} displacement operators. Paralleling \eqref{eq:ax_rebit}, we then introduce the set of real anti-symmetric phase point operators
\begin{align}
A^{(1)}_{\bmu} \coloneqq \frac{1}{2^n} \sum_{\bma \in \P^{(1)}_n} (-1)^{[\bmu, \bma]} D_{\bma} \text{ for any } \bmu \in \P_n,
\label{eq:A_IH}
\end{align}

By \eqref{eq:ax_rebit} and \eqref{eq:A_IH}, each $A_\bmu$ splits up as
\begin{align}
A_\bmu = A^{(0)}_\bmu + A^{(1)}_\bmu.
\end{align} 
We can correspondingly split up the Wigner representation of $\rho$ as 
\begin{align}
W_\rho(\bmu) &= \frac{1}{2^n}\tr[A^\dagger_\bmu \rho ]
=\frac{1}{2^n}\tr\left[\left(A^{(0)}_{\bmu} + A^{(1)}_\bmu \right)^\dagger\rho \right] \notag \\
&= \frac{1}{2^n}\tr\left[A^{(0)\dagger}_{\bmu} \rho \right] +  \frac{1}{2^n}\tr\left[A^{(1)\dagger}_{\bmu}\rho \right] \notag \\ &\eqqcolon W_{\rho}^{(0)}(\bmu) + W_{\rho}^{(1)}(\bmu),
\end{align}
where we have defined
\begin{equation}
W_{\rho}^{(1)}(\bmx) :=\frac{1}{2^n}\tr\left[A^{(1)\dagger}_{\bm{x}}\rho \right]=-\frac{1}{2^n}\tr\left[A^{(1)}_{\bm{x}}\rho \right].
\end{equation}

We can then prove
\begin{lemma}
	\label{lemma:ours_delfosse_relationship}
	Given any $n$-qubit state $\rho$, 
	\begin{align}
		&\mathfrak{Re}(W_\rho(\bmu)) = W^{(0)}_\rho(\bmu) \\
		&i\mathfrak{Im}(W_\rho(\bmu)) = W^{(1)}_\rho(\bmu)
	\end{align}
\end{lemma}
\begin{proof}
	Because $A^{(0)}_\bmu$ and $A^{(1)}_\bmu$ are respectively real symmetric and real anti-symmetric, we have ${A^{(0)\dagger}_\bmu = A^{(0)T}_\bmu = A^{(0)}_\bmu}$ and $A^{(1)\dagger}_\bmu = A^{(1)T}_\bmu = -A_\bmu$. Thus for $k=0,1$, we obtain
	\begin{align}
	\left[W^{(k)}_\rho(\bmu)\right]^* &= \frac{1}{2^n}\tr\left[A^{(k)\dagger}_{\bmu}\rho \right]^* \notag \\
	&= \frac{1}{2^n}\tr\left[(A^{(k)\dagger}_{\bmu}\rho)^\dagger \right] \notag \\
	&= \frac{1}{2^n}\tr\left[(-1)^k(A^{(k)\dagger}_{\bmu}\rho)\right]\notag \\ &= (-1)^k W^{(k)}_\rho(\bmu).
	\end{align}
	This implies $W^{(0)}_\rho(\bmu)$ is the real component of $W_\rho(\bmu)$ while $W^{(1)}_\rho(\bmu)$ is its imaginary component.
\end{proof}

By combining Lemmas \ref{lemma:rho_W_real_im_correspondence} and \ref{lemma:ours_delfosse_relationship}, we arrive at
\begin{align}
	W_{\mathfrak{Re}[\rho]}(\bmu) = W^{(0)}_\bmu(\rho).
\end{align} 
When $\rho$ is an $n$-rebit state, $\mathfrak{Re}(\rho) = \rho$, which implies ${W_\rho(\bmu) = W^{(0)}_\rho(\bmu)}$. 

\subsection{Wigner representation of qubit channels}
\label{appx:wigner_channel_rep}

We recall from the main text that the Wigner representation of a channel $\E: \B(\H^n_2) \rightarrow \B(\H^m_2)$ is the linear map $W_\E : \P_n \rightarrow \P_m$ on phase space defined as
\begin{align}
W_{\E}(\bmv |\bmu) \coloneqq 2^{2n} W_{\J(\E)}(\bmu \oplus \bmv),
\end{align}
for all $\bmv \in \P_m, \bmu \in \P_n$, where the Choi state~\cite{watrous_2018} of $\E$, $\J(\E)$, is defined as $\J(\E) = (\I \otimes \E) \ketbra{\phi^+_n}$ for the canonical maximally entangled state $\ket{\phi^+_n}$ on two sets of $n$ qubits,
\begin{align}
 	\ket{\phi^+_n}\coloneqq \frac{1}{\sqrt{2n}} \left(\sum_{\bmk \in \{0,1\}^n} \ket{\bmk} \otimes \ket{\bmk}\right).
\end{align}
One can straightforwardly verify that $\ket{\phi^+_n}$ is stabilized by $\langle Z_i Z_{n+i}, X_i X_{n+i}\rangle_{i=1,\dots,n}$ and is therefore a CSS state. 

The factorization property (\ref{A_property:qubit_phase_point_op_tensor_product}) of phase point operators implies that 
\begin{align}
W_{\E}(\bmv |\bmu) = \frac{2^n}{2^m} \tr \left[ \left( A^\dagger_{\bmu} \otimes A^\dagger_{\bmv} \right) \J(\E) \right].
\end{align}
Using the identity ${\E(X) = 2^n \tr_{1,\dots,n} \left[ (X^T \otimes \id^{\otimes m}) \J(\E) \right]}$ for transposition taken with respect to the computational basis, and recalling that $A_\bmu$ is real in the computational basis, we then conclude
\begin{align}
W_{\E}(\bmv |\bmu) = \frac{1}{2^m} \tr[A^\dagger_{\bmv} \E(A^*_{\bmu})] = \frac{1}{2^m} \tr[A^\dagger_{\bmv} \E(A_{\bmu})]. \label{eq:W_E_defn}
\end{align}

Therefore, if $\sigma = \E(\rho)$, then we obtain \eqref{eq:W_E_in_action} from the main text, i.e., 
\begin{align}
W_{\sigma}(\bmv) &= \frac{1}{2^m} \tr\left[A^\dagger_{\bmv}\E(\rho)\right] \notag \\
&= \frac{1}{2^m} \tr\left[\E\left( \sum_{\bmu \in \P_n} W_{\rho}(\bmu) A_{\bmu} \right) A^\dagger_{\bmv}\right] \notag \\ 
&= \frac{1}{2^m}\sum_{\bmu \in \P_n} \tr\left[ A^\dagger_{\bmv} \E\left(A_{\bmu} \right)\right]  W_{\rho}(\bmu) \notag \\ 
&= \sum_{\bmu \in \P_n} W_\E(\bmv | \bmu)  W_{\rho}(\bmu).
\end{align}
We thereby see that if $\E$ maps $\rho$ to $\sigma$, then $W_\E$ is a matrix that maps $W_\rho$ to $W_\sigma$, which justifies regarding $W_\E$ as the representation of $\E$ on phase space.

By property \ref{A_property:identity} of the phase point operators, we have that $\sum_{\bmv \in \P_m} A_\bmv = 2^m \id_m$. By applying this to the alternative formulation of $W_\E$ in \eqref{eq:W_E_defn}, we see that
\begin{align}
	\sum_{\bmv \in \P_m} W_\E(\bmv | \bmu) &= \frac{1}{2^m} \sum_{\bmv \in \P_m} \tr\left[A^\dagger_{\bmv} \E(A_{\bmu})\right]\notag \\ &= \frac{1}{2^m} \tr\left[\left(\sum_{\bmv \in \P_m} A_{\bmv}\right)^\dagger \E(A_{\bmu})\right] \notag \\ &= \frac{1}{2^m} \tr[2^m \id_m \E(A_\bmu)] \notag \\ &= \tr[\E(A_\bmu)].
\end{align}
Then recalling that $\tr[A_\bmu] = 1$ (property \ref{A_property:unit_trace}), we obtain \eqref{eq:W_normalization} from the main text, i.e.,
\begin{align}
	\sum_{\bmv \in \P_m} W_\E(\bmv | \bmu) = \tr[A_\bmu] = 1.
\end{align}
This means every column of $W_\E$ sums up to 1. 

Finally, we show that the representation we have chosen respects sequential and parallel composition of processes.

Let $\E: \B(\H^l_2) \rightarrow \B(\H^k_2)$ and $\F: \B(\H^n_2) \rightarrow \B(\H^m_2)$ be two multiqubit channels. Since $\{A_{\bmx}\}_{\bmx \in \P_m}$ are a complex orthogonal basis for $2^m \times 2^m$ complex matrices under the Hilbert-Schmidt inner product, we have that ${\F(A_\bmu) = \frac{1}{2^m} \sum_{\bmx \in \P_m} \tr[A^\dagger_{\bmx} \F(A_\bmu)]} A_{\bmx}$. Therefore, when $m = l$, we obtain
\begin{align}
	W_{[\E \circ \F]}(\bmv | \bmu) &= \frac{1}{2^k} \tr[A^\dagger_\bmv \E \circ \F (A_\bmu)] \notag\\
	&= \frac{1}{2^k} \tr\left[A^\dagger_\bmv \E \left(\frac{1}{2^m} \sum_{\bmx \in \P_m} \tr[A^\dagger_\bmx \F (A_\bmu)] A_\bmx\right)\right] \notag\\
	&= \sum_{\bmx \in \P_m} \frac{1}{2^k} \tr[A^\dagger_\bmv \E(A_\bmx)]  \frac{1}{2^m}  \tr[A^\dagger_\bmx \F (A_\bmu)] \notag\\
	&= \sum_{\bmx \in \P_m} W_\E(\bmv|\bmx) W_\F(\bmx|\bmu),
	\end{align}
or in matrix notation,
\begin{align}
	W_{\E \circ \F} = W_\E W_\F.
\end{align}
Furthermore, due to the factorization property of the phase point operators, we have that
\begin{align}
	W_{\E \otimes \F}(\bmx \oplus \bmy|\bmu \oplus \bmv) &= \frac{1}{2^{(m+k)}} \tr[A^\dagger_\bmx \otimes A^\dagger_\bmy \E \otimes \F (A_\bmu \otimes A_\bmv)] \notag\\ 
	&= \frac{1}{2^k} \tr[A^\dagger_\bmx \E(A_\bmu)] \frac{1}{2^m} \tr[A^\dagger_\bmy \F(A_\bmv)] \notag\\ 
	&= W_\E(\bmx|\bmu) W_\F(\bmy|\bmv),
\end{align} 
or in matrix notation
\begin{align}
	W_{\E \otimes \F} = W_\E \otimes W_\F.
\end{align}

\section{\uppercase{Completely CSS-preserving operations}} \label{appx:CSS}
\subsection{Completely CSS-preserving unitaries} \label{appx:CCSSP_vs_CSSP}
The group of CSS-preserving unitaries on $n$ qubits~\cite{delfosse2015} can be generated as
\begin{equation}
\G_+(n) \coloneqq \langle H^{\otimes n}, \CNOT(i,j), X_i, Z_i\rangle_{i,j=1,\dots,n, i \neq j}.
\end{equation}
We will also find it useful to note the following conjugation relations of the collective Hadamard gate,
\begin{align}
H^{\otimes n} \CNOT(i,j) &=\CNOT(j,i) H^{\otimes n} \label{eq:conjugation_collective_Hadamard_1}\\
H^{\otimes n} X(\bma) &= Z(\bma) H^{\otimes n}, \label{eq:conjugation_collective_Hadamard_2}
\end{align}
which respectively hold for all $ i,j \in \{1,\dots n\}$ where $i \neq j$ and $n$-bit strings $\bma \in \{0,1\}^n $.

\begin{lemma} \label{lemma:CCSSP_vs_CSSP}
The group of \emph{completely CSS-preserving unitaries} on $n$ qubits is 
\begin{equation}
	\G(n) \coloneqq \langle \CNOT(i,j), Z_i, X_i \rangle_{i,j = 1,\dots,n, i \neq j}.
\end{equation}
\end{lemma} 
\begin{proof}
Let $U_+$ be any CSS-preserving unitary on $n$ qubits. We first observe that $U_+$ is either in $\G(n)$ or is a unitary from $\G(n)$ followed by the collective Hadamard gate on $n$ qubits, i.e., $U_+ = \left[H^{\otimes n}\right]^b U$ for some binary digit $b \in \{0,1\}$ and unitary $U \in \G(n)$. This follows from the conjugation relations \eqref{eq:conjugation_collective_Hadamard_1} and \eqref{eq:conjugation_collective_Hadamard_2} alongside the fact that $H^{\otimes n}$ is self-inverse.

If $U_+$ is in $\G(n)$, then because $\G(n)$ is a subset of $\G_+(n')$ for all $n' \ge n$, $U_+$ must be \emph{completely} CSS-preserving. If $U_+$ is not in $\G(n)$, then we must have $U_+ = H^{\otimes n} U$ for some $U \in \G(n)$, which implies $U_+$ cannot be completely CSS-preserving since $H^{\otimes n}$ is not completely CSS-preserving. Therefore, $U_+$ is completely CSS-preserving if and only if it is in $\G(n)$.
\end{proof}

\subsection{Completely CSS-preserving measurements}
Throughout the rest of this appendix, we extend, wherever necessary, the notion of being completely CSS-preserving to trace-decreasing operations -- i.e., a trace-decreasing operation $\E$ from $n$ to $n'$ qubits is completely CSS-preserving if, given any CSS state $\rho$ on $m+n$ qubits, we have that $(\I_m \otimes \E)(\rho)$ is always a (possibly subnormalised) CSS state on $(m+n')$ qubits. 

The projective measurement of any $n$-qubit Pauli observable $S$ is carried out using projectors ${P(\pm S) \coloneqq \frac{1}{2}(\id_n \pm S)}$ corresponding to the $\pm 1$ outcomes. Post-selection for the $\pm 1$ outcome is then carried out by the operation ${\P(\pm S) \coloneqq P(\pm S) (\cdot) P(\pm S)}$.

\begin{lemma}
	\label{lemma:CSS_proj}
	Post-selecting the $\pm 1$ outcome in the projective measurement of a CSS observable is completely CSS-preserving.   
\end{lemma}
\begin{proof}
	Let $S$ be a CSS observable on $n$ qubits and $\ket{\psi}$ be a CSS state on $m+n$ qubits for any $m \ge 0$. Then let $S_1,\dots,S_{m+n}$ be a set of $m+n$ CSS observables that generate the stabilizer group $\S(\ket{\psi})$ of $\ket{\psi}$. 
	
	Post-selecting the $\pm 1$ outcome in a projective measurement of $S$ on the last $n$ qubits of $\ket{\psi}$ yields the possibly subnormalised output
	\begin{align}
	\ket{\phi_\pm} &\coloneqq \left[\id^{\otimes m} \otimes P(\pm S)\right] \ket{\psi} \notag \\  
	&=   \left[\frac{1}{2}\left(\id^{\otimes m+n} \pm \id_m \otimes S\right)\right] \ket{\psi} \notag \\ 
	&=  P(\pm S') \ket{\psi}, 
	\end{align}     
	where we have defined the CSS observable $S' \coloneqq \id^{\otimes m} \otimes S$. There are now two possibilities:
	\begin{enumerate}
		\item[(a)] that $S'$ commutes with every generator of $\S(\ket{\psi})$, so $S'$ or $-S'$ must stabilize $\ket{\psi}$. Therefore, either $\ket{\phi_+} = \ket{\psi}$ and $\ket{\phi_-} = 0$ or vice versa, so $\ket{\phi_\pm}$ are possibly subnormalised CSS states.
		
		\item[(b)] that, without loss of generality, $S'$ does \emph{not} commute with just one CSS observable $S_1$ that generates $\S(\ket{\psi})$. This follows from the fact that, in any set of $m+n$ CSS observables that generate $\S(\ket{\psi})$, those that do \emph{not} commute with $S'$ must \emph{all} be $X$ or $Z$-type, so by picking one such generator and multiplying all others by it, we obtain another set of $m+n$ CSS observables generating $\S(\ket{\psi})$ in which only one generator does \emph{not} commute with $S'$. Then the states $\ket{\phi_\pm}$ have norm $\frac{1}{\sqrt{2}}$ and are stabilized respectively by $\langle \pm S, S_2,\dots,S_{m+n}\rangle$, so $\ket{\phi_\pm}$ are subnormalised CSS states. 
	\end{enumerate}
	Therefore, given any pure CSS state $\ket{\psi}$ on $m+n$ qubits, post-selecting the $\pm 1$ outcome in the projective measurement of a CSS observable on the last $n$ qubits of $\ket{\psi}$ always produces a (possibly subnormalised) CSS state. Since every CSS state is a statistical mixture of pure CSS states, we arrive at the Lemma result.
\end{proof}

\subsection{CSS circuits}\label{appx:CSS_operations}
In this section, we show that the subset of stabilizer operations covered by \lemref{lemma:CCSSP_powers}, which we referred to as \emph{CSS circuits}, are completely CSS-preserving.

To reiterate, a \emph{CSS circuit} is any sequence of the following four \emph{primitive CSS channels}:
\begin{enumerate}
	\item Introducing a CSS state on any number of qubits,
	\item Performing a completely CSS-preserving unitary,
	\item Projective measurement of any CSS observable, with the possibility of performing different sequences of primitive CSS channels depending on outcome, 
	\item Discarding any number of qubits, 
\end{enumerate}
as well as statistical mixtures of such sequences. We emphasise that all CSS circuits are trace-preserving. 

Any sequence of primitive CSS channels can be executed as a binary tree where the root node represents inputting qubits, the leaf nodes represent outputting qubits, and internal nodes represent primitive CSS channels. An illustrative example is provided in \figref{fig:binary_tree_conversion_eg}, from which we see that a sequence of primitive CSS channels can produce outputs distinguished by the sequences of measurement outcomes leading up to them.

\begin{figure}[h]
    \centering \includegraphics[width=0.9\linewidth]{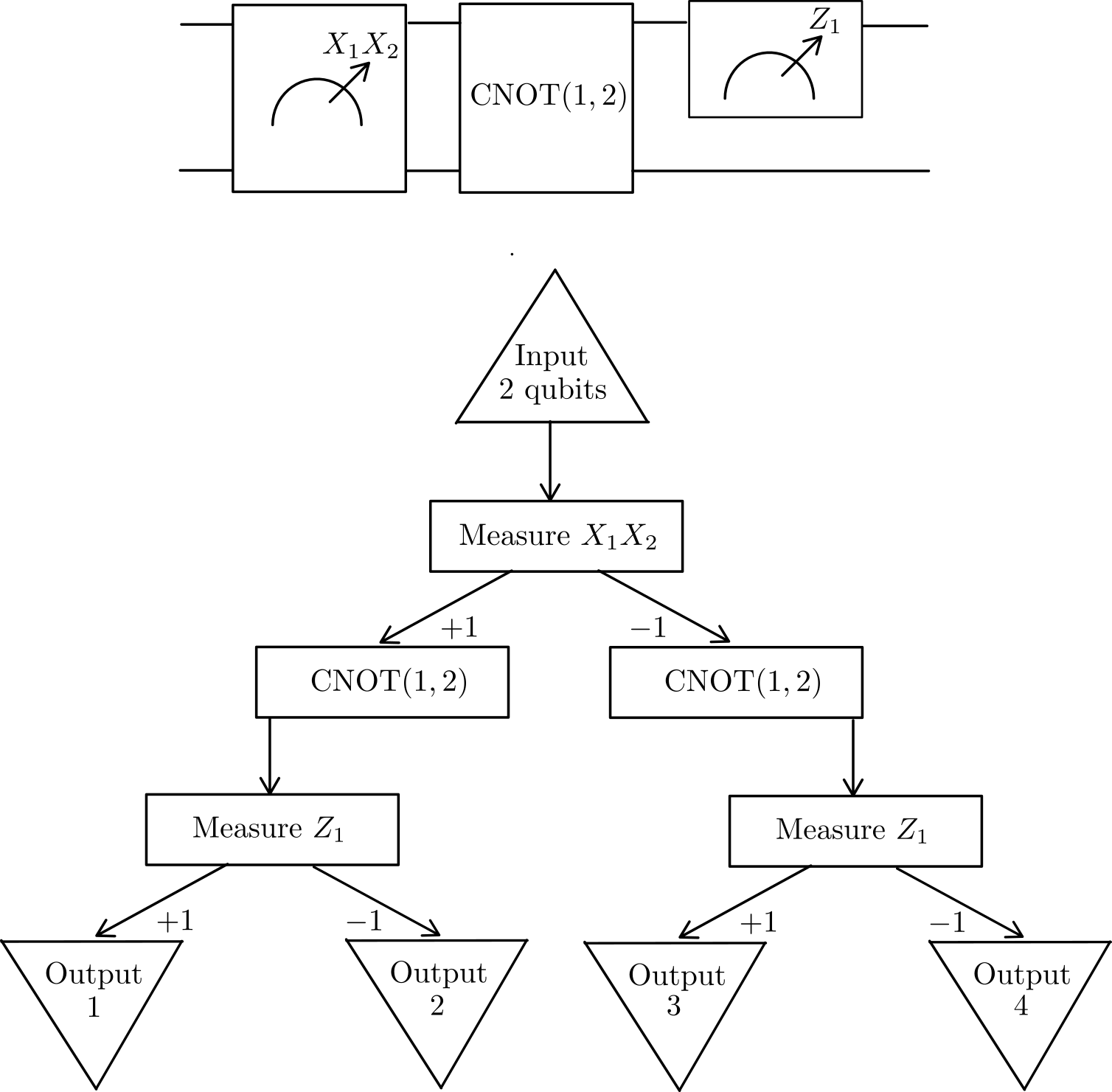}
	\caption{The sequence of primitive CSS channels on the top can be executed as the binary tree on the bottom.}
	\label{fig:binary_tree_conversion_eg}
\end{figure}

Different numbers of ancillary qubits may be introduced on different branches of a tree representing a sequence of primitive CSS channels. However, one can arbitrarily increase the number of qubits introduced on any branch, without affecting what it does, by introducing the maximally mixed state on qubits that are immediately discarded just before the branch's output. Furthermore, different branches may have different lengths. However, one can arbitrarily lengthen any branch, without affecting what it does, by inserting identity channels just before the branch's output. Since introducing the maximally mixed state and the identity channel are both primitive CSS channels, we can, without loss of generality, only consider sequences of primitive CSS channels executed as binary trees where every branch has the same length and introduces the same number of ancillary qubits. 

\begin{lemma}\label{lemma:CSS_operation_canonical_form}
	A CSS circuit $\E$ on $n$ qubits is a statistical mixture of channels $\E_i$ representing sequences of primitive CSS channels, in which each sequence $\E_i$ is a sum of (possibly trace-decreasing) operations $\E_{i,j}$ generating its distinguishable outputs. Thus one can write
	\begin{align}
	\E(\rho) &= \sum_i p_i \E_i,  \text{ where } \E_i = \sum_j \E_{i,j},  \notag \\\E_{i,j} &= \tr_{\R}\left[K_{i,j}(\rho\otimes \sigma_{i,j})K_{i,j}^\dagger\right], \text{ and }  \notag\\  K_{i,j} &= \prod_{l=1}^N P(S_{(i,j),l}) U_{(i,j),l},
	\end{align}
	in which $\{p_i\}$ forms a probability distribution, $\sigma_{i,j}$ is a CSS state on $m$ ancillary qubits, $\R$ is a subset of the $(n+m)$ input and ancillary qubits, $U_{(i,j),l}$ is a completely CSS-preserving unitary and $P(S_{(i,j),l})$ projects onto the $+1$ eigenspace of the CSS observable $S_{(i,j),l}$. Moreover, $P(S_{(i,j),1}),\dots,P(S_{(i,j),N})$ gives the sequence of measurement outcomes that operationally distinguish the $j$th possible output of sequence $\E_i$.
\end{lemma}
\begin{proof}
	Without loss of generality, every branch of every sequence forming the mixture of $\E$ introduces the \emph{same} number of ancillary qubits $m$ and has \emph{same} length $N$. One can show that the $j$th branch of the $i$th sequence must generate the channel $\E_{i,j}$ by induction over the steps of the branch. 
\end{proof}

Because each possible output from a sequence of primitive CSS channels results from a unique sequence of measurement outcomes, it is operationally meaningful to prepare a state conditioned upon obtaining the $j$th possible output in the $i$th sequence from the statistical mixture forming $\E$. We therefore have the following corollary:
\begin{corollary} \label{corollary:tagging}
	We can record which operationally distinguishable output $\E_{i,j}$ from a CSS circuit $\E$ was obtained using a classical register, 
	\begin{align}
	\E'(\rho) \coloneqq \sum_{i,j} p_i\ \E_{i,j}(\rho) \otimes \ketbra{r_{i,j}},
	\end{align}
	where $\ket{r_{i,j}}$ is a computational basis state on multiple qubits. We note that, since $\ket{r_{i,j}}$ is a CSS state, $\E'$ can be carried out as a CSS circuit. 
\end{corollary}
In the next subsection, we use this corollary to obtain the trace-preserving CSS code projection studied throughout this paper.

We are finally in a position to prove \lemref{lemma:CCSSP_powers}, which is reproduced below. To this end, it is convenient to define the unique 0-qubit state 1 as CSS.
\begin{customlem}{3}
    Any CSS circuit is completely CSS-preserving.
\end{customlem}
\begin{proof}
	The decomposition of CSS circuits given in \lemref{lemma:CSS_operation_canonical_form} implies that if (i) performing completely CSS-preserving unitaries, (ii) conditioning on the $+1$ outcome in the projective measurement of a CSS observable, (iii) introducing a CSS state and (iv) discarding any number qubits are completely CSS-preserving, then all CSS circuits are completely CSS-preserving. 
	
	Now (i) is completely CSS-preserving by definition, we proved that (ii) is completely CSS-preserving in \lemref{lemma:CSS_proj}, and since the tensor product of two CSS states is always a CSS state, (iii) is completely CSS preserving.
	
	Therefore, to prove that all CSS circuits are completely CSS-preserving, we just have to prove that discarding any number of qubits is completely CSS-preserving. 
	
	Consider discarding $l$ qubits from $n$, where $n \ge l \ge 1$. Since we can freely relabel subsystems, we need only consider discarding the \emph{last} $l$ qubits. Let $\ket{\psi}$ be a pure CSS state on $m+n$ qubits for any $m \ge 0$. Discarding the last $l$ qubits of $\ket{\psi}$ then produces the state $\sigma \coloneqq \I_m \otimes \tr_{n-l+1,\dots,n}[\ketbra{\psi}]$. Since tracing out is unaffected by first performing a computational basis measurement on the last $l$ qubits, we have that
    \begin{align}
        \sigma = \sum_{\bmk \in \{0,1\}^l} &\tr_{m+n-l+1,\dots,m+n}\big[\id^{\otimes (m+n-l)} \otimes \ketbra{\bmk} \nonumber\\
        &\hspace{5pt} (\ketbra{\psi}) \id^{\otimes (m+n-l)}  \otimes \ketbra{\bmk} \big].
    \end{align}  
    We then observe that 
    \begin{align}
        &\id^{\otimes (m+n-l)} \otimes \ketbra{\bmk} \notag \\
       = &P((-1)^{k_1} Z_{m+n-l+1}) \circ \dots \circ P((-1)^{k_l} Z_{m+n}),
    \end{align}
    and so by \lemref{lemma:CSS_proj}, $\left(\id^{\otimes m+n-l} \otimes \ketbra{\bmk}\right)\ket{\psi}$ is a possibly subnormalised pure CSS state ${\sqrt{p_\bmk} \ket{\phi_\bmk} \otimes \ket{\bmk}}$, where $p_{\bmk}$ is the probability of getting the $\ket{\bmk}$ outcome in the computational basis measurement, and $\ket{\phi_\bmk}$ must be a (normalised) CSS state to keep the full state CSS. We thus obtain
	\begin{align}
		\sigma &= \sum_{\bmk \in \{0,1\}^l} \tr_{m+n-l+1, \dots, m+n}\left(p_\bmk \ketbra{\phi_\bmk} \otimes \ketbra{\bmk}\right) \notag \\ &= \sum_{\bmk \in \{0,1\}^l} p_\bmk \ketbra{\phi_\bmk},
	\end{align} 
	which is a CSS state on $(m+n-l)$ qubits. We conclude that, given any pure CSS state on $m+n$ qubits, discarding any $l \le n$ of its final $n$ qubits always produces a CSS state. This is true if and only if discarding any number of qubits is completely CSS-preserving.
\end{proof}

\subsubsection{Omission of the collective Hadamard gate} \label{appx:collective_hadamard_omission}
The collective Hadamard gate promotes CSS circuits to a subset of stabilizer circuits where CSS states play role of stabilizer states. One can reasonably ask why we have excluded the collective Hadamard gate from the construction of CSS circuits. Our justification is that one can conjugate the collective Hadamard gate past any primitive CSS channel and leave another primitive CSS channel behind. This follows from the conjugation relations given by \eqref{eq:conjugation_collective_Hadamard_1} and \eqref{eq:conjugation_collective_Hadamard_2} for completely CSS-preserving unitaries and projective measurements of CSS observables, from the cyclic property of the trace for discarding qubits, and from
\begin{align}
    &H^{\otimes n} \otimes \id^{\otimes m} (\rho \otimes \sigma) H^{\otimes n} \otimes \id^{\otimes m}  \nonumber\\
    =&H^{\otimes (n+m)} (\rho \otimes (H^{\otimes m} \sigma H^{\otimes m})) H^{\otimes (n+m)} 
\end{align} 
for introducing a CSS state on $m$ ancillary qubits to an $n$-qubit system, where we note that $H^{\otimes m} \sigma H^{\otimes m}$ is also a CSS state because $H^{\otimes m}$ is CSS-preserving on $m$ qubits. Therefore, circuits from this wider subset are operationally equivalent to CSS circuits followed by the collective Hadamard gate conditioned upon obtaining certain outputs, and are therefore not more powerful as magic distillation protocols. 

\subsection{CSS code projections}\label{appx:CSS_code_proj}
An $[[n,k]]$ CSS code, where $n \ge 1$ and $n > k \ge 0$, is a vector space $\C$ stabilized by a subgroup $\S$ of $n$-qubit Pauli observables such that $-\id^{\otimes n} \notin \S$ and $\S$ can be generated from $n-k$ independent and commuting \emph{CSS observables} $S_1,\dots,S_{n-k}$, of which none is the identity.

\begin{lemma}
	\label{lemma:CNOT_gaussian_elim}
	Let $\S \coloneqq \langle (-1)^{b_i} S_i \rangle_{i=1,\dots,n-k}$ be the stabilizer group of an $[[n,k]]$ CSS code, where each $S_i$ is a positive CSS observable and each $b_i$ is a binary digit. Then there exists a completely CSS-preserving unitary $U$ such that
	\begin{align}
	U^\dagger [(-1)^{b_i}S_i] U = \begin{cases}
	   Z_{k+i}& \text{ if $S_i$ is $Z$-type,}\\
	   X_{k+i}& \text{ if $S_i$ is $X$-type.}
	\end{cases}
	\end{align}
\end{lemma}
\begin{proof}
	Let $r$ be the number of $Z$-type generators for $\S$. We first construct the completely CSS-preserving unitary $U$ for the case where all $Z$-type generators of $\S$ appear before $X$-type ones, i.e., 
	\begin{align}
	S_i = \begin{cases}
	Z(\bmu_i)& \text{ for } 1 \le i \le r\\
	X(\bmv_i)& \text{ for } r < i \le n-k
	\end{cases} \label{eq:local_2}
	\end{align}
	where each $\bmu_i, \bmv_i$ is a non-zero $n$-dimensional binary vector. 
	
	Let us define $\Z$ as the set of $Z$-type generators for $\S$ \emph{without their signs}, i.e., $\Z \coloneqq \{Z(\bmu_1),\dots,Z(\bmu_r)\}$. We now prove that there exists a sequence of $\CNOT$ operations that transforms $Z(\bmu_i)$ to $Z_i$ for all $1 \le i \le r$.  
	
	We observe that $\Z$ is a subset of positive $Z$-type observables on $n$ qubits, which form an $n$-dimensional vector space $V$ over the field $\mathbb{F}_2$. Choosing the basis $\{Z_1,\dots,Z_n\}$ for $V$, we can simply write $Z(\bmu_i)$ as $\bmu_i$, and we further have that $\{\bmu_1,\dots,\bmu_r\}$ are linearly independent. Therefore, the matrix $M_Z$ formed by taking members of $\Z$ as columns has rank $r$. Using Gauss-Jordan elimination, we can convert $M_Z$ into its unique reduced row echelon form $R_Z$,     
	\begin{align}
	M_Z \coloneqq \begin{bmatrix}
	\vert & \dots& \vert \\
	\bmu_1 & \dots& \bmu_r \\
	\vert & \dots& \vert
	\end{bmatrix} \mapsto R_Z \coloneqq \begin{bmatrix}
	I_{r,r}\\
	0_{n-r,r}
	\end{bmatrix}, \label{eq:Gauss_Jordan_Z}
	\end{align}
 	where $I_{r,r}$ is an $r \times r$ identity matrix while $0_{n-r,r}$ is a $(n-r) \times r$ null matrix. 
 	
 	On vector spaces over $\mathbb{F}_2$, Gauss-Jordan elimination consists of row swaps and additions. We now show how both can be done on any matrix $M$ whose columns are elements of $V$ in the basis $\{Z_1,\dots,Z_n\}$ using $\CNOT$-gates: 
	\begin{enumerate}
		\item \textit{Swapping rows $j$ and $l$.} This corresponds to swapping qubits $j$ and $l$, which is carried out by performing $\CNOT(l,j) \CNOT(j,l) \CNOT(a,b)$ on each positive $Z$-type observable forming a column in $M$.
		
		\item \textit{Adding row $j$ to row $l$.} The action of $CNOT(j,l)$ on $Z_m$ is
		\begin{align}
		\CNOT(l,j) Z_m \CNOT(l,j) = \begin{cases}
		Z_j Z_l& \text{ for } m = j,\\
		Z_m& \text{ otherwise.}
		\end{cases}
		\label{eq:CNOT_workings}
		\end{align}
		Therefore, given any vector $\bmu$ in $V$, we have that 
		\begin{align}
		\CNOT(l,j) Z(\bmu) \CNOT_{l,j} &= Z(\bmu + u_j \bme_l) 
		\end{align} 
		where arithmetic is modulo 2 and $\bme_l$ gives co-ordinates for the $l$th basis vector of $V$. In words, $\CNOT(j,l)$ adds the $j$th component of $\bmu$ to the $l$th. Therefore, performing $\CNOT(l,j)$ on each positive $Z$-type observable forming a column in $M$ would add the $j$th row of $M$ to the $l$th row.   
	\end{enumerate} 

    We conclude that there exists a sequence of $\CNOT$-gates that, when performed on each element of $\Z$, accomplishes the Gauss-Jordan elimination in \eqref{eq:Gauss_Jordan_Z}. Denoting this sequence of $\CNOT$-gates by the unitary $U_Z$, we have that $U_Z (Z(\bmu_i)) U^\dagger_Z = Z_i$ for all $1 \le i \le r$. 
	
	We next consider $X$-type generators for $\S$ \emph{without their signs}. The action of $\CNOT(j,l)$ on $X_m$ is
	\begin{align}
	\CNOT(j,l) X_m \CNOT(j,l) = 
	\begin{cases}
	X_jX_l& \text{ for } m = j,\\
	X_m& \text{ otherwise,}
	\end{cases}
	\label{eq:CNOT_workings_Z}
	\end{align}
	so $U_Z$ only transforms positive $X$-type observables into other positive $X$-type observables. In other words, we can find non-zero $n$-bit strings $\{\bmv'_{r+1},\dots,\bmv'_{n-k}\}$ such that $U_Z(X(\bmv_i))U^\dagger_Z = X(\bmv'_i)$ for all $r < i \le n-k$. 
	
	However, since the $Z$-type generators of $\S$ commute with the $X$-type generators, $X(\bmv'_i)$ must commute with $Z_1,\dots,Z_r$. Therefore, $X(\bmv'_i)$ must act trivially on qubits 1 through $r$, so the first $r$ bits of $\bmv'_i$ must be 0. Therefore $X(\bmv'_i) = \id^{\otimes r} \otimes X(\bmv''_i)$ forall $r < i \le n-k$, where $\bmv''_i$ are the last $(n-r)$ bits of $\bmv'_i$.
	
	Everything we have done for $\mathcal{Z}$ can now be repeated for $\mathcal{X} \coloneqq \{X(\bmv''_{r+1}),\dots,X(\bmv''_{n-k})\}$ on qubits $r+1$ through $n$. The only thing that needs to be checked is that row addition in any matrix whose columns are elements from the vector space of positive $(n-r)$-qubit $X$-type observables can be performed by $\CNOT$-gates. This can be confirmed using \eqref{eq:CNOT_workings_Z}, which implies
	\begin{align}
		\CNOT(j,l) X(\bmv) \CNOT(j,l) &= X(\bmv + v_j \bme_l)
	\end{align}
	where arithmetic is modulo 2 so $\CNOT(j,l)$ adds the $j$th component of $\bmv$ to the $l$th.
	
	We conclude that there also exists a sequence $U_X$ of $\CNOT$-gates such that $U_X(X(\bmv'_i))U^\dagger_X = X_i$ for all ${r < i \le n-k}$. Furthermore, $U_X$ acts trivially on qubits 1 through $r$, which implies $U_X(Z_i)U^\dagger_X = Z_i$ for all $1 \le i \le r$. Defining $\U_X \coloneqq U_X(\cdot)U^\dagger_X$ and $\U_Z \coloneqq U_Z(\cdot)U^\dagger_Z$, we therefore have
	\begin{align}
	&(\U_Z \circ \U_X) [(-1)^{b_i} S_i]  \notag \\ 
 =&\begin{cases}
	(-1)^{b_i} X_i& \text{ for } 1 \le i \le r\\
	(-1)^{b_i} Z_i& \text{ for } r < i \le n-k.
	\end{cases}  
	\end{align}
	We now define unitaries $U_C$ and $U_{\text{SWAP}}$, which respectively remove the signs from the generators of $\S$ and moves qubits 1 through $n-k$ to $k+1$ through $n$,
    \begin{align}
        U_C &\coloneqq \left[\prod_{i=r+1}^{n-k} X_i^{b_i}\right] \left[ \prod_{i=1}^r Z_i^{b_i}\right], \notag \\ 
        U_{\text{SWAP}} &\coloneqq \prod_{i=0}^{n-k-1} \text{SWAP}(n-i,n-k-i),
    \end{align}
    where $\text{SWAP}(i,j)\coloneqq\CNOT(i,j)\CNOT(j,i)\CNOT(i,j)$ swaps qubits $i$ and $j$. Defining ${\U_C \coloneqq U_C (\cdot) U_C}$ and $U_{\text{SWAP}} \coloneqq U_{\text{SWAP}}(\cdot)U_{\text{SWAP}}$, we obtain
	\begin{align}
	&(\U_{\text{SWAP}} \circ \U_C \circ \U_Z \circ \U_X) [(-1)^{b_i} S_i]  \notag \\ =& \begin{cases}
	X_i& \text{ for } 1 \le i \le r\\
	Z_i& \text{ for } r+1 \le i \le n-k.
	\end{cases}  
	\end{align}
	Since $U^\dagger \coloneqq U_C U_Z U_X$ is formed from $\CNOT$-, single-qubit $X$- and $Z$- gates, it is a completely CSS-preserving unitary that accomplishes the Lemma's claim for the ordering of $S_i$ given by \eqref{eq:local_2}. 
	
	By appropriate swaps among the last $n-k$ qubits, which we have seen is completely CSS-preserving, we can construct a completely CSS-preserving unitary accomplishing the Lemma's claim for any ordering of $S_i$.
\end{proof} 
As an immediate consequence of \lemref{lemma:CNOT_gaussian_elim},
\begin{corollary}\label{corollary:completely_CSS_decoding}
	Every $[[n,k]]$ CSS code has a completely CSS-preserving encoding unitary.
\end{corollary}
\begin{proof}
 Let $\C$ be an $[[n,k]]$ CSS code whose stabilizer group is generated by CSS observables $S_1,\dots,S_{n-k}$, and let $\bms$ be a $k$-bit string (if $k=0$, then $\bms$ is the empty string). By \lemref{lemma:CNOT_gaussian_elim}, there exists a completely CSS-preserving unitary $U$ such that
\begin{align} 
	\ket{\bms_\C} &\coloneqq U\left(\ket{\bms} \bigotimes_{i=1}^{n-k} \ket{\phi_i}\right) \notag\\ \text{ where } \ket{\phi_i} &\coloneqq \begin{cases}
		\ket{0} &\text{ if $S_i$ is $Z$-type}\\ 
		\ket{+} &\text{ if $S_i$ is $X$-type.}
	\end{cases} \label{eq:ur_logical_basis_states} 
\end{align}
In words, $U$ can encode any computational basis state $\ket{\bms}$ on $k$ physical qubits (if $k=0$, then $\ket{\bms} = 1$) as a \emph{CSS logical basis state} $\ket{\bms_\C}$ in $\C$.
\end{proof} 

With Corollary \ref{corollary:completely_CSS_decoding} in hand, we can show that every trace-preserving CSS code projection can be executed as a CSS circuit, which leads to the following Lemma from the main text:
\ntoonestoch*
\begin{proof}
 Let $\C$ be an $[[n,k]]$ CSS code generated by CSS observables $S_1,\dots,S_{n-k}$.  We can then define a quantum channel $\F$ that carries out the following sequence of primitive CSS channels:
\begin{enumerate}
	\item Projectively measure $S_1,\dots,S_{n-k}$, which is equivalent to the syndrome measurement of $\C$. 
	\item If the no-error syndrome is obtained, decode onto the first $k$ qubit using a completely CSS-preserving decoding unitary, which always exists in accordance with Corollary \ref{corollary:completely_CSS_decoding}, and discard the final $n-k$ qubits.
	\item Otherwise, discard all qubits and prepare a $k$-qubit CSS state $\sigma$. 
\end{enumerate}
Since $\F$ is a CSS circuit by construction, it is completely CSS-preserving. By Corollary \ref{corollary:tagging}, we can convert $\F$ into a trace-preserving code projection $\E$ for $\C$ by distinguishing the output of the no-error syndrome from those of all other syndrome using a single-qubit classical register, i.e. by preparing an ancillary qubit in the state $\ket{0}$ if we obtain the no-error syndrome output, and in the state $\ket{1}$ otherwise. Corollary \ref{corollary:tagging} confirms that $\E$ can also be executed as a CSS circuit and is therefore completely CSS-preserving. We conclude by \lemref{lemma:CCSSP_powers} that $\E$ is stochastically represented.  
\end{proof}

\section{\uppercase{Entropic constraints on completely CSS-preserving protocols}}
\label{appx:general_bound}

In this section we prove \thmref{thrm:general_qubit_bound}, which gives entropic constraints on generic completely CSS-preserving protocols for qubits.

We first define the sets of positively-represented and real-represented quantum states in any Wigner representation $W$ as
\begin{align}
    W^+ &\coloneqq \{\rho :  W_\rho(\bmu) >0 , W_\rho(\bmu)\in \mathbb{R} , \forall \bmu \in \P\}, \\
    W^\mathbb{R} &\coloneqq \{\rho :  W_\rho(\bmu)\in \mathbb{R} , \forall \bmu \in \P\},
\end{align}

We then have the following Lemma, which is a generalisation of Theorem~11 of Ref.~\cite{koukoulekidis2022constraints}. 
\begin{lemma} \label{lemma:nick_restatement}
Let $\rho$ and $\tau$ be states of a $d$-dimensional qudit such that $\rho \in W^{\mathbb{R}}$ and $\tau \in W^+$ in some generalised Gross's Wigner representation $W$. Furthermore, let ${\E : \B(\H_{d}) \mapsto \B(\H_{d'})}$ be a stochastically represented channel. Then the $\alpha$-R\'{e}nyi divergence $D_\alpha (\cdot || \cdot)$ is well-defined and satisfies the following properties for $\alpha\in \A$:
\begin{enumerate}
    \item \label{property:D_nonnegative} $D_\alpha(W_\rho || W_\tau) \ge 0$. 
    \item \label{property:D_zero} $D_\alpha(W_\rho || W_\tau)=0$ if and only if $\rho =\tau$.
    \item \label{property:D_multiplicative} $D_\alpha (W_{\rho^{\otimes n}} ||W_{\tau^{\otimes n}}) = n D_\alpha(W_\rho || W_\tau) $ for all $n\in \mathbb{N}$.
    \item \label{property:D_dataprocessing} $D_\alpha(W_\rho|| W_\tau) \ge D_\alpha(W_{\E(\rho)}|| W_{\E(\tau)}) $ for all stochastically represented $\E$ such that $\E(\tau) \in \W^+$.
\end{enumerate}
\end{lemma}
\begin{proof} 
In general, $W_\rho$ is a quasiprobability distribution, but for $\alpha \in \A $ we see that $W_\rho(\bmu)^\alpha\ge 0$ for all $\bmu \in \P$. Therefore $D_\alpha (W_\rho || W_\tau)$ is always well-defined and real-valued. The proofs of \ref{property:D_nonnegative}-\ref{property:D_dataprocessing} are then identical to the proof given for Theorem~11 in Ref.\cite{koukoulekidis2022constraints}. 
\end{proof}

Importantly for our purposes, this abstract but general result applies to input and output systems of any (even or odd) finite dimension. With this in hand, we can now give a proof of \thmref{thrm:general_qubit_bound}, which we restate for clarity:

\generalBound*
\begin{proof}
Since $\tau$ and $\tau'$ are in the interior of $\D_{css}$, they are positively represented. Moreover, we have established in \thmref{theorem:J_CSS_stoch_rep} that every trace-preserving and completely CSS-preserving operation is stochasically represented. The results of \lemref{lemma:nick_restatement} thus apply, from which properties \ref{property:D_multiplicative} and \ref{property:D_dataprocessing} combine to give the Lemma result.
\end{proof}

\section{\uppercase{Entropic constraints on code projection protocols}}

\subsection{Structure of the necessary conditions}
\label{appx:structure_necessary_condition}

We recall from \secref{sec:structure_CSS} that if there exists an $n$-to-$k$ CSS code projection that achieves the distillation ${\rho^{\otimes n} \mapsto \rho'}$ with acceptance probability $p$, then 
\begin{align}
    \Delta D_\alpha \ge 0
\end{align}
for all $\alpha\in \A$, where
\begin{align}
   \Delta D_\alpha
   =n D_\alpha(W_\rho ||W_{\frac{\id}{2}}) - D_\alpha\left(W_{\rho_p} \, || \, W_{\tau_{n,k}}\right).
   \label{eq:Delta_alpha_n}
\end{align}
We recall the following output states from the system and reference processes,
\begin{align}
    \rho_p &\coloneqq p\rho' \otimes\ketbra{0}+(1-p)\sigma \otimes\ketbra{1} \text{ and } \\
    \tau_{n,k} &\coloneqq 2^{k-n} \frac{\id^{\otimes k}}{2^k} \otimes \ketbra{0}+(1-2^{k-n})\sigma \otimes \ketbra{1},
\end{align}
where $\sigma$ is the CSS state prepared after a failed run.

We now prove the following lemma, which enables us to simplify the expression of our constraint functions $\Delta D_\alpha$. To this end, we will find it useful to first introduce the general mean $Q_\alpha(\bmw || \bmr)$ on a quasiprobabilty distribution ${\bmw \coloneqq (w_1, \dots w_N)^T}$ and probability distribution ${\bmr \coloneqq (r_1, \dots r_N)^T}$, 
\begin{align}
Q_\alpha(\bmw || \bmr)\coloneqq 2^{(\alpha -1) D_\alpha(\bmw || \bmr)} = \sum_{i=1}^N w_i^\alpha r_i^{1-\alpha}.
\end{align}

   \begin{lemma}\label{lemma:general_mean} 
   	Consider the following pairs of $k$-rebit states $(\rho_0, \rho_1)$ and $(\tau_0,\tau_1) \in \mathrm{Int}(\D_{css})$. Moreover let $\psi_0$ and $\psi_1$ be two distinct computational basis states on $m$ rebits, and let  $\{p_i\}_{i=0}^1$, $\{q_i\}_{i=0}^1$ be valid probability distributions. We then have the identity
    \begin{align}
    Q_\alpha(  W_{\sum_{i} p_i \rho_i \otimes \psi_i} ||  &W_{\sum_i q_i \tau_i \otimes \psi_i} ) = \notag \\
    &\sum_{i \in \{ 0 , 1\} }p_i^\alpha q_i^{1-\alpha} Q_\alpha(W_{\rho_i} || W_{\tau_i} ) ,
    \label{eq:Q_alpha_equality}
\end{align}
which in turn implies the inequality 
    \begin{align}
    Q_\alpha(  W_{\sum_{i} p_i \rho_i \otimes \psi_i} ||  &W_{\sum_i q_i \tau_i \otimes \psi_i} )\ge \notag \\
    &p_i^\alpha q_i^{1-\alpha} Q_\alpha(W_{\rho_i} || W_{\tau_i} ) ,  
    \label{eq:Q_alpha_inequality}
\end{align}
for each $i \in \{0,1\}$.
\end{lemma}

\begin{proof}
Since $\psi_0$ and $\psi_1$ are orthogonal, by \eqref{eq:A_orthogonal} this implies
\begin{equation}
\label{eq:overlappp}
\tr[\psi_0 \psi_1] = 2^m \sum_{\bmu \in \P_m} W_{\psi_0}(\bmu)W_{\psi_1}(\bmu) =0.
\end{equation}

As $\psi_i \in \D_{css}$ for each $i\in \{0,1\}$, we must have $W_{\psi_i}(\bmu) \ge 0$. We thus conclude from \eqref{eq:overlappp} that 
\begin{align}
    V_0 &\coloneqq \supp(W_{\psi_0}) \subseteq \ker(W_{\psi_1}), \\
    V_1 &\coloneqq \supp(W_{\psi_1}) \subseteq \ker(W_{\psi_0}).
\end{align}

With this in hand, we can explicitly evaluate:
 \begin{align}
      &Q_\alpha(  W_{\sum_{i} p_i \rho_i \otimes \psi_i} ||  W_{\sum_i q_i \tau_i \otimes \psi_i} )= \notag \\
    & \sum_{i \in \{0,1\}}\sum_{\bmu \in \P_k} \sum_{\bmv \in V_i} \left( p_i W_{\rho_i}(\bmu) W_{\psi_i}(\bmv)  \right)^\alpha \left( q_i W_{\tau_i}(\bmu) W_{\psi_i}(\bmv)  \right)^{1-\alpha}   \notag \\
    &= \sum_{i \in \{0,1\} } p_i^\alpha q_i^{1-\alpha} \sum_{\bmu \in \P_k} W_{\rho_i}(\bmu)^\alpha W_{\tau_i}(\bmu)^{1-\alpha} \notag \\
    &=\sum_{i \in \{0,1\} } p_i^\alpha q_i^{1-\alpha} Q_\alpha(W_{\rho_i} || W_{\tau_i}),
\end{align}
where in the second equality we used the normalisation of our chosen representation $W$. The inequality in the Lemma statement then follows from the fact that both terms on the right hand side of \eqref{eq:Q_alpha_equality} must be non-negative for all $\alpha \in \A$.\end{proof}

With this property in hand, we obtain the following Lemma, which makes the non-trivial $n$-dependence in $\Delta D_\alpha$ more explicit.
\begin{widetext}
\begin{lemma}
 \label{lemma:f__n_expanded}
Let $\mu_k$ denote the maximally mixed state on $k$ qubits. The function $\Delta D_\alpha$ can then be expressed as 
\begin{align}
   \Delta D_\alpha =  n \left(1- H_\alpha[W_\rho]\right)  +k - \frac{1}{\alpha -1}\log \left[ p^\alpha Q_\alpha\left(W_{\rho'}||W_{\mu_k}\right)+(1-p)^\alpha \left(\frac{1}{2^{n-k} -1}\right)^{\alpha -1} \right],
 \label{eq:f__n_expanded}
\end{align}
\end{lemma}
\begin{proof}
By \lemref{lemma:general_mean}, we have the expansion
\begin{align}
  Q_\alpha \left(W_{\rho_p} \, || \, W_{\tau_{n,k}}\right) 
  &=p^\alpha \left( 2^{k-n}\right)^{1-\alpha} Q_\alpha\left(W_{\rho' } \, || \, W_{\mu_k}\right)+ (1-p)^\alpha  \left(1 - 2^{k-n}\right)^{1-\alpha} Q_\alpha\left(W_{\sigma } \, || \, W_{\sigma }\right) \notag \\
  &= \left(2^{n-k}\right)^{\alpha -1} \left[ p^\alpha Q_\alpha\left(W_{\rho'}|| W_{\mu_k} \right)+(1-p)^\alpha \left(\frac{1}{2^{n-k} -1}\right)^{\alpha -1} \right],
\end{align}
where in the last equality we have used $Q_\alpha(\bmp||\bmp) = 1$ for all probability distributions $\bmp$. Therefore,
\begin{align}
    D_\alpha\left(W_{\rho_p} \, || \, W_{\tau_{n,k}}\right) = n-k + \frac{1}{\alpha -1}\log \left[ p^\alpha Q_\alpha\left(W_{\rho'}|| W_{\mu_k}\right)+(1-p)^\alpha \left(\frac{1}{2^{n-k} -1}\right)^{\alpha -1} \right].
    \label{eq:DDD}
\end{align}
Substituting \eqref{eq:DDD} and $D_\alpha\left(W_\rho \Big|\Big| W_{\frac{\id}{2}}\right)= 2 - H_\alpha[W_\rho]$ into \eqref{eq:Delta_alpha_n} gives the result as claimed. \end{proof}
\end{widetext}

\subsection{Constraints are independent of the choice of CSS state $\sigma$ prepared on failed code projection}
\label{appx:bounds_sigma_indep}

By inspection, the form for $\Delta D_\alpha$ given in \lemref{lemma:f__n_expanded} has no $\sigma$-dependence, which implies 
\begin{corollary}
The entropic constraints $\Delta D_\alpha \ge 0$ on $n$-to-$k$ CSS code projection protocols are independent of which CSS state $\sigma$ is prepared following failed runs. 
\end{corollary}

This result also follows from resource-theoretic arguments. Consider a CSS circuit $\E$ on $k+1$ qubits that performs a $Z$-basis measurement on the last qubit and re-prepares the first $k$ qubits in a CSS state $\omega$ conditioned upon the $-1$ outcome. Thus one can write
\begin{align}
  \E_\omega(\cdot) \coloneqq \I \otimes P_0 (\cdot) + \omega \tr \otimes P_1 (\cdot),
\end{align}
where $P_k (\cdot) \coloneqq \ketbra{k} (\cdot) \ketbra{k}$, from which one straightforwardly verifies that $\E_\omega(\rho_p)$ and $\E_\omega(\tau_{n,k})$ simply replaces $\sigma$ in $\rho_p$ and $\tau_{n,k}$ respectively by $\omega$. Since $\E$ is stochastically represented, given any $\omega$ in the interior of $\D_{css}$, we have by Property \ref{property:D_dataprocessing} in \lemref{lemma:nick_restatement} that
\begin{align}
    D_\alpha\left(W_{\rho_p} || W_{\tau_{n,k}}\right) &\ge  D_\alpha\left(W_{\E_\omega[\rho_p]} || W_{\E_\omega[\tau_{n,k}]}\right)  \notag \\ 
    D_\alpha\left(W_{\rho_p} || W_{\tau_{n,k}}\right)&\ge D_\alpha\left(W_{\E_\sigma \circ\E_\omega[\rho_p]} || W_{\E_\sigma \circ \E_\omega [\tau_{n,k}]}\right) \notag \\ 
    &=   D_\alpha\left(W_{\rho_p} ||  W_{\tau_{n,k}}\right).
\end{align}
We therefore conclude that
\begin{align}
    D_\alpha\left(\E_\omega(W_{\rho_p}) || \E_\omega (W_{\tau_{n,k}})\right) = D_\alpha\left(W_{\rho_p} || W_{\tau_{n,k}}\right),
\end{align}
so the entropic constraints on CSS code projections are unaffected by varying $\sigma$ in the interior of $\D_{css}$.

\subsection{Proof of \lemref{lemma:properties_fn}}
\label{appx:fn_properties_proof}

We now prove the following properties of the constraint function $\Delta D_\alpha$ from the main text:
\fnProperties*
\begin{proof} 
	Let us denote the maximally mixed state on $k$ qubits by $\mu_k$. We further simplify notation by defining the constants $c_1 \coloneqq p^\alpha Q_\alpha(W_{\rho'} || W_{\mu_k}) $ and ${c_2 \coloneqq (1-p)^\alpha}$.

\textbf{\textit{Proof of \ref{fproperty_concavity}:}} Let us define the following function
\begin{align}
    g(n) \coloneqq \left[ c_1+c_2 \left(\frac{1}{2^{n-k} -1}\right)^{\alpha -1} \right].
\end{align}
This means that from \lemref{lemma:f__n_expanded} we can write
\begin{equation}
   \Delta D_\alpha =  n \left(1- H_\alpha[W_{\rho}]\right)  +k- \frac{1}{\alpha -1} \log g(n),
\end{equation}
and since the first term is linear we need only check the second derivative of the second term to establish that $\Delta D_\alpha$ is concave. We have
\begin{widetext}
\begin{align}
    &\partial_n^2 \Delta D_\alpha = -\frac{1}{\alpha-1} \partial_n^2  \log g(n) = \\
    & -\frac{  \ln 2 \, c_2 2^{k+n} \left(c_1 \left(2^k+(\alpha -1) 2^n\right) \left(2^{n-k}-1\right)^{\alpha }+c_2 \left(2^n-2^k\right)\right)}{\left(2^n-2^k\right) \left(c_1 2^k \left(2^{n-k}-1\right)^{\alpha }+c_2 \left(2^n-2^k\right)\right)^2}. \nonumber
\end{align}
\end{widetext}

Since $c_1,c_2 \ge 0$ for all $\rho'$ and $p$, the term in square brackets is non-negative for all $n > k, \alpha>1, \rho'$ and $p$ (strictly positive for $p<1$), which implies $\partial_n^2 \Delta D_\alpha$ is non-positive everywhere on our restricted domain. Therefore $\Delta D_\alpha$ is concave, as claimed.

\textbf{\textit{Proof of \ref{fproperty_lim_one}:}}
Recalling that $\alpha > 1$, we have from \lemref{lemma:f__n_expanded} that
\begin{align}
\lim_{n\rightarrow k^+}\Delta D_\alpha  
&= - k H_\alpha[W_{\rho'}] - \notag \\ \frac{1}{\alpha -1}&\lim_{n\rightarrow k^+} \left\{\log \left[c_1+c_2 \left(\frac{1}{2^{n-k}-1}\right)^{\alpha -1} \right]\right\} \notag \\ &= - \infty<0,
\end{align}
so long as $c_2 > 1$, which is true if and only if $p < 1$. 

\textbf{\textit{Proof of \ref{fproperty_lim_infty}:}}
We have:
\begin{align}
\lim_{n\rightarrow \infty}\Delta D_\alpha 
&= k - \frac{1}{\alpha -1} \log [c_1] + \lim_{n \rightarrow \infty}\left\{n \left(1- H_\alpha[W_\rho]\right) \right\} \notag \\
&=H_\alpha[W_{p \rho'}] -k + \lim_{n \rightarrow \infty}\left\{n \left(1- H_\alpha[W_{\rho}]\right) \right\} \notag \\
&= \begin{cases} -\infty , & H_\alpha[W_\rho]> 1, \\
+ \infty , & H_\alpha[W_\rho]< 1, \\
H_\alpha[W_{p \rho'}] -k , & \mathrm{ otherwise}. \end{cases}
\end{align}

Therefore, if $H_\alpha[W_\rho]> 1$ then $\lim_{n\rightarrow  \infty}\Delta D_\alpha<0$, as claimed. 

This completes the proof.\end{proof}

\subsection{Analytic bounds on code length in qudit code projection protocols}
\label{appx:extension_n_m}

In this section, we consider $[[n,k]]$ stabilizer code projections for quantum systems of \emph{arbitrary} Hilbert space dimension $d$, which are stochastic under some generalised Gross's Wigner representation.

Following \eqref{eq:CSS_code_reduction_channel}, we can define the following trace-preserving projection $\E$ on a $[[n,k]]$ stabilizer code $\C$ for $d$-dimensional quantum systems,
\begin{align}
    \E(\cdot) \coloneqq \tr_{k+1,\dots,n}\left[\U \circ \P(\cdot)\right] \otimes &\ketbra{0} + \notag\\
     &\tr[\overline{\P} (\cdot)] \sigma \otimes \ketbra{1},
\end{align}
where $\U$ and $\P$ are respectively the decoding channel and codespace projection for $\C$, $\overline{\P}$ is the projection onto the orthogonal complement of $\C$, and $\sigma$ is positively represented under some generalised Gross's Wigner representation of our choice. This channel transforms $n$ copies of an input noisy magic state $\rho$ as 
\begin{align}
   \E \left[ \rho^{\otimes n } \right] = p \rho' \otimes \ketbra{0} + (1-p) \sigma \otimes \ketbra{1} \eqqcolon \rho_p,
\end{align}
where we assume the output state $\rho'$ following successful code projection is $\delta$-close to $k$ copies of our pure target magic state $\psi$, as measured by the trace distance $\norm{\rho - \sigma}_1$ where $\norm{X}_1 \coloneqq \tr \abs{X} = \tr\sqrt{X^\dagger X}$ is the trace norm (also known as the Schatten-1 norm). Formally, we assume 
\begin{equation}
\norm{\rho' - \psi^{\otimes k}}_1 \le \delta< \norm{\rho - \psi}_1.
\end{equation}
We also define the Frobenius norm (also known as the Schatten-2 norm) as
\begin{align*}
    \norm{X}_2 \coloneqq \sqrt{\tr[X^\dagger X]}.
\end{align*}
We further define the $\ell_1$- and $\ell_2$-norms, respectively, of a vector $\bmw \in \mathbb{R}^d$ as
\begin{align}
    \norm{\bmw}_1 &\coloneqq \sum_{i=1}^d \abs{w_i} , \\
    \norm{\bmw}_2 &\coloneqq \left[\sum_{i=1}^d \abs{w_i}^2\right]^{\frac{1}{2}}.
\end{align}

We now make use of the following result from the literature on real vector spaces (e.g. see ~\cite{bhatia2013matrix}), which is a consequence of the Cauchy-Schwarz inequality.
\begin{lemma} \label{lemma:L1_upper_bound}
    For all $\bmw \in \mathbb{R}^{d}$ 
    \begin{align}
     \norm{\bmw}_1 \le \sqrt{d} \norm{\bmw}_2.
    \end{align}
\end{lemma}
This result enables us to show that vanishingly small variations in quantum states correspond to vanishingly small variations in their Wigner representations: 
\begin{lemma} \label{lem:trace_distance_wigner}
 If $ \norm{\rho - \sigma}_1 \le \epsilon$ then for any generalised Gross's Wigner representation $W$, we have
\begin{align} \label{eq:trace_distance_general}
     \norm{W_\rho - W_\sigma}_1 \le \sqrt{d} \epsilon.
\end{align}
\end{lemma}
\begin{proof}

To simplify notation, we first define the state difference $\Delta \coloneqq \rho - \sigma$ such that $W_\Delta = W_\rho - W_\sigma$. Since the Schatten-$p$ norms are non-increasing with respect to $p$~\cite{bhatia2013matrix}, we obtain 
\begin{align}
    \norm{\rho - \sigma}_1 &\ge \norm{\rho - \sigma}_2 =  \norm{\sum_\bmx W_\Delta(\bmx) A_\bmx}_2\nonumber\\
    & = \sqrt{ \sum_{\bmx,\bmy} W_\Delta^*(\bmx) W_\Delta(\bmy) \tr [A_\bmx^\dagger A_\bmy]} \nonumber\\
    &=\sqrt{ \sum_{\bmx,\bmy} W_\Delta^*(\bmx) W_\Delta(\bmy) d \delta_{\bmx,\bmy}}\nonumber\\
    & = \sqrt{d} \norm{W_\Delta}_2 \ge \frac{1}{\sqrt{d}} \norm{W_\Delta}_1, 
\end{align}
where in the second inequality we employ \lemref{lemma:L1_upper_bound}.
Therefore $\norm{W_\rho - W_\sigma}_1 \le \sqrt{d} \norm{\rho - \sigma}_1$.
\end{proof}

\begin{lemma} \label{lem:continuity_renyi_entropy}
Let $\rho$ and $\sigma$ be two quantum states of a $d$-dimensional qudit such that $\norm{\rho - \sigma}_1 \le \epsilon$. Then given any generalised Gross's Wigner representation $W$, 
\begin{align}
 &\abs{  H_\alpha[W_\rho] - H_\alpha [W_\sigma] } \le \frac{\alpha}{\alpha -1} \log [1 + \epsilon d^{\frac{5}{2}}].
\end{align}
\end{lemma}
\begin{proof}
Theorem 7 (2) of Ref.~\cite{woods2019resource} applies to quasiprobability distributions and tells us that for two $d^2$-dimensional distributions $\bmw,\bmw'$, we have the following continuity statement on the $\alpha$-R\'{e}nyi entropies
\begin{align}
    \abs{H_\alpha(\bmw) - H_\alpha(\bmw')} &\le \notag \\ \frac{\alpha}{\alpha-1} &\log [1 + \norm{\bmw -\bmw'}_1 d^2 ].
\end{align}
The proof of this just relies on the monotonicity of the $p$-norms
$\norm{\bmw}_p \coloneqq \left( \sum_{i=1}^{d^2} \abs{w_i}^p \right)^{1/p}$, i.e., for $1\le \alpha<\beta\le \infty$,
$\norm{\bmw}_{\alpha} \ge \norm{\bmw}_\beta$, which also holds for quasidistributions. The Lemma result then follows from \lemref{lem:trace_distance_wigner}.
\end{proof}

We are now in a position to prove \thmref{thrm:upper_bound_n_m}, which we restate here for clarity:

\UpperBoundnk*
\begin{proof}
Let $\mu_k$ dneote the maximally mixed state of $k$ qudits with Hilbert space dimension $d$. Following the same proof strategy as that of \lemref{lemma:general_mean}, we obtain
\begin{align}
    Q_\alpha(W_{\rho_p} || W_{\tau_{n,k}}) \ge \frac{p^\alpha}{d^{(n-k)(1-\alpha)}}  Q_{\alpha}( W_{\rho'} || W_{\mu_k}).
\end{align}
Since $\log(\cdot)/(\alpha-1)$ is a monotonically increasing function for $\alpha >1$, it therefore follows that
\begin{align}
    &D_\alpha(W_{\rho_p} || W_{\tau_{n,k}}) \ge \notag \\  &\frac{1}{\alpha -1 } \log \left[\frac{p^\alpha}{d^{(n-k)(1-\alpha)}} Q_{\alpha}( W_{\rho'}  || W_{\mu_k}) \right] \notag  \\
    &= D_\alpha (W_{\rho'} || W_{\mu_k }  ) + \frac{\alpha}{\alpha -1} \log p + (n-k)\log d \notag \\
    &= k \log d -H_\alpha(W_{\rho'})+ \frac{\alpha}{\alpha -1} \log p + n\log d ,
\end{align}
where in the final equality we have used the identity $D_\alpha(W_\rho || W_{\frac{\id_d}{d}}) = 2 \log d - H_\alpha(W_\rho) $. 

We can now make use of the continuity of the R\'{e}nyi entropy as stated in \lemref{lem:continuity_renyi_entropy} to further lower-bound this divergence as
\begin{align}
D_\alpha(W_{\rho_p} || &W_{\tau_{n,k}}) \ge k \left[\log d -  H_\alpha (W_{\psi}) \right]  \nonumber\\
&- \frac{\alpha}{1 -\alpha}\log \frac{p}{1+ \delta d^{\frac{5}{2}}} + n \log d,
\end{align}
where the equality follows from the factorisation of the $\alpha$-R\'{e}nyi entropy over subsystems. This gives rise to the following upper bound on the relative entropy difference $\Delta D_\alpha \coloneqq n D_\alpha(W_\rho || W_{\frac{\id}{d}})   -D_\alpha(W_{\rho_p} || W_{\tau_{n,k}})$,
\begin{align}
    0 \le \Delta D_\alpha  &\le     n[  \log d- H_\alpha (W_\rho)] + k \left[ H_\alpha (W_{\psi}) -\log d \right] \nonumber\\ 
    &\hspace{20pt}+ \frac{\alpha}{1 -\alpha}\log \frac{p}{1+ \delta d^{\frac{5}{2}}}.
\end{align}
This gives a weaker but still necessary constraint on stochastic transformations accomplishing $\rho^{\otimes n } \mapsto \rho_p$ and $\left(\frac{\id}{d}\right)^{\otimes n} \mapsto \tau_{n,k}$, which we can rearrange as
\begin{align}
    n[H_\alpha (W_\rho)-\log d ] &\le k \left[ H_\alpha (W_{\psi}) - \log d  \right] \notag \\ &+\frac{\alpha}{1 -\alpha}\log \frac{p}{1+ \delta d^{\frac{5}{2}}}. 
\label{apeq:upperbound_unarranged}
\end{align}

For $H_\alpha(W_\rho) < \log d$, we can rearrange \eqref{apeq:upperbound_unarranged} to obtain the lower bound in \eqref{eq:analytic_lower}. For $H_\alpha(W_\rho) > \log d$, we obtain the upper bound in \eqref{eq:analytic_upper}.
\end{proof}

\section{\uppercase{Decomposition of CSS magic distillation into code projections}}\label{sec:campbell_browne_reduction}
The goal of this appendix is to prove \thmref{theorem:CSS_campbell_browne}, which generalises Theorem 1 of Ref.~\cite{campbell_browne} to an arbitrary number of output qubits in the CSS setting.  

\begin{theorem} \label{theorem:CSS_campbell_browne}
    An \emph{$n$-to-$k$ CSS magic distillation protocol} is any CSS circuit that takes in $n \ge 2$ qubits and outputs onto the first $1 \le k < n$ qubits. Every $n$-to-$k$ CSS magic distillation protocol $\E$ can be decomposed as a sum of CSS code projections, followed by preparing CSS states and completely CSS-preserving post-processing. Thus one can write
	\begin{align}
		\E(\rho) = \sum_j p_j \E_j,\  \E_j \coloneqq \U_j \circ \left(\K_j(\rho) \otimes \ketbra{\varphi_j} \right), 
	\end{align}
	where $p_j$ is a probability, $\U_j$ is a completely CSS-preserving unitary channel on $k$ qubits, $\K_j$ is the codespace projection of an $[[n,k_j]]$ CSS code for some integer $k_j$ in the range $0 \le k_j \le k$, and $\ket{\varphi_j}$ is a CSS state on $(k-k_j)$ qubits. \label{theorem:CSS_Campbell_Browne}
\end{theorem}
\begin{proof}
	To bring $\E$ into the desired form, we proceed in four steps (all auxiliary lemmas used will be presented after this proof):
	\begin{enumerate}
	\item By \lemref{lemma:CSS_magic_distillation_useful_form}, any $n$-to-$k$ CSS magic distillation protocol $\E$  can be decomposed as a sum of $n$-qubit operations 
		\begin{align}
		\E(\rho) = \sum_i q_i \E_i(\rho)
		\end{align}
		where $q_i$ is a probability and
		\begin{align} 
			\E_i(\rho) \coloneqq  \tr_{k+1,\dots,n+m}\left[K_i(\rho \otimes \ketbra{\psi_i}) K^\dagger_i\right]
		\end{align}
		for a CSS state $\ket{\psi_i}$ on $m$ ancillary qubits and Kraus operator $K_i$. This Kraus operator has the form
		\begin{align}
			K_i = U_i \left(\prod_{l=1}^N P(S_{i,l})\right),
		\end{align}
		where $U_i$ is a completely CSS-preserving unitary and $P(S_{i,l})$ is the projection onto the $+1$ eigenspace of a CSS observable $S_{i,l}$.

    \item By \lemref{lemma:glue}, each operation $\E_i$ can be decomposed into a sum 
    \begin{align}
        \E_i(\rho) =\sum_{\bms \in \{0,1\}^{n+m-k}} \E_{i,\bms}(\rho),
    \end{align}
    where each operation $\E_{i,\bms}$ first introduces $m$ ancillary qubits in the CSS state $\ket{\psi_i}$, then post-selects the $+1$ outcome in a sequence of projective measurements of CSS observables $S_{i,N}, \dots, S_{i,1}$, and then performs a CSS code projection on the input and ancillary qubits. Thus one can write (note that $\K$ depends on $i$ and $\bms$)
    \begin{align}
    \E_{i,\bms}(\rho) = \K \circ \P(S_{i,1}) \circ \dots \circ \P(S_{i,N})[\rho \otimes \ketbra{\psi_i}],
    \end{align}
    in which $\P(S_{i,j})$ post-selects the $+1$ outcome in a projective measurement of the CSS observable $S_{i,j}$, and $\K$ is the code projection of an $[[n+m,k]]$ CSS code.
    
    \item By repeated applications of \lemref{lemma:measure_then_decode}, we find that $\E_{i,\bms}$ performs a CSS code projection on the input and ancillary qubits, followed by preparing a CSS state and completely CSS-preserving post-processing. Thus one can write
    \begin{align}
    \E_{i,\bms}(\rho) = q'\ \U' \circ (\K'(\rho \otimes \ketbra{\psi_i})\otimes \ketbra{\varphi'}), \label{eq:local_1}
    \end{align}
    where $q'$ is a probability, $\U'$ is a completely CSS-preserving unitary channel on $k$ qubits, $\ket{\varphi'}$ is a CSS state on $k-k'$ qubits for some integer $k'$ in the range $0 \le k' \le k$, and $\K'$ is a code projection for an $[[n+m,k']]$ CSS code.
    
    \item By \lemref{lemma:remove_ancillas}, each CSS code projection $\K'$ on the input and ancillary qubits from \eqref{eq:local_1} can be reduced to a CSS code projection on the input qubits \emph{alone}, followed by preparing a CSS state and completely CSS-preserving post-processing. Thus one can write
    \begin{align}
    \K'(\rho \otimes \ketbra{\psi_i}) = q''\ \U'' \circ (\K''(\rho) \otimes \ketbra{\varphi''}),
    \end{align}
    where $q''$ is a probability, $\U''$ is a completely CSS-preserving unitary channel on $k'$ qubits, $\K''$ is the code projection for an $[[n,k'']]$ CSS code for some integer $k''$ in the range $0 \le k'' \le k'$, and $\ket{\varphi''}$ is a CSS state on $k'-k''$ qubits. Substituting back then immediately yields the result. \qedhere
\end{enumerate}
\end{proof}

\subsection*{Auxiliary lemmas}
\label{sec:auxiliary}
Before turning to the proofs of Lemmas used in each step of the main proof, we first present a result that will be useful throughout. 
	\begin{lemma}\label{lemma:CSS_unitary_observable_conjugate}
		Given any completely CSS-preserving unitary $U$ and CSS observable $S$ on $n$ qubits, ${S' \coloneqq U^\dagger S U}$ is another CSS observable of the same type as $S$. This further implies ${P(\pm S) U = U P(\pm S')}$.
	\end{lemma}
\begin{proof}
For convenience, let us label the $n$ qubits as $1,\dots,n$. Let $\bma$ be an arbitrary $n$-bit string, and $\bme_j$ be the $n$-bit string with 1 in its $j$th entry and 0 everywhere else. We then have the following conjugation relations 
\begin{align}
&Z_j [X(\bma)] Z_j = (-1)^{a_j} X(\bma)\\
&Z_j [Z(\bma)] Z_j = Z(\bma) \\
&X_j [X(\bma)] X_j = X(\bma)\\
& X_j [Z(\bma)] X_j = X_j (-1)^{a_j} Z(\bma)\\
&\CNOT(i,j) [X(\bma)] \CNOT(i,j) =  X(\bma + a_i \bme_j)\\
&\CNOT(i,j) [Z(\bma)] \CNOT(i,j) =  Z(\bma + a_j \bme_i),  
\end{align}
for any $i, j$ from the range $1,\dots,n$, where arithmetic is modulo 2. Since $S$ is of the form $\pm X(\bma)$ or $\pm Z(\bma)$, and $U$ is a product of $\CNOT(i,j), Z_j$ and $X_j$ for $i,j$ in the range $1,\dots,n$, we immediately arrive at the Lemma result.
\end{proof}

It is furthermore useful to recall that, given an $[[n,k]]$ CSS code $\C$ whose stabilizer group is generated by CSS observables $S_1,\dots,S_{n-k}$, Corollary \ref{corollary:completely_CSS_decoding} implies the code projection $\K$ of $\C$ can be represented using a completely CSS-preserving encoding unitary $U$ as  
\begin{align}
	&\K(\rho) \coloneqq \tr_{k+1,\dots,n}\left[U^\dagger P(\rho)P U\right], \nonumber\\
	&\text{ where }U^\dagger S_i U = \begin{cases} 
		Z_{k+i} \text{ if $S_i$ is $Z$-type}\\
		X_{k+i} \text{ if $S_i$ is $X$-type}
	\end{cases}  \nonumber\\
	&\text{ and } P = \prod_{i=1}^{n-k} P(S_i). \label{eq:generic_code_projection} 
\end{align} 
We remark that in the $k=0$ case, $\K$ simply projects onto a pure $n$-qubit CSS state and then discards it.

\subsubsection*{Step 1: Standard form for CSS magic distillation protocols}
	\begin{lemma}\label{lemma:CSS_magic_distillation_useful_form} 
		Any $n$-to-$k$ CSS magic distillation protocol $\E$ can be decomposed as a sum of $n$-qubit operations 
		\begin{align}
		\E(\rho) = \sum_i p_i \E_i(\rho),
		\end{align}
		where $p_i$ is a probability and
		\begin{align}
		 \E_i(\rho) \coloneqq  \tr_{k+1,\dots,n+m}\left[K_i(\rho \otimes \ketbra{\psi_i}) K^\dagger_i\right],
		\end{align}
		for a CSS state $\ket{\psi_i}$ on $m$ ancillary qubits and Kraus operator $K_i$. Furthermore, $K_i$ has the form
		\begin{align}
			 K_i = U_i \left(\prod_{l=1}^N P(S_{i,l})\right),
		\end{align}
		 where $U_i$ is a completely CSS-preserving unitary and $P(S_{i,l})$ projects onto the $+1$ eigenspace of a CSS observable $S_{i,l}$. 
	\end{lemma}
\begin{proof}
    By \lemref{lemma:CSS_operation_canonical_form}, we can decomppose $\E$ as follows
    \begin{align}
	\E(\rho) = \sum_i q_i \E_i(\rho),
	\end{align}
	 where $q_i$ is a probability and
	 \begin{align}
	 \E_i(\rho) = \tr_{k+1,\dots,n+m}\left[K_i(\rho\otimes \sigma_i)K_i^\dagger\right]
	\end{align}
    for a CSS state $\sigma_i$ on $m$ ancillary qubits and Kraus operator $K_i$. Furthermore, $K_i$ has the form
    \begin{align}
    	K_i = \prod_{l=1}^N P(S_{i,l}) U_{i,l},
    \end{align}
    where $U_{i,l}$ is a completely CSS-preserving unitary and $P(S_{i,l})$ projects onto the $+1$ eigenspace of the CSS observable $S_{(i,j),l}$. We then conjugate every completely CSS-preserving unitary $U_{(i,j),l}$ to the beginning of its Kraus operator $K_{i,j}$ as shown in \lemref{lemma:CSS_unitary_observable_conjugate}, where we compose them into a single completely CSS-preserving unitary $U_i$. Decomposing each CSS state $\sigma_i$ on ancillary qubits into a statistical mixture of pure CSS states then yields the Lemma result. \qedhere 
\end{proof}

\subsubsection*{Step 2: Exposing the decoding}
\begin{lemma}\label{lemma:glue}
	Consider a channel $\E$ on $n \ge 2$ qubits that performs a completely CSS-preserving unitary $U$ and then discards the final $n-k$ qubits where $1 \le k < n$, 
	\begin{align}
		\E(\rho) = \tr_{k+1, \dots, n}\left[U^\dagger (\rho)U \right].
	\end{align}
	Then $\E$ can be decomposed into a sum over CSS code projections $\C_\bms$ indexed by the computational basis $\{\bms\}$ on the discarded qubits,
	\begin{align}
		\E(\rho) = \sum_{\bms \in \{0,1\}^{n-k}} \K_\bms(\rho).
	\end{align} 
\end{lemma}
\begin{proof}
	The channel $\E$ is unchanged by performing a measurement in the computational basis $\{\ket{\bms}\}$ of the final $n-k$ qubits before discarding them. Thus one can write
	\begin{align}
		\E(\rho) = \sum_{\bms \in \{0,1\}^{n-k}} \K_\bms(\rho),
	\end{align}
	where we have defined
	\begin{align} 
	\K_\bms(\rho) \coloneqq \tr_{k+1, \dots, n}\left[\id^{\otimes k} \otimes \ketbra{\bms} (U^\dagger(\rho)U) \id^{\otimes k} \otimes \ketbra{\bms}\right].
	\end{align}
  Let $U_\bms$ be a completely CSS-preserving unitary on the final $n-k$ qubits defined as ${U^\dagger_\bms  \coloneqq \bigotimes_{i=1}^{n-k} (X)^{s_i}}$. By the cylic property of the trace, we obtain
	\begin{align}
		\K_\bms(\rho) &= \tr_{k+1, \dots, n}\left[K(\rho)K^\dagger\right],
	\end{align}
	for the Kraus operator
	\begin{align}
		K &\coloneqq \id^{\otimes k} \otimes \left(U^\dagger_\bms\ketbra{\bms}\right) U^\dagger  \\
		& = \id^{\otimes k} \otimes \ketbra{\bm{0}} \left([\id^{\otimes k} \otimes U^\dagger_\bms] U^\dagger\right). 
	\end{align}
	Conjugating $U' \coloneqq U \left(\id^{\otimes k} \otimes U_\bms\right)$ past $\id^{\otimes k} \otimes \ketbra{\bm{0}} = \prod_{i=1}^{n-k} P(Z_{k+i})$ in accordance with \lemref{lemma:CSS_unitary_observable_conjugate}, we see by comparison with \eqref{eq:generic_code_projection} that $\K_\bms$ is a code projection for an $[[n,k]]$ CSS code stabilized by $\langle U' Z_{k+1} U^{'\dagger}, \dots, U' Z_n U^{'\dagger}\rangle$.\qedhere 
\end{proof} 

\subsubsection*{Step 3: Removing the projections}
\begin{lemma}\label{lemma:measure_then_decode}
		Let $\E$ be an operation on $n$ qubits that projectively measures a CSS observable $S$, post-selects the $+1$ outcome, and then carries out a code projection $\K$ of an $[[n,k]]$ CSS code $\C$. Thus one can write
		\begin{align}
			\E(\rho) \coloneqq \K \circ \P(S)(\rho),\ \P(S)(\cdot) \coloneqq P(S)(\cdot)P(S) 
		\end{align}
		
		There are three possibilities for how $\E$ transforms $\rho$:
		\begin{enumerate}
			\item that $\E(\rho) = 0$ for all $\rho$.
			\item that $\E$ is a code projection $\K'$ for another $[[n,k]]$ CSS code $\C'$, followed by completely CSS-preserving unitary post-processing $\tilde{\U}$. Thus one can write
			\begin{align}
			\E(\rho) = p\ \tilde{\U} \circ \K'(\rho),
			\end{align}
			where $p > 0$ is a probability. One can further find logical bases $\left\{\ket{\bms_{\C'}} | \bms \in \{0,1\}^k \right\}$ and $\left\{\ket{\bms_\C}|\bms \in \{0,1\}^k \right\}$, respectively generated by completely CSS-preseving encoding unitaries for the new code $\C'$ and the old code $\C$, such that 
			\begin{align}
				P(S)\ket{\bms_\C} \propto \ket{\bms_{\C'}}.
			\end{align}
			
			\item for $k \ge 1$, that $\E$ is a code projection $\K'$ for an $[[n,k-1]]$ CSS code $\C'$, followed by preparing a CSS state $\ket{\varphi}$ on a single qubit and completely CSS-preserving unitary post-processing $\tilde{\U}$, i.e.
			\begin{align}
				\E(\rho) = \tilde{\U} \circ [\K'(\rho) \otimes \ketbra{\varphi}].
			\end{align}
			Furthermore, we can find logical bases $\left\{\ket{\bms'_{\C'}} | \bms \in \{0,1\}^{k-1} \right\}$ and $\left\{\ket{\bms_\C}|\bms \in \{0,1\}^k \right\}$, respectively generated by completely CSS-preseving encoding unitaries for the new code $\C'$ and the old code $\C$, such that 
			\begin{align}
				P(S) \ket{f(\bms')_\C} \propto \ket{\bms'_{\C'}},
			\end{align}
			for some function $f:\{0,1\}^{k-1} \rightarrow \{0,1\}^k$.
		\end{enumerate} 
	\end{lemma}
\begin{proof}
	Let $S_{k+1},\dots,S_n$ be a set of $n-k$ CSS observables that generate the stabilizer group of $\C$. Beginning with the representation for the code projection $\K$ given by \eqref{eq:generic_code_projection}, we can conjugate the encoding unitary past the codespace projector in accordance with \lemref{lemma:CSS_unitary_observable_conjugate}, and arrive at the alternative representation
	\begin{align}
	\K(\rho) = \tr_{k+1, \dots, n}\left[P U^\dagger(\rho) U P \right], \label{eq:code_projection_decode_first}	
	\end{align}
	where $U$ is a completely CSS-preserving encoding unitary for $\C$ and $P$ is now its no-error syndrome projector. Thus one can write
	\begin{align}
	 &U^\dagger S_i U = C_i \text{ and } P \coloneqq \prod_{i=k+1}^{n} P(C_i), \text{ where }\\
	&C_i \coloneqq \begin{cases}
	Z_i& \text{ if $S_i$ is $Z$-type}\\
	X_i& \text{ if $S_i$ is $X$-type}\\
	\end{cases} \text{ for } i = k+1,\dots,n.
	\end{align}
	Following \eqref{eq:ur_logical_basis_states}, $U$ generates the following logical basis for $\C$:
	\begin{align}
	\forall \bms \in \{0,1\}^k: &\ket{\bms_\C} \coloneqq U \left(\ket{\bms} \bigotimes_{i=k+1}^{n} \ket{\phi_i}\right) \label{eq:generic_logical_basis}\\
	\text{ where }& \ket{\phi_i} \coloneqq \begin{cases}
	\ket{0} &\text{ if $S_i$ is $Z$-type}\\
	\ket{+} &\text{ if $S_i$ is $X$-type.}
	\end{cases} 
	\end{align}
	By expressing $P$ in \eqref{eq:code_projection_decode_first} as $P = \id^{\otimes k} \bigotimes_{i=1}^{n-k} \ketbra{\phi_i}$ and expanding $\id^{\otimes k}$ in the computational basis, we arrive at another alternative form for $\K$,
		\begin{align}
		\K(\rho) = K(\rho)K^\dagger, \text{ where } K^\dagger \coloneqq \sum_{\bms \in \{0,1\}^k} \ketbra{\bms_\C}{\bms}. \label{eq:code_projection_logical_basis_states} 
		\end{align} 
	We now show how $\E$ can be manipulated into one of the forms stated by the Lemma depending on the relationship between $S$ and the CSS observables generating the stabilizer group for $\C$.
	
	\begin{enumerate}
	\item[(i)] \textbf{$S$ does \emph{not} commute with at least one generator $S_{k+1},\dots,S_n$.} We assume without loss of generality that $S$ does not commute with $S_{k+1}$. We now show how $\E$ can be manipulated into the \emph{second} form stated by the Lemma.
	
	Using the form of $\K$ given in \eqref{eq:code_projection_logical_basis_states}, we can express $\E$ as 
	\begin{align}
	\E(\rho) = M(\rho)M^\dagger, \label{eq:projection_initial_logic_states_rep}
	\end{align} 
	where $M$ is the Kraus operator
	\begin{align}
		M^\dagger \coloneqq \sum_{\bms \in \{0,1\}^k} P(S) \ketbra{\bms_\C}{\bms}. 
	\end{align}
	
	The stabilizer group $\S(\ket{\bms}_\C)$ for the logical basis state $\ket{\bms_\C}$ in $\C$ can be generated from the set of $n$ CSS observables $\{ (-1)^{s_1} U Z_1 U^\dagger, \dots, (-1)^{s_k} U Z_k U^\dagger, S_{k+1}, \dots, S_n \}$. We can multiply all other members of this set that do \emph{not} commute with $S$ by $S_{k+1}$ and, as all CSS observables that do not commute with $S$ are of the same type, arrive at a set of $n$ CSS observables $\{S_1,\dots,S_n\}$ such that
	\begin{align}
		\S(\ket{\bms_\C}) = \langle (-1)^{s_1} S_1, \dots, (-1)^{s_k} S_k, S_{k+1},\dots,S_n\rangle.
	\end{align}
	This implies
	\begin{align}
	P(S)\ket{\bms_\C} = \frac{1}{\sqrt{2}} \ket{\psi_\bms} 
	\end{align} 
	where $\ket{\psi_\bms}$ is a CSS state stabilized by 
	\begin{align}
	\S(\ket{\psi_\bms}) = \langle (-1)^{s_1} S_1, \dots, &(-1)^{s_k} S_k, S, S_{k+2} \dots, S_n \rangle.
	\end{align}
	
	Applying \lemref{lemma:CNOT_gaussian_elim} to $S_1,\dots, S_k, S, S_{k+2}, \dots, S_n$, we can find a completely CSS-preserving encoding unitary $U'$ for the $[[n,k]]$ CSS code $\C'$ stabilized by $\langle S, S_{k+2},\dots,S_n\rangle$ that generates a logical basis $\{\ket{\bms}_{\C'}\}$ in $\C'$ such that $\ket{\bms_{\C'}}$ shares the stabilizer group of $\ket{\psi_\bms}$. Therefore, $\ket{\bms_{\C'}}$ and $\ket{\psi_\bms}$ only differ up to a phase, and one can write 
	\begin{align}
	P(S)\ket{\bms_{\C}} = \frac{1}{\sqrt{2}} e^{-i\theta_\bms}\ket{\bms_{\C'}}. \label{eq:logical_basis_states_relationship_non_commuting}
	\end{align}
	Substituting into \eqref{eq:projection_initial_logic_states_rep}, we obtain
	\begin{align}
	\E(\rho) = \frac{1}{2} \tilde{U} \left[\K'(\rho)\right] \tilde{U}^\dagger, \label{eq:E_final_form_non_commuting}
	\end{align}
	where we have defined the following unitary on $k$ qubits to adjust for the phase differences between $\ket{\psi_\bms}$ and $\ket{\bms_{\C'}}$,
	\begin{align}
	\tilde{U} \coloneqq \sum_{\bms \in \{0,1\}^k} e^{i\theta_\bms}\ketbra{\bms},
	\end{align}
	as well as the operation
	\begin{align}
	\K'(\cdot) \coloneqq K'(\cdot)K^{'\dagger},\ K^{'\dagger} \coloneqq \sum_{\bms \in \{0,1\}^k} \ketbra{\bms_\C'}{\bms}.
	\end{align}
	By comparison with\eqref{eq:code_projection_logical_basis_states}, we see that $\K'$ is the CSS code projection of $\C'$. \eqref{eq:logical_basis_states_relationship_non_commuting} and \eqref{eq:E_final_form_non_commuting} now match the second statement of the Lemma provided that $\tilde{U}$ is completely CSS-preserving, which we now proceed to show.
	
	We first express $\E$ using the form of $\K$ given by \eqref{eq:code_projection_decode_first} as
	\begin{align}
		\E(\rho) = \tr_{k+1, \dots, n}\left[P U^\dagger P(S)\rho P(S) U P\right]. \label{eq:E_rho_pre_proj_conj}
	\end{align}
	Let $\ket{\psi}$ be a CSS state on $k$ qubits expressed as $\ket{\psi} = \sum_{\bms \in \{0,1\}^k} c_\bms \ket{\bms}$ in the computational basis. Since $U'$ is completely CSS-preserving, it encodes $\ket{\psi}$ as another CSS state $\ket{\psi_{\C'}} = \sum_\bms c_\bms \ket{\bms_{\C'}}$ in $\C'$. By inputting $\ketbra{\psi_{\C'}}$ into the two forms of $\E$ given by \eqref{eq:E_final_form_non_commuting} and \eqref{eq:E_rho_pre_proj_conj} and equating their outputs, we obtain 
	\begin{align}
	\frac{1}{2} \tilde{U}\ketbra{\psi}\tilde{U}^{\dagger} = \tr_{k+1,\dots,n}\left[P U^\dagger P(S)(\ketbra{\psi_{\C'}})P(S) U P\right]. 
	\end{align}
	We see from the proof of \lemref{lemma:CCSSP_powers} that the right hand side carries out a completely CSS-preserving operation. Therefore, since $\ket{\psi_{\C'}}$ is CSS, $\tilde{U}\ket{\psi}$ must also be CSS. As this argument applies to \emph{any} pure CSS state $\ket{\psi}$ on $k$ qubits, we conclude that $\tilde{U}$ is CSS-preserving. 
	
	From the proof of \lemref{lemma:CCSSP_vs_CSSP}, we can express $\tilde{U}$ as $\tilde{U} = \left[H^{\otimes k}\right]^a V$, where $V$ is a \emph{completely} CSS-preserving unitary and $a$ is a binary digit. While the $\CNOT$-gate and single-qubit $X$- and $Z$-gates map computational basis states onto computational basis states, the collective Hadamard gate does not. Since $\tilde{U}$ is diagonal in the computational basis, we conclude that $a=0$, so $\tilde{U}$ is indeed completely CSS-preserving.
	
	\item[(ii)] \textbf{$S$ commutes with $S_{k+1},\dots, S_n$ and is independent of them.} This is only possible when $k \ge 1$. In this case, we show that $\E$ can be manipulated into the \emph{third} form stated in the Lemma.
	
	By conjugating $P(S)$ past $U$ in \eqref{eq:E_rho_pre_proj_conj} using \lemref{lemma:CSS_unitary_observable_conjugate}, we can express $\E$ as
	\begin{align}
		\E(\rho) = \tr_{k+1, \dots, n}\left[P P(S') U^\dagger(\rho) U P(S') P \right], \label{eq:E_rho_expanded}
 	\end{align}
 	where $S' \coloneqq U^\dagger S U$ is another CSS observable. 
 	
 	Consider the term 
 	\begin{align}
 		P(S') P = P(S') P(C_{k+1}) \dots P(C_n) \label{eq:P_S'_P_expanded}
 	\end{align}
 	from \eqref{eq:E_rho_expanded}. Since $S'$ commutes with $C_{k+1},\dots,C_n$, and $C_i$ is a CSS observable on qubit $i$ alone, $S'$ must be trivial on any qubit $i$ out of the last $n-k$ where $C_i$ is a different type of CSS observable from $S'$. 
 	
 	Given any two Pauli observables $O_2$ and $O_2$, we have that
 	\begin{align}
 		 P(O_1)P(O_2) = P(O_1 O_2)P(O_2). \label{eq:new_set_obs_commute_independent}
 	\end{align}
 	We can therefore apply \eqref{eq:new_set_obs_commute_independent} to $S'$ and the $C_i$ in \eqref{eq:P_S'_P_expanded} to eliminate the part of $P(S')$ that acts on the last $n-k$ qubits, i.e., there exists a CSS observable $S''$ on the first $k$ qubits \emph{alone} such that 
 	\begin{align}
 		P(S') P = [P(S'') \otimes \id^{\otimes (n-k)}] P.  \label{eq:P_S'_reduction}
 	\end{align}
 	Since $S'$ must also be independent of $C_k+1,\dots,C_n$, we conclude that $S''$ is not proportional to the identity.
 	
 	We now consider the cases where $S''$ is $Z$-type and $X$-type separately. To this end, it is first useful to define a completely CSS-preserving $k$-qubit unitary that moves the $j$th qubit in a block of $k$ to the $k$th position,
 	\begin{align}
 	\SWAP(j \rightarrow k) \coloneqq \begin{cases}
 		\prod\limits_{i=j}^k \text{SWAP}(i,i+1)& \text{ if } j < k\\
 		\id^{\otimes k}&  \text{ otherwise,} 
 	\end{cases}
 	\end{align}
 	 where $\text{SWAP}(i,j)$ swaps qubits $i$ and $j$ and is carried out by $\CNOT(i,j) \CNOT(j,i) \CNOT(i,j)$.
	\begin{enumerate}
		\item \textbf{$S''$ is $Z$-type.} We therefore represent $S''$ as $S'' = (-1)^b Z(\bmz)$, where $b$ is a binary digit and $\bmz$ is an $k$-bit string of which at least one bit $j$ is such that $z_j = 1$. 
		
		Consider the following completely CSS-preserving unitary on $k$ qubits
		\begin{align}
			\tilde{U}^\dagger \coloneqq \SWAP(j \rightarrow k) \circ X^b_j \circ U_C,
		\end{align}
		where $U_C$ is defined as
		\begin{align} 
		 U_C \coloneqq \begin{cases} 
		\prod\limits_{\substack{i=1\\z_i = 1, i \neq j}}^k \CNOT(i,j)& \text{ if } \exists i \neq j: z_i = 1,\\
		\id^{\otimes k}& \text{ otherwise.}
		\end{cases}  
		\end{align}
		Because $U_C$ acts as $U_C[Z(\bmz)]U^\dagger_C = Z_j$, we have that 
		\begin{align}
		\tilde{U}^\dagger P(S'') \tilde{U} = P(Z_k), \label{eq:P_S''_Z_rotate}
		\end{align}
		Using \eqref{eq:P_S''_Z_rotate} to substitute $P(S'')$ in \eqref{eq:P_S'_reduction}, we have by \eqref{eq:E_rho_expanded} that 
		\begin{align}
		\E(\rho) = \tilde{U}\left(\tr_{k+1, \dots, n}\left[P' U^{'\dagger} (\rho) U' P' \right] \right) \tilde{U}^\dagger,
		\end{align}
		where we have defined
		\begin{align}
			 P' \coloneqq P(Z_k) P \text{ and } U' \coloneqq U \tilde{U} \otimes \id^{\otimes (n-k)}.
		\end{align}  
		Since $k$th qubit outputted by the part inside $\tilde{U}(\cdot)\tilde{U}^\dagger$ is always $\ket{0}$, we can simply discard the $k$th qubit as well and re-prepare it. Thus one can write
		\begin{align}
		\E(\rho) = \tilde{U}\left(\K'(\rho) \otimes \ketbra{0}\right) \tilde{U}^\dagger,
		\end{align}
		where we have defined 
		\begin{align}
			\K'(\rho) \coloneqq \tr_{k, \dots, n}\left[P' U^{'\dagger} (\rho) U' P' \right].
		\end{align}
		From \eqref{eq:code_projection_decode_first}, we see that $\K'$ is the code projection of an $[[n,k-1]]$ CSS code $\C'$ stabilized by $\langle U' Z_k U^{'\dagger}, U' C_{k+1} U^{'\dagger},\dots, U'C_n U^{'\dagger} \rangle$ with a completely CSS-preserving encoding unitary $U'$. By \eqref{eq:generic_logical_basis}, $U'$ generates the following logical basis for $\C'$:
		\begin{align}
			\forall \bms' \in \{0,1\}^{k-1}: \ket{\bms_\C'} &\coloneqq U'\ket{\bms}\ket{0}\ket{\Phi}\\
			 \text{ where }& \ket{\Phi} \coloneqq \bigotimes_{i=1}^{n-k} \ket{\phi_i}. 
		\end{align}
		We can then relate the logical basis state $\ket{\bms'_{\C'}}$ to a logical basis state in the old code $\C$ as
		\begin{align}
		\ket{\bms'_{\C'}} &= \left[U \tilde{U} \otimes \id^{\otimes (n-k)}\right] P(Z_k) P \ket{\bms'}\ket{0}\ket{\Phi} \notag \\
		& =  U P\left(S''\right)\otimes \id^{\otimes (n-k)}P \left[\left(\tilde{U}\ket{\bms'}\ket{0}\right)\ket{\Phi}\right] \notag \\
		& = U P(S') P \ket{f(\bms')}\ket{\Phi} \notag \\
		& = P(S) U \ket{f(\bms')}\ket{\Phi} \notag \\
		& = P(S)\ket{f(\bms')_\C}, 
		\end{align} 
		where we have defined $\ket{f(\bms')} \coloneqq \tilde{U}\ket{\bms'}\ket{0}$, used \eqref{eq:P_S''_Z_rotate} to obtain the second equality and \eqref{eq:P_S'_reduction} to obtain the third. Explicitly, $f$ is defined by
		\begin{align}
		f(\bms') = &(s'_1,\dots,s'_{j-1}, s, s'_j,\dots, s'_{k-1}), \text {where } \nonumber \\
		&s \coloneqq \left(\sum_{i=1}^{j-1} s'_i z_i\right) + b +\left(\sum_{i=j+1}^{k} s'_{i-1} z_i\right),
		\end{align} 
		and arithmetic is modulo 2.
		
		\item \textbf{$S''$ is $X$-type.} We therefore represent $S''$ as $S'' = (-1)^b X(\bmx)$, where $b$ is a binary digit and $\bmx$ is an $k$-bit string of which at least one bit $j$ is such that $x_j = 1$.
		
		Consider the following completely CSS-preserving unitary on $k$ qubits,
		\begin{align}
		\tilde{U}^\dagger \coloneqq \SWAP_{j \rightarrow k}\circ Z^b_j\circ U_{C},
		\end{align}
		in which $U_C$ is defined as
		\begin{align}
		U_C \coloneqq \begin{cases} 
		\prod\limits_{\substack{i=1\\x_i = 1, i \neq j}}^k \CNOT(j,i)& \text{ if } \exists i \neq j: x_i = 1,\\
		\id^{\otimes k}& \text{ otherwise.}
		\end{cases}
		\end{align}
		Because $U_C [X(\bmx)] U^\dagger_C = X_j$, we have that 
		\begin{align}
		\tilde{U}^\dagger P(S'') \tilde{U} = P(X_k). \label{eq:P_S''_X_rotate}
		\end{align}
		Using \eqref{eq:P_S''_X_rotate} to substitute $P(S'')$ in \eqref{eq:P_S'_reduction}, we have by \eqref{eq:E_rho_expanded} that 
		\begin{align}
		\E(\rho) = \tilde{U}\left(\tr_{k+1, \dots, n}\left[P' U^{'\dagger} (\rho) U' P' \right] \right) \tilde{U}^\dagger,
		\end{align}
		where we have defined
		\begin{align}
			P' \coloneqq P(X_k) P,\ U' \coloneqq U \tilde{U} \otimes \id^{\otimes (n-k)}.
		\end{align}  
		Since $k$th qubit outputted by the part inside $\tilde{U}(\cdot)\tilde{U}^\dagger$ is always $\ket{+}$, we can simply discard the $k$th qubit as well and re-prepare it. Thus one can write
		\begin{align}
		\E(\rho) = \tilde{U}\left(\K'(\rho) \otimes \ketbra{+}\right) \tilde{U}^\dagger,
		\end{align} 
		where we have defined
		\begin{align}
			 \K'(\rho) \coloneqq \tr_{k, \dots, n}\left[P' U^{'\dagger} (\rho) U' P' \right].
		\end{align}
		From \eqref{eq:generic_code_projection}, we see that $\K'$ is a code projection for an $[[n,k-1]]$ CSS code $\C'$ stabilized by $\langle U' X_k U^{'\dagger}, U' C_{k+1} U^{'\dagger}, \dots, U' C_n U^{'\dagger} \rangle$ with a completely CSS-preserving encoding unitary $U'$.  By \eqref{eq:generic_logical_basis}, $U'$ generates the following logical basis for $\C'$:
		\begin{align}
		\forall \bms' \in \{0,1\}^{k-1}: \ket{\bms_\C'} &\coloneqq U'\ket{\bms}\ket{+}\ket{\Phi}\\
		\text{ where } \ket{\Phi} &\coloneqq \bigotimes_{i=1}^{n-k} \ket{\phi_i}. 
		\end{align}
		We can then relate the logical basis state $\ket{\bms'_{\C'}}$ in the new code $\C'$ to a logical basis state in the old code $\C$ as 
		\begin{align}
		\ket{\bms'_{\C'}} &= \left[U\tilde{U} \otimes \id^{(n-k)}\right] P(X_k) P \ket{\bms'}\ket{+}\ket{\Phi} \notag \\
		&= \sqrt{2} \left[ U \tilde{U} \otimes \id^{\otimes (n-k)} \right] P(X_k) P \ket{\bms'}\ket{0}\ket{\Phi} \notag \\
		&= \sqrt{2}\ U P(S'') \otimes \id^{\otimes (n-k)} P \left(\tilde{U}\ket{\bms'}\ket{0}\right)\ket{\Phi} \notag \\
		&= \sqrt{2}\ U P(S') P \ket{f(\bms')}\ket{\Phi} \notag \\
		&= \sqrt{2}\ P(S) U \ket{f(\bms')}\ket{\Phi} \notag \\
		&= \sqrt{2}\ P(S) \ket{f(\bms')_\C},
		\end{align}
		where we defined $\ket{f(\bms')} \coloneqq \tilde{U} \ket{\bms'}\ket{0}$, used \eqref{eq:P_S''_X_rotate} to obtain the second equality and \eqref{eq:P_S'_reduction} to obtain the third. Explicitly, $f{(\bms') = (s'_1, \dots, s'_{j-1}, 0, s'_j, \dots,s'_{k-1})}$. 
	\end{enumerate}
	\item[(iii)] \textbf{$S$ is not independent of $S_{k+1},\dots,S_n$.} This implies either $-S$ or $S$ stabilizes $\C$. In the former case, we see from \eqref{eq:projection_initial_logic_states_rep} that $\E(\rho) = 0$ for all $\rho$. In the latter case, we have $P(S)\ket{\bms_{\C}} = \ket{\bms_{\C}}$, which implies $\E = \K$ by \eqref{eq:projection_initial_logic_states_rep}. Together these equations match the Lemma's second form. \qedhere
	\end{enumerate}
 \end{proof}

\subsubsection*{Step 4. Removing ancillary qubits}
\begin{lemma} \label{lemma:codespace_input_ancilla_split}
	Let $\C$ be an ${[[n+m,k]]}$ CSS codes where $n \ge 1, n > k$ and $m > 0$, and let $\{\ket{\bms_\C}\}$ be a logical basis for $\C$ generated by a completely CSS-preserving encoding unitary $U$ such that each logical basis state $\ket{\bms_\C}$ factorises over $n$ and $m$ qubits, i.e., 
	\begin{align}
		\ket{\bms_\C} = \ket{\psi_\bms} \otimes \ket{\psi}, \label{eq:logical_basis_state_product_decomp}
	\end{align} 
	where $\ket{\psi}$ is a CSS state on $m$ qubits. We then have that 
	\begin{align}
		\ket{\psi_\bms} = e^{-i\theta_\bms}\ket{\bms_{\C'}},
	\end{align} 
	where $\{\ket{\bms_{\C'}}\}$ is a logical basis for an $[[n,k]]$ CSS code $\C'$ generated by a completely CSS-preserving unitary.
\end{lemma}
\begin{proof}
	By \eqref{eq:ur_logical_basis_states}, the stabilizer group for the logical basis state $\ket{\bms_\C}$ can be related to the stabilizer group $\S(\C)$ of $\C$ as 
	\begin{align}
		\S(\ket{\bms_\C}) = \langle (-1)^{s_1} S_1, \dots, (-1)^{s_k} S_k \rangle \times \S(\C), \label{eq:original_stabilizer_group_logical_basis_states}
	\end{align} 
	where we have defined $S_i \coloneqq U Z_i U^\dagger$ for $i=1,\dots,k$.
	
 We observe that $\S(\C)$ is a direct sum of $\mathbb{F}_2$-subspaces $\S_X(\C)$ and $\S_Z(\C)$ corresponding to $X$- and $Z$-type stabilizers for $\C$. Because $\ket{\psi}$ is a CSS state, its stabilizer group $\S(\ket{\psi})$ is similarly a direct sum of $\mathbb{F}_2$-subspaces $\S_Z(\ket{\psi})$ and $\S_X(\ket{\psi})$ corresponding to $X$- and $Z$-type stabilizers for $\ket{\psi}$.
 
 Since $\{\ket{\bms_\C}\}$ span $\C$, \eqref{eq:logical_basis_state_product_decomp} implies that, if $S$ stabilizes $\ket{\psi}$, then $\id^{\otimes n} \otimes S$ stabilizes $\C$. Therefore, $\id^{\otimes n} \otimes \S_Z(\ket{\psi})$ and $\id^{\otimes n} \otimes \S_X(\ket{\psi})$ are $\mathbb{F}_2$-subspaces of $\S_Z(\C)$ and $\S_X(\C)$ respectively. By applying the basis extension theorem separately to the $Z$ and $X$ cases, we can represent $\S(\C)$ as
	\begin{align}
		\S(\C) = \langle S_{k+1}, \dots, S_n, \id^{\otimes n} \otimes T_1,\dots,\id^{\otimes n} \otimes T_m \rangle, \label{eq:codespace_generators_ancillas_removed}
	\end{align}
	where $S_{k+1},\dots,S_n$ and $T_1,\dots,T_m$ are all CSS observables such that $\S(\ket{\psi}) = \langle T_1,\dots,T_m \rangle$.
 
 The stabilizer group of $\ket{\bms_\C}$ can then be represented as
	\begin{align}
		\S(\ket{\bms_\C}) = \langle (-1)^{s_1} S_1, &\dots, (-1)^{s_k} S_k, S_{k+1},\dots,S_n, \nonumber \\
		 &\id^{\otimes n} \otimes T_1,\dots,\id^{\otimes n} \otimes T_m  \rangle, \label{eq:logic_states_generators_remove_ancillas}
	\end{align}
	from which we note that $S_1,\dots,S_n$ are all members of $\S(\ket{\bm{0}_\C})$, the stabilizer group for the zero logical basis state in $\C$. Since $\ket{\bm{0}_\C}$ and $\ket{\psi}$ are CSS (see \eqref{eq:ur_logical_basis_states}), $\ket{\psi_{\bm{0}}}$ must be CSS as well (as discarding the last $m$ qubits is CSS-preserving). Therefore, $\S(\ket{\bm{0}_\C})$ is a direct sum of $\mathbb{F}_2$-subspaces $\id^{\otimes n} \otimes \S(\ket{\psi})$ and ${\S(\ket{\psi_{\bm{0}}}) \otimes \id^{\otimes m}}$, where $\S(\ket{\psi_{\bm{0}}})$ is the stabilizer group of $\ket{\psi_{\bm{0}}}$.   Consequently, we can express $S_i$ as $S_i = S'_i \otimes T'_i$, where $S'_i$ and $T'_i$ are CSS observables of the same type as $S_i$ that respectively stabilize $\ket{\psi_{\bm{0}}}$ and $\ket{\psi}$. Therefore, by multiplying each $S_i$ in \eqref{eq:logic_states_generators_remove_ancillas} by appropriate generators of the same type in ${\{\id^{\otimes } \otimes T_1,\dots \id^{\otimes n} \otimes T_m\}}$, we can represent the stabilizer group of $\ket{\bms_\C}$ as 
	\begin{align}
	\S(\ket{\bms_\C}) = \S(\ket{\psi_\bms}) \otimes \S(\ket{\psi}),
	\end{align} 
	where the stabilizer group $\S(\ket{\psi_\bms})$ is generated by
	\begin{align}
		\S(\ket{\psi_\bms}) = \langle (-1)^{s_1} S'_1, \dots, (-1)^{s_k} S'_k, S'_{k+1},\dots,S'_n \rangle,
	\end{align}
	in which $S'_1,\dots,S'_k$ are $Z$-type becausse $S_1,\dots,S_k$ are $Z$-type. By applying \lemref{lemma:CNOT_gaussian_elim} to $S'_1,\dots, S'_n$, we can find a logical basis $\{\ket{\bms_{\C'}}\}$ for an $[[n,k]]$ CSS code $\C'$ stabilized by $\langle S'_{k+1},\dots,S'_n\rangle$, generated by a completely CSS-preserving encoding unitary, such that $\ket{\bms_{\C'}}$ shares the stabilizer group of $\ket{\psi_\bms}$. Therefore, $\ket{\bms_{\C'}}$ and $\ket{\psi_\bms}$ only differ up to a phase, which implies the Lemma. \qedhere
\end{proof}

\begin{lemma}\label{lemma:remove_ancillas}
	Let $\K$ be the code projection for an ${[[n+m,k]]}$ CSS code $\C$ where $n\ge 1, n > k$ and $m > 0$. Then given any $m$-qubit CSS state $\ket{\psi}$, we have that $\C([\cdot] \otimes \ketbra{\psi})$ is equivalent to a CSS code projection on $n$ qubits \emph{alone}, followed by preparing a CSS state and completely CSS-preserving post-processing, i.e.
	\begin{align}
		\K(\rho \otimes \ketbra{\psi}) = p\ \tilde{\U} \circ \left(\tilde{\K}(\rho) \otimes \ketbra{\varphi}\right),
	\end{align} 
	where $p$ is a probability, $\tilde{\U}$ is a completely CSS-preserving unitary channel on $k$ qubits, $\tilde{\K}$ is a code projection for an $[[n,k']]$ CSS code where $0 \le k' \le k$, and $\ket{\varphi}$ is a CSS state on $k-k'$ qubits. 
\end{lemma}
\begin{proof}
	Let $\{S_{n+1},\dots,S_{n+m}\}$ be a set of CSS observables that generate the stabilizer group defining $\ket{\psi}$. Then $\K(\rho \otimes \ketbra{\psi})$ is equivalent to
	\begin{align}
		\K(\rho \otimes \ketbra{\psi}) = \K(\mathbf{P}[\rho \otimes \ketbra{\psi}] \mathbf{P}),
	\end{align}
	where $\mathbf{P}$ projects the last $m$ qubits onto $\ket{\psi}$, i.e.
	\begin{align}
	\mathbf{P} \coloneqq \id^{\otimes n} \otimes \ketbra{\psi} = \prod_{i=n+1}^{n+m} P(\id^{\otimes n} \otimes S_i).
	\end{align}
	
	By applying \lemref{lemma:measure_then_decode} to each projection carried out by $\mathbf{P}$, we obtain
	\begin{align}
		\K(\rho \otimes \ketbra{\psi}) = p\ \U \circ \K'(\rho \otimes \ketbra{\psi}) \otimes \ketbra{\varphi}, \label{eq:do_auxiliary_projections}
	\end{align} 
	where $p$ is a probability, $\U$ is a completely CSS-preserving unitary channel on $k$ qubits, $\K'$ is a code projection for an $[[n+m,k']]$ CSS code $\C'$ where $0 \le k' \le k$, and $\ket{\varphi}$ is a CSS state on $k-k'$ qubits. \lemref{lemma:measure_then_decode} further implies there exists a logical basis $\{\ket{\bms_{\C'}}|\bms \in \{0,1\}^{k'} \}$ for the new code $\C'$, generated by a completely CSS-preserving encoding unitary, which can be related to some state $\ket{\Psi_\bms}$ on $n+m$ qubits as  
	\begin{align}
		\mathbf{P}\ket{\Psi_\bms} = \left(\id^{\otimes n} \otimes \ketbra{\psi}\right) \ket{\Psi_\bms} \propto \ket{\bms_{\C'}}. \label{eq:project_prop}
	\end{align} 
	\eqref{eq:project_prop} immmediately implies
	\begin{align}
	\ket{\bms'_{\C'}} = \ket{\psi_\bms} \otimes \ket{\psi}, \label{eq:s_K'_decomp}
	\end{align} 
	where $\ket{\psi_\bms}$ is a state on the first $n$ qubits \emph{alone}. By \lemref{lemma:codespace_input_ancilla_split}, each $\ket{\psi_\bms}$ is, \emph{up to a phase that may vary with $\bms$}, a logical basis state $\ket{\bms_{\C''}}$ for an $[[n,k']]$ CSS code $\C''$ with only $n$ physical qubits. Thus one can write 
	\begin{align}
	   \ket{\bms_{\C'}} = e^{-i \theta_\bms} \ket{\bms_{\C''}} \otimes \ket{\psi}. \label{eq:logical_basis_state_up_to_phase}
	\end{align}
		
	By \eqref{eq:code_projection_logical_basis_states}, we can express the code projection $\K'$ for $\C'$ in terms of the logical basis $\{\ket{\bms_{\C'}}\}$ as 
	\begin{align}
		\K'(\cdot) = K(\cdot)K^\dagger,\ K^\dagger \coloneqq \sum_{\bms \in \{0,1\}^{k'}} \ket{\bms_{\C'}}\bra{\bms}.\label{eq:C'_initial}
	\end{align}
	We can then use \eqref{eq:logical_basis_state_up_to_phase} to show that 
    \begin{align}
		\K'(\rho \otimes \ketbra{\psi}) = U' \left[\tilde{\K}(\rho)\right] U^{'\dagger}, \label{eq:removing_ancillas} 
	\end{align} 
	where we have defined the following unitary on $k'$ qubits to adjust for the phase differences between $\ket{\bms_{\C'}}$ and $\ket{\psi_\bms}$,
	\begin{align}
		 U' \coloneqq \sum_{\bms \in \{0,1\}^{k'}} e^{i\theta_\bms} \ketbra{\bms},
	\end{align}
	as well as 
	\begin{align}
		\tilde{\K}(\cdot) \coloneqq \tilde{K}(\cdot)\tilde{K}^\dagger,\ \tilde{K}^\dagger \coloneqq \sum_{\bms \in \{0,1\}^{k'}} \ket{\bms_{\C''}}\bra{\bms},
	\end{align} 
	By comparison with \eqref{eq:code_projection_logical_basis_states}, we see that $\tilde{\K}$ is the code projection for the $[[n,k']]$ CSS code $\C''$. Following a similar argument to that at the end of case (iii) in the proof of \lemref{lemma:measure_then_decode}, we can demonstrate that $U'$ is completely CSS-preserving. Substituting back immeidately yields the Lemma result. \qedhere	
\end{proof}

\section{\uppercase{Complex relative Majorization}}
\label{appx:complex_maj}

In this appendix, we sketch how our entropic constraints can be extended to arbitrary input and output states -- i.e. those whose density matrices in the computational basis may not be real.

We refer to any complex-valued distribution $\bmw$ on $N$ elements as a \emph{complex quasidistribution} if its components sum to 1, i.e. 
\begin{align}
	\sum_{i=1}^N \bmw^{(i)} = 1. \label{eq:normalization_quasi},
\end{align}
Given complex quasidistributions $\bmw, \bmw'$ and reference probability distributions $\bmr,\bmr'$, we can extend relative majorization very naturally to a partial order $\succ_C$ defined by
\begin{align}
    (\bmw,\bmr) \succ_C (\bmw',\bmr') \iff A\bmw = \bmw',\ A\bmr = \bmr'.
\end{align}

Any complex quasidistribution $\bmw$ on $N$ elements can be decomposed into real and imaginary parts as
\begin{align}
\bmw = \bmw_R + i\bmw_I,
\end{align} 
where we have introduced the notation $(\cdot)_R \coloneqq \mathfrak{Re}(\cdot)$ and $(\cdot)_I \coloneqq \mathfrak{Im}(\cdot)$. By \eqref{eq:normalization_quasi}, we have
\begin{align}
	\sum_{i=1}^N w^{(i)}_R = 1,\ \sum_{i=1}^N w^{(i)}_I = 0.
\end{align}
Therefore, given any complex quasidistribution $\bmw$ on $N$ elements, we can construct a valid \emph{quasiprobability distribution} on $2N$ elements using the map
\begin{align}
    \mu(\bmw) \coloneqq \bmw_R \oplus \bmw_I.
\end{align}
Furthermore, given any probability distribution $\bmr$ on $N$ elements, we can also construct a probability distribution on $2N$ elements using the map 
\begin{align}
	\nu(\bmr) \coloneqq \frac{1}{2}(\bmr \oplus \bmr).
\end{align}
We can now prove the following theorem:
\begin{theorem}
Given complex quasidistributions $\bmw, \bmw'$ and probability distributions $\bmr,\bmr'$, we have that
\begin{align}
    (\bmw, \bmr) \succ_C (\bmw', \bmr') \implies (\mu(\bmw), \nu(\bmr) ) \succ (\mu(\bmw') , \nu(\bmr')).  
\end{align}  
\end{theorem}
\begin{proof}
    Since a stochastic matrix processes the real and imaginary parts of a vector independently, we conclude that $(\bmw, \bmr) \succ_C (\bmw', \bmr')$ implies
\begin{align}
    \begin{bmatrix}
        \bmw_R' \\
        \bmw_I'
    \end{bmatrix} &=   \begin{bmatrix}
        A & 0\\
        0 & A
    \end{bmatrix}\begin{bmatrix}
        \bmw_R \\
        \bmw_I
    \end{bmatrix},   \begin{bmatrix}
        \frac{1}{2} \bmr' \\
       \frac{1}{2}\bmr'
    \end{bmatrix} &=   \begin{bmatrix}
        A & 0\\
        0 & A
    \end{bmatrix} \begin{bmatrix}
      \frac{1}{2} \bmr \\
       \frac{1}{2} \bmr
    \end{bmatrix}
\end{align}
for some stochastic matrix A. Defining $\tilde{A} \coloneqq A \oplus A$, we can rewrite this more compactly as
\begin{align}
    \mu(\bm{w}') = \tilde{A} \mu(\bmw),\ \nu(\bmr') = \tilde{A} \nu(\bmr),
\end{align}  
where $\tilde{A}$ is stochastic whenever $A$ is stochastic, completing the proof.
\end{proof}

\begin{corollary}
   Let $\w_\rho \coloneqq \mathfrak{Re}[W_\rho] \oplus \mathfrak{Im}[W_\rho]$ and $\r_\tau \coloneqq \frac{1}{2}(W_\tau \oplus W_\tau)$. If there exists a completely CSS-preserving channel $\E$ such that $\E(\rho) = \rho'$ and $\E(\tau) = \tau'$, where $\tau$ and $\tau'$ are both in the interior of $\D_{css}$, then for all $\alpha \in \A$,
    \begin{align}
        D_\alpha (\w_\rho || \r_\tau) \ge  D_\alpha (\w_{\rho'} || \r_{\tau'}).
    \end{align}    
\end{corollary}